\documentclass[a4paper, 11pt]{article}

\usepackage{amsmath,amssymb,graphicx,amstext}
\usepackage{tikz}
\usepackage{IEEEtrantools}
\usepackage{caption}
\usepackage{subcaption}
\usepackage{textcomp}
\usepackage{booktabs}
\usepackage{changes}
\usepackage{rotating}
\usepackage{amsthm}

\usepackage[top=2 cm,bottom=2 cm,left=2 cm, right=2 cm]{geometry}

\newtheorem{theorem}{Theorem}

\newcommand{\footremember}[2]{%
   \footnote{#2}
    \newcounter{#1}
    \setcounter{#1}{\value{footnote}}%
}
\newcommand{\footrecall}[1]{%
    \footnotemark[\value{#1}]%
} 

\usetikzlibrary{matrix, arrows,decorations.pathreplacing}
\usetikzlibrary{decorations.pathmorphing}
\usetikzlibrary{shapes}

\def\BRI1#1#2{
\begin{scope}[shift={#1}, rotate={#2}]
	\filldraw[rounded corners=1pt, fill=black, draw=black] (7,6.5)--(7,5.5)--(6.8,5.35)--(6.9,5.2)--(7.05,5.35)--(7.2,5.2)--(7.3,5.35)--(7.1,5.5)--(7.1,6.5)--cycle;
\end{scope}
}

\def\BRI2#1#2{
\begin{scope}[shift={#1}, rotate={#2}]
	\filldraw[rounded corners=1pt, fill=black, draw=black] (7,6.0)--(7,5.5)--(6.8,5.35)--(6.9,5.2)--(7.05,5.35)--(7.2,5.2)--(7.3,5.35)--(7.1,5.5)--(7.1,6.0)--cycle;
\end{scope}
}

\def\BR#1{
\begin{scope}[shift={#1}]
	\filldraw[fill=red, draw=black] circle (1mm);
\end{scope}
}

\def\BKI1#1{
\begin{scope}[shift={#1}]
	\filldraw[fill=blue, draw=black] (0.1,0.1)--(-0.1,0.1)--(-0.1,-0.1)--(0.1,-0.1)--cycle;
\end{scope}
}

\def\BZR1#1{
\begin{scope}[shift={#1}]
	\filldraw[fill=yellow, draw=black] (0.15,0)--(0,0.15)--(-0.15,0)--(0,-0.15)--cycle;
\end{scope}
}

\def\BZRp#1{
\begin{scope}[shift={#1}]
	\filldraw[fill=yellow, draw=black] (0.15,0)--(0,0.15)--(-0.15,0)--(0,-0.15)--cycle;
	\fill[black] circle (0.5mm);
\end{scope}
}

\title{Mathematical Modelling and Analysis of the Brassinosteroid and Gibberellin Signalling Pathways and their Interactions}

\author{Henry R. Allen\footremember{Dnd}{Department of Mathematics, Fulton Building, University of Dundee, Dundee, United Kingdom, DD1 4HN}, Mariya Ptashnyk\footrecall{Dnd}\ ${\rm ^,}$\footremember{CrA}{Corresponding Author}}

\begin{document}

\maketitle
 
\begin{abstract}
The plant hormones brassinosteroid (BR) and gibberellin (GA) have important roles in a wide range of processes involved in plant growth and development. In this paper we derive and analyse  new mathematical models for the BR signalling pathway and for the crosstalk between the BR and GA signalling pathways. To analyse the effects of spatial heterogeneity of the signalling processes, along with spatially-homogeneous ODE models we derive  coupled PDE-ODE systems modelling   the temporal and spatial dynamics of molecules involved in the signalling pathways. The values of the parameters in the model for the BR signalling pathway  are  determined using  experimental data on the gene expression of BR biosynthetic enzymes. 
The stability of  steady state solutions of our mathematical model, shown for a wide range of parameters,  can be related to the BR homeostasis which is essential for proper function of plant cells.  
Solutions of the mathematical model for the BR signalling pathway can exhibit oscillatory behaviour only for relatively large values of parameters associated with transcription factor  brassinazole-resistant1's (BZR)  phosphorylation state, suggesting that this process may be important in governing the stability of signalling processes.
 Comparison between ODE and PDE-ODE models demonstrates  distinct spatial distribution  in the level of BR in the cell cytoplasm,  however the spatial heterogeneity  has significant effect on the dynamics of the averaged solutions only in the case when we have oscillations in solutions for at least one of the models, i.e.\  for possibly  biologically not relevant parameter values.  Our results for the  crosstalk model suggest that the interaction between transcription factors  BZR and DELLA exerts more influence on the dynamics of the signalling pathways than BZR-mediated biosynthesis of GA, suggesting that the interaction between transcription factors may constitute the principal mechanism of the crosstalk between  the BR and GA signalling pathways. In general, perturbations in  the GA signalling pathway have  larger effects  on the dynamics of components of the BR signalling pathway than perturbations in the BR signalling pathway on the dynamics of  the GA pathway.  The perturbation in the crosstalk mechanism also has a larger effect on the dynamics of the BR pathway than of the GA pathway. 
Large changes in the dynamics of the GA signalling processes can be observed only in the case where there are disturbances in both the BR signalling pathway and the crosstalk mechanism. Those results highlight the robustness  of  the GA signalling pathway.
\end{abstract}

{\small \textbf{Key words.}
 plant modelling;  hormone crosstalk  signalling;  homeostasis in plants;  stability  analysis
 
\qquad  \qquad \qquad and  Hopf bifurcation;  PDE-ODE systems 
 
\textbf{AMS subject classification.}  34Cxx,  35Q92,  65Nxx,  92Bxx,  92Cxx
}


\section{Introduction}
The sessile nature of plant life highlights the importance of efficient regulatory mechanisms allowing  plants to respond to environmental stimuli and to adapt to changing environmental conditions. Plants have developed a set of highly integrated signalling pathways. Plant hormones, e.g.~auxin, gibberellin, cytokinin, brassinosteroids, ethylene, are key signalling molecules and their activities depend on the cellular context and  interactions between them.\\
\\
The family of steroidal plant hormones brassinosteroids (BRs) is responsible for the regulation and control of a wide range of essential processes including responses to stresses \cite{Bajguz_A_2009,Gruszka_D_2013}, photomorphogenesis \cite{Belkhadir_Y_2014,Zhu_JY_2013}, root growth \cite{Mussig_C_2003}, and stomatal development \cite{Kim_TW_2012}. Over the last 47 years, the effects of brassinosteroids on plant cells and plants as a whole, as well as their signalling pathways have been studied in detail \cite{Clouse_S_2015}. In particular the signalling pathway of brassinolide (BL), the most biologically active of the discovered BRs, has been examined in great detail, and is now one of the most understood pathways in plant biolo\-gy \cite{Belkhadir_Y_2014,Clouse_S_2011,Kim_TW_2010,Li_Lei_2005,She_J_2011,Wang_W_2014,Wang_ZY_2001,Yang_CJ_2011,Zhu_JY_2013}. BR signalling functions by controlling the expression of various genes regulating developmental processes, of which 1000s have been identified \cite{Sunetal_Y_2010}, including several genes which regulate the production of proteins that act as enzymes during BR biosynthesis \cite{Tanaka_K_2005}. To ensure controlled growth, homeostasis of BRs in plant tissue is carefully maintained by negative feedback in the BR signalling pathway \cite{Tanaka_K_2005}. Absence of BRs and/or perturbed BR signalling have also been linked to many growth defects, including dwarfism and male sterility \cite{Clouse_S_1996,Clouse_S_1998}.\\
\\
Gibberellins (GAs) are another family of plant hormones involved in many developmental processes  in plants, including seed germination, stem elongation, leaf expansion, trichrome development, pollen maturation and the induction of flowering \cite{Achard_P_2008,Daviere_JM_2013}. There are over 130 categorized gibberellins, a few of which are bioactive molecules, the most common being GA\textsubscript{1}, GA\textsubscript{3} and GA\textsubscript{4} \cite{Yamaguchi_S_2008}.\\
\\  
Along with detailed information on individual plant hormone signal transduction pathways, their target genes, and their effects, it is now also known that interactions between various molecules involved in different signalling pathways have an effect on physiological phenomena in plants \cite{Bai_MY_2012,Belkhadir_Y_2014}. For example both auxin and brassinosteroids play a role in the patterning of vascular shoot bundles \cite{Ibanes_M_2009}, as well as exhibiting some cross-regulation of biosynthesis via the BRX gene \cite{Sankar_M_2011}. The GA signalling pathway exhibits a high level of interaction with the BR signalling pathway regulating  growth \cite{Clouse_S_1998,Ubeda-Tomas_S_2009} and responses to stresses \cite{Achard_P_2008,Bajguz_A_2009} among others. The interactions between different signalling pathways are called crosstalks, and the investigation of their mechanisms is key to better understand plant hormonal responses to external  stimuli   \cite{Albrecht_C_2012,Belkhadir_Y_2014,Gruszka_D_2013,Yang_CJ_2011}. Despite the need for better understanding of the mechanisms of crosstalk between hormone signalling pathways, it is hard to obtain  experimentally  quantitative data on the dynamics of all molecules involved in the signalling pathways and interactions between them. For example, it is very difficult to measure the dynamics of hormones such as BR in real time, in part due to their occurring naturally at extremely low levels. Therefore development of accurate mathematical mo\-dels of hormone activity can help to analyse and better understand interactions between signalling pathways and their impact on plant growth and development. 
Hence, the main aim of this paper is to derive and analyse novel mathematical models for the BR signalling pathway, and  the crosstalk between the BR and GA signalling pathways.\\
\\
Along with many modelling results for the gibberellin and auxin signalling pathways \cite{Gordon_S_2009,Liu_J_2010,Middleton_A_2010,Middleton_A_2012,Muraro_D_2011}, only a few models can be found for BR-related signalling processes. A logic model for the  activation states of the components of the BR and auxin signalling pathways and their interactions was derived in \cite{Sankar_M_2011}, and provided a qualitative description of the dynamics of the BR pathway.  In a simple model for BR-mediated root growth, proposed in \cite{VanEsse_G_2012}, the growth  dynamics is assumed to be dependent on the quantity of receptor-bound BL, which was considered to be constant.  A system of ordinary differential equations was considered to describe the dynamics of  BR-regulated transcription factor BES1 (BRI1-EMS SUPRESSOR 1) and its interactions with  R2R3-MYB transcription factor BRAVO, related to division  of plant stem cells \cite{Frigola_D_2017a, Frigola_D_2014}.  To our knowledge there are no previous results on  mathematical modelling of  the crosstalk between BR and GA signalling pathways. \\
\\
In our mathematical model for the BR signalling pathway, we consider the dynamics  of BR, free and bound receptors, inhibitors, and phosphorylated and dephosphorylated transcription factors, not considered in previous models. The spatially homogeneous  dynamics of the molecules involved in the signalling pathway that we consider are modelled by a system of six ordinary differential equations (ODEs). To analyse the  effect of spatial heterogeneity of the signalling processes on the dynamics of BR, we derive a coupled model composed of partial differential equations (PDEs) for BR, inhibitor, and phosphorylated transcription factor,  ODEs for receptors, defined on the cell membrane, and the ODE for dephosphorlyated transcription factors,  localised in the cell nucleus.  Along with spatial distribution of the concentration of BR in the cell cytoplasm,  we observe similar dynamics  for solutions  of the ODE and  averaged solutions of  the PDE-ODE models when those solutions  converge to a steady state as $t \to \infty$. However spatial heterogeneity has significant effect in the case when at least one of the two models has periodic solutions, which is determined for possibly biologically not relevant parameter values. \\
\\
To model the crosstalk between  BR and GA signalling pathways we first  rigorously derive a reduced model for  the GA signalling pathway from the full GA signalling pathway  model proposed in \cite{Middleton_A_2012}.  Then we couple the model for the BR signalling pathway with the reduced model for GA signalling pathway by considering three different interaction mechanisms between BR and GA pathways. By analysing the effect of different interaction mechanisms on the dynamics of molecules involved in the signalling processes, we determine  that one of these mechanisms has a more significant effect on the dynamics of the signalling pathways than the other mechanisms. Similar to the BR signalling pathway model, we also  consider  the influence of spatial heterogeneity of the signalling processes on the dynamics of solutions of the BR-GA crosstalk model. Using the mathe\-matical models developed here, we analyse how   interactions between the BR and GA signalling pathways  depend on the model parameters and the strength of interaction mechanisms. We observed that, in general, parameter changes in both pathways have a stronger effect on the components of the BR signalling pathway than on the components of the GA signalling pathway. Our results also suggest that the interaction between transcription factors exerts more influence on the dynamics of the signalling pathways  than BR signalling-mediated GA biosynthesis. Further, our results suggest that perturbations in the GA signalling pathway have larger effects on the dynamics of components of the BR signalling pathway  than the  perturbations in the BR signalling pathway on the dynamics of  components of the GA signalling pathway, apart from in the case when  we have disturbances in both the BR signalling pathway and the crosstalk mechanism. 
\\
\\
The structure of this paper is as follows. In Section~\ref{sec:bio} a biological overview of the BR and GA signalling pathways, and their interactions is given. In Section~\ref{sec:BR_derive} we derive the mathematical model for the BR signalling pathway and estimate  the values for the model parameters using experimental data from \cite{Tanaka_K_2005}. In Section~\ref{sec:BR_analysis} we perform qualitative analysis of the model for the BR signalling pathway, examining how  the behaviour of solutions of the mathematical model depends on the values of the model parameters. We also define the set of parameters for which the system of ODEs  has stable stationary solutions and the set of parameters for which it undergoes Hopf bifurcation. In Section~\ref{sec:BR_space} we extend our model for the BR signalling pathway to examine the effects of spatial heterogeneity in the signalling processes. In Section~\ref{sec:Crosstalk} we consider the reduction of the mathematical model for the GA signalling pathway, proposed in \cite{Middleton_A_2012},  derive a new model for the crosstalk between BR and GA signalling pathways, and analyse the influence of different interaction mechanisms and changes in the dynamics of one of the signalling pathways on the qualitative and quantitative behaviours of the coupled system. We also analyse the influence of spatial heterogeneity of the signalling processes on the interactions between the BR and GA signalling pathways. We summarise and  discuss our  results  in Section \ref{sec:Discussion}.

\section{Biological Background}\label{sec:bio}
The mathematical modelling and analysis of the BR signalling pathway and of the interactions between the BR and GA signalling pathways is the main aim of this paper. In this section we present an overview of the BR and GA signalling pathways, and the interactions between them.
\subsection{The BR Signalling Pathway}
The signalling process starts at the cell plasma-membrane with the perception of BR by the receptor BRASSINOSTEROID INSENSITIVE1 (BRI1) \cite{Li_Jianming_1997}. Upon BR binding to BRI1, two main events then take place, association of BRI1 to a co-receptor, BRI1-ASSOCIATED RECEPTOR KINASE1 (BAK1), and dissociation of the inhibitor protein BRI1 KINASE INHIBITOR1 (BKI1).  This triggers a transphosphorylation cascade between BRI1 and BAK1, leading further to phosphorylation of BRASSINOSTEROID-SIGNALLING KINASE1 (BSK1), another membrane-bound kinase. Next, BSK1 phosphorylates the protein phosphatase BRI1-SUPPRESSOR1 (BSU1), which dephosphorylates a protein kinase BRASSINOSTEROID INSENSITIVE2 (BIN2), eventually leading to its degradation \cite{Ryu_H_2010}. In the absence of BR, phosphorylated BIN2 has a role in phosphorylating the two transcription factors, BRASSINAZOLE RESISTANT1 (BZR1) and BRI1-EMS-SUPPRESSOR1 (BES1) \cite{Li_Lei_2005}, also known as BZR2 (for the purposes of this paper it is unnecessary to distinguish between the two, so we refer to them jointly as BZR). When phosphorylated, BZR is less stable and thus more unlikely to activate or repress any of the 1000s of genes associated with BR signalling \cite{Ryu_H_2007}, it is also thought that association of phosphorylated BZR to a 14-3-3 protein inhibits its entry to the nucleus. PROTEIN PHOSPHOTASE 2A (PP2A) is responsible for the de-phosphorylation of BZR, which allows its entry into the nucleus and then its activation of BR responsive genes.\\
\\
BZR functions as a repressor of certain genes associated with the biosynthesis of BR, notably for example CONSTITUTIVE PHOTOMORPHOGENESIS AND DWARFISM~(CPD), DWARF4~(DWF4), ROTUNDIFOLIA3~(ROT3) and BRASSINOSTEROID-6-OXIDASE 1~(BR6ox1) \cite{Tanaka_K_2005}. That is, active, de-phosphorylated BZR inhibits production of BR. So, high levels of BR cause low levels of phosphorylated BZR, leading to inhibition of BR biosynthesis and decreasing levels of BR. Conversely, low levels of BR lead to high levels of phosphorylated BZR and activation of BR biosynthesis, increasing the levels of BR. This completes the negative feedback loop of the BR signalling pathway.

\subsection{The GA Signalling Pathway}
Gibberellin Signalling is achieved by enhancing the degradation of DELLA proteins, which influence the expression of GA-responsive genes \cite{Achard_P_2009}. GA molecules are perceived by the GA receptor, GIBBERELLIN INSENSITIVE DWARF1 (GID1), a nuclear-localised protein \cite{Ueguchi-Tanaka_M_2005}. Analysis of GID1's structure revealed that it has a GA-binding pocket, with a flexible extension adjacent \cite{Shimada_A_2008}. When GA binds to GID1, this extension undergoes conformational change, and co\-vers the GA-binding pocket. When closed, the upper surface of this lid binds to DELLA proteins to form the GA.GID1.DELLA complex. The formation of the GA.GID1.DELLA  complex enhances the degradation of DELLA proteins by mediating proteasome-dependent destabilization of DELLA proteins.\\
\\
The GA signalling pathway exhibits negative feedback due to the influence of DELLAs on the expression of several genes which code components of the signalling pathway \cite{Yamaguchi_S_2008}. First, DELLA activates the GID1-encoding gene, leading to an increase in the translation of the GID1 protein. This means that in the absence of DELLAs, GID1 concentration also decreases which will slow down the proteasome-induced DELLA degradation, and that an abundance of DELLA leads to the production of more GID1 and enhances the DELLA degradation. Next, DELLA activates the transcription of genes encoding the enzymes GA 20-oxidase (GA20ox) and GA 3-oxidase (GA3ox). These enzymes catalyse seve\-ral reaction steps in the GA biosynthesis pathway, meaning an abundance of DELLA increases both GA and GID1, leading to degradation of DELLA. Lastly, DELLA represses its own gene transcription.

\subsection{Crosstalk between the BR and GA Signalling Pathways}
The interaction of BRs and GAs has been receiving much attention, due to their shared nature as critical plant growth regulators, combined with the fact that they share many overlapping functions such as regulation of cell elongation \cite{Catterou_M_2001,Ubeda-Tomas_S_2009} and plant responses to abiotic stress \cite{Ahammed_G_2015,Colebrook_E_2014}. However despite the inte\-rest, the exact mechanisms of these interactions have remained largely unclear, save from the fact that they control expression of several genes \cite{Bouquin_T_2001}. There has been much evidence that the signalling processes of BR and GA converge at the level of BZR and DELLA interaction. Direct crosstalk in this fashion was shown in \cite{Li_QF_2012}, where it was shown that  overexpression of DELLA proteins reduced both the abundance and transcriptional activity of BZR. This was found to be due to the formation of a complex of DELLA and BZR, which removed BZR's transcriptional ability. There is also evidence of BRs regulating the biosynthesis of GAs, the so called ``GA Synthesis'' model of BR-GA crosstalk. Two main proposals have been made for the existence of this type of interaction. The effects of BR mutants on GA synthesis were examined in  \cite{Tong_H_2014}, and it was concluded that BZR enhances GA synthesis by activating synthesis of the GA3ox enzyme. In contrast to this the findings in \cite{Unterholzner_S_2015} describe a much larger role for BZR in regulating GA synthesis. They provide evidence for a model where BZR activates the synthesis of the GA20ox enzyme, in addition to the effects described in \cite{Tong_H_2014}. Thus BZR would influence the biosynthesis of GA in exactly the same manner as  DELLAs for other interactions, however  the significance of this mechanism of crosstalk is not yet established \cite{Ross_J_2016,Tong_H_2016,Unterholzner_S_2016}.

\section{Derivation of a Mathematical Model for the BR Signalling Pathway}\label{sec:BR_derive}
In this section we derive a mathematical model of the BR signalling pathway. In order to build a simple, yet sufficiently accurate and efficient model incorporating BR biosynthesis negative feedback, we first construct a reduced reaction schematic that describes the pathway mechanism. This reduction is achieved via a simplification of two principal parts of the signalling pathway: the complex BR biosynthesis network, and the cytoplasm localized phosphorylation cascade. Hence, we build a model focussing on three key components, hormone (BR), inhibitor (BKI1) and transcription factor (BZR), Fig.~\ref{fig:BR_Reaction}.\\
\\
In the mathematical model we consider the binding of free BR molecules to the BRI1 receptors leading to the dissociation of BKI1. This is modelled as an almost instantaneous reaction, with BR + BRI1.BKI1 interacting and resulting into BR.BRI1 + BKI1. In order to model the effects of the signalling cascade induced by the membrane-bound receptors and subsequent effects on the phosphorylation state of BZR, we assume that the effects of active receptors in triggering the cascade may be approximated by the free BKI1 that is released upon this activation. We further assume that the free BKI1  catalyses  dephosphorylation of BZR-p and  destabilises the BIN2 proteins, thus reducing the phosphorylation of BZR.

\begin{figure}[!ht]\centering
\begin{tikzpicture}
\node (a) at (0,0){BR + BRI1.BKI1};
\node (b) at (4,0){BR.BRI1 + BKI1};
\node (c) at (7,0){BZR};
\node (d) at (9,0){BZR-p};
\node (e) at (-1.1,-1){$\emptyset$};

\draw[transform canvas={yshift=0.4ex}, black,-left to] (a) -- node[above]{$\beta_{b}$} (b);
\draw[transform canvas={yshift=-0.4ex}, black,-left to] (b) -- node[below]{$\beta_{k}$} (a);
\draw[transform canvas={yshift=0.4ex}, black,-left to] (c) -- node[above]{$\rho_{z}$} (d);
\draw[transform canvas={yshift=-0.4ex}, black,-left to] (d) -- node[below]{$\delta_{z}$} (c);
\draw[black,-|] (5,0.3)to[out=45,in=135] (7.8,0.5);
\draw[black,->] (5,-0.3)to[out=315,in=225] (7.8,-0.5);
\draw[black,-|] (c)to[out=150,in=30]node[above]{$\alpha_{b}$} (-0.9,0.2);
\draw[black,->] (-1.1,-0.3) --node[left]{$\mu_{b}$} (e);
\end{tikzpicture}
\caption{Reaction schematic of the reduced BR signalling pathway.}
\label{fig:BR_Reaction}
\end{figure}
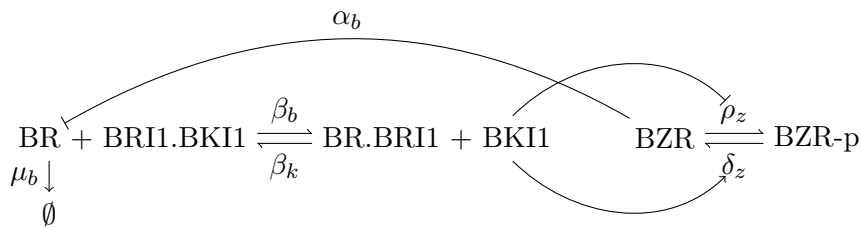

\noindent We denote by $b$ the concentration of hormone BR, by $k$ the concentration of inhibitor BKI1, by $r_{k}$ the concentration of receptor-inhibitor complex BRI1.BKI1, by $r_{b}$ the concentration of receptor-hormone complex BR.BRI1, by $z$ the concentration of (dephosphorylated) transcription factor BZR, and by $z_{p}$ the concentration of (phosphorylated) transcription factor BZR-p. Then assuming spatial homogeneity of the signalling processes,  the interactions between $b$, $k$, $r_{k}$, $r_{b}$, $z$, and $z_{p}$ are described by the system of six ordinary differential equations
\begin{equation}\label{BRode}
\begin{aligned}
\frac{db}{dt} & = \beta_{k}r_{b}k - \beta_{b}r_{k}b + \frac{\alpha_{b}}{1+(\theta_{b}z)^{h_{b}}} - \mu_{b}b,
\\
\frac{dk}{dt} & = \beta_{b}r_{k}b - \beta_{k}r_{b}k,
\\
\frac{dr_{k}}{dt} & = \beta_{k}r_{b}k - \beta_{b}r_{k}b,
\\
\frac{dr_{b}}{dt} & = \beta_{b}r_{k}b - \beta_{k}r_{b}k,
\\
\frac{dz}{dt} & = \delta_{z}z_{p}k - \rho_{z}\frac{z}{1+(\theta_{z}k)^{h_{z}}},
\\
\frac{dz_{p}}{dt} & = -\delta_{z}z_{p}k + \rho_{z}\frac{z}{1+(\theta_{z}k)^{h_{z}}}.
\end{aligned}
\end{equation}
 Here $\beta_{b}$ is the binding rate of $b$ to $r_{k}$, and $\beta_{k}$ is the binding rate of $k$ to $r_{b}$. We model this as only two reactions by assuming that when either BR or BKI1 are bound to BRI1 either dissociates sufficiently fast that the levels of the BR.BRI1.BKI1 remains roughly zero. We assume the reaction to occur in some finite closed volume, so the loss of BR is only described by the degradation coefficient $\mu_{b}$.\\
\\
The phosphorylation state of BZR is modelled as being dependent on the levels of free BKI1. We justify this by noting that upon signalling activation, the receptor phosphorylates BSU1 which governs the phosphorylation of BZR via BIN2. Thus since BKI1 is also released upon BR binding, we may model these effects by assuming that free BKI1 activates (or catalyses) the dephosphorylation of BZR-p at rate $\delta_{z}$. In BKI1's absence BZR is phosphorylated at rate $\rho_{z}$, and when BKI1 is present it inhibits the phosphorylation of BZR such that when $k=1/\theta_{z}$, the rate of phosphorylation is halved.\\
\\
Finally we model BR biosynthesis as being directly inhibited by BZR. BZR represses the expression of several genes encoding enzymes, namely CPD, DWF4, ROT3 and BR6ox1, that are required for the conversion of many of the precursors involved in BR biosynthesis. Hence we use a Hill function with exponent $h_{b}>1$ to model the cumulative inhibitory effect of BZR on these genes. We estimate the parameter $h_{b}$ by considering the detailed BR biosynthesis pathway(s) presented in \cite{Chung_Y_2013}. The BR-mediated enzymes that are involved in the biosynthesis are CPD, DWF4, ROT3, and BR6ox1 which mediate four, five, six, and five steps in the biosynthetic pathway respectively.  Since BR biosynthetic enzymes act multiplicatively at different steps of the reaction network, the expressions modelling their actions can be approximated by the product of these expressions, which would mean that the exponents of these functions would be summed,  thus we consider $h_{b} = 20$.\\
\\
\noindent  The model equations \eqref{BRode} imply that the total concentrations of BKI1, BRI1 and BZR are conserved, thus we consider $k + r_{k} = K_{tot}$, $r_{b} + r_{k} = R_{tot}$, and $z + z_{p} = Z_{tot}$, and derive a reduced model:
\begin{equation}\label{eq:redBR}
\begin{aligned}
\frac{db}{dt} & = \beta_{k}(R_{tot}-K_{tot}+k)k - \beta_{b}(K_{tot}-k)b + \frac{\alpha_{b}}{1+(\theta_{b}z)^{h_{b}}} - \mu_{b}b,
\\
\frac{dk}{dt} & = \beta_{b}(K_{tot}-k)b - \beta_{k}(R_{tot}-K_{tot}+k)k,
\\
\frac{dz}{dt} & = \delta_{z}(Z_{tot}-z)k - \rho_{z}\frac{z}{1+(\theta_{z}k)^{h_{z}}}.
\end{aligned}
\end{equation}
 Various values for $R_{tot}$ were reported in \cite{VanEsse_G_2011}, and we chose the value for WT seedling roots.
We further assume that $K_{tot} = R_{tot}$ in order that the receptor should have the ability to be completely inactive, but not be saturated by BKI1.  We also use physiological values reported in the literature in order to write $\beta_{k}$ and $\mu_{b}$ in terms of other parameters, for full calculations see \ref{app:constraints}.
As such we are left with $8$ parameters for which we have no direct estimate, namely $\beta_{b}$, $\alpha_{b}$, $\theta_{b}$, $\delta_{z}$, $Z_{tot}$, $\rho_{z}$, $\theta_{z}$ and $h_{z}$. These parameters were estimated indirectly by validating the numerical solutions of the mathematical model (\ref{eq:redBR}) against experimental results, using numerical optimisation algorithms.\\
\\
\noindent  By deriving the steady state concentration of BR, denoted $[BR]_{0}$, from the level of endogenous 24-epiBL reported in \cite{Wang_L_2014}, we were able to write the rate of BR degradation $\mu_{b}$ in terms of $\alpha_{b}$, $\theta_{b}$, $Z_{tot}$, $\delta_{z}$, $\rho_{z}$, $\theta_{z}$, $h_{z}$ and $h_{b}$ as follows
\begin{equation}\label{eq:mub}
\mu_{b} = \frac{\alpha_{b}}{[BR]_{0}\left(1 + \left(\theta_{b}\frac{Z_{tot}\delta_{z}[BKI1]_{0}(1+(\theta_{z}[BKI1]_{0})^{h_{z}})}{\rho_{z} + \delta_{z}[BKI1]_{0}(1+(\theta_{z}[BKI1]_{0})^{h_{z}})}\right)^{h_{b}}\right)}. 
\end{equation}
  We can  write $\beta_{k}$ in terms of $\beta_{b}$ and other known parameters in two ways. First, using $[BR]_{0}$ in conjunction with the dissociation constant of BR.BRI1, denoted $K_{d}$, reported in \cite{Wang_ZY_2001} we can estimate the steady state concentration of BKI1, denoted $[BKI1]_{0}$, and  write $\beta_{k}$ in terms of $[BR]_{0}$, $[BKI1]_{0}$, $K_{tot}$, $R_{tot}$, and $\beta_{b}$ as follows 
\begin{equation}\label{eq:betak1}
\beta_{k} = \frac{(K_{tot}-[BKI1]_{0})[BR]_{0}}{(R_{tot}-K_{tot}+[BKI1]_{0})[BKI1]_{0}}\beta_{b}, 
\end{equation}
 by assuming that the dissociation constant for BR.BRI1  depends on the steady state concentrations of BR, BR.BKI1 and BR.BRI1. For the second expression we considered the dissociation of both BR and BKI1 from BR.BRI1.BKI1, using the value for the dissociation constant of BKI1, denoted $K_{m}$, reported in \cite{Wang_J_2014}, as well as $K_{d}$, to directly write $\beta_{k}$ in terms of $\beta_{b}$ as 
\begin{equation}\label{eq:betak2}
\beta_{k} = \frac{K_{d}}{K_{m}}\beta_{b}. 
\end{equation}
 These two constraints  \eqref{eq:betak1} and \eqref{eq:betak2}  on $\beta_k$  correspond to two different mechanisms for the interactions between BR, BRI1 and BKI1. In \eqref{eq:betak1} binding of BR to BRI1.BKI1 causes instantaneous dissociation of BKI1 and formation of BR.BRI1, likewise binding of BKI1 to BR.BRI1 causes instantaneous dissociation of BR and formation of BRI1.BKI1. In \eqref{eq:betak2} binding of BR to BRI1.BKI1 or binding of BKI1 to BR.BRI1 leads to the formation of BR.BRI1.BKI1, which may then dissociate into either BKI1 and BR.BRI1, or BR and BRI1.BKI1.  Model \eqref{eq:redBR} was fitted to experimental data  using both conditions \eqref{eq:betak1} and \eqref{eq:betak2} in order to compare their effects, see Figs.~\ref{fig:BRode} and \ref{fig:BRode2}. 
\\
\\
\noindent The experimental data from \cite{Tanaka_K_2005}, used to determine model parameters, give the relative BR biosynthetic gene expression of CPD, DWF4, ROT3, and BR6ox1, and were measured by RT-PCR analysis, then converted to give  relative values with the initial values equal to one. Values were measured for three independent experiments (three replicates), and the data presented by points and (where available) error bars correspond to the mean and standard error respectively. Gene expression was measured in both Wild-Type~(WT) and bri1-401 mutant (where perception of BR by BRI1 is inhibited) plants, and this was accounted for in our parameter estimation by assuming that the parameter $\beta_{b}$ was greater for the WT than for the mutant. $\beta_{k}$ was allowed to vary freely for the mutant case since constraints \eqref{eq:betak1} and \eqref{eq:betak2} are not definitely valid in this case.  Both of these phenotypes were grown under control conditions as well as two other cases: one where plants were grown in a medium containing $5~\mu$M of Brassinazole (BRZ), a BR-specific biosynthesis inhibitor, and one grown in a medium containing 
$0.1~\mu$M of Brassinolide~(BL) having first been grown in the medium containing $5~\mu$M BRZ for two days. Data comparing the control case with the case of growth in the $5 ~\mu$M BRZ medium were recorded for five days. Data comparing further growth after two days of the BRZ medium case with the case of addition of $0.1~\mu$M BL to the medium were recorded for  further $24$ h, and as such for these cases the initial conditions were taken to be the values of the numerical solution to the model for the $5~\mu$M BRZ medium at time $t = 2$ days. For the control conditions we made no amendments to the model, for the case of plants growing in the BRZ-medium we imposed bounds upon the parameters such that $\alpha_{b}$ should be smaller in this case since addition of BRZ reduces the biosynthesis of BR. In order to examine the case where BL was added to the growth medium, an extra term governing influx of exogenous BL and efflux of endogenous BR was added to the equation describing BR dynamics
\begin{equation*}
\frac{db}{dt} = \beta_{k}(R_{tot}-K_{tot}+k)k - \beta_{b}(K_{tot}-k)b + \frac{\alpha_{b}}{1+(\theta_{b}z)^{h_{b}}} - \mu_{b}b + \phi_{b}(\omega_{b}-b),
\end{equation*}
where $\phi_{b}$ is the relative permeability of the cell membrane to BR and was one of the optimised parameters for the relevant cases, and $\omega_{b}$ is assumed to be $0.1~\mu$M in accordance with the experimental procedure.\\
\\
\noindent Optimisation was achieved by comparing the biosynthetic expression defined by the numerical solutions of model (\ref{eq:redBR}), the term $1/(1+(\theta_{b}z)^{h_{b}})$, with expe\-rimental data presented in \cite{Tanaka_K_2005}.  In order to compare  experimental data with the output from our model we first normalised the simulation data by their initial values such that they took values comparable to the experimental data. We then took the mean of the four gene data sets, weighted by the number of times the respective proteins appear in the biosynthetic pathway. The optimisation was carried out in Python using the \verb|curve_fit| function in the \verb|SciPy| module \cite{Scipy}. \verb|curve_fit| applies nonlinear least squares minimization using the trust region reflective algorithm as default, with a default tolerance of $10^{-8}$.  The model was fitted to the data set for each case sequentially, starting with the WT under control conditions since this data set was the largest. The parameters $\mu_{b}$ and $\beta_{k}$ were replaced by the expressions \eqref{eq:mub} and \eqref{eq:betak1} or \eqref{eq:betak2}, respectively, for the WT data under control conditions since this is the only case where such parameter constraints are definitely valid. The parameters generated from the fitting for WT under control conditions were then used as the initial guesses for all other cases, where $\mu_{b}$ and $\beta_{k}$ were also allowed to be fitted. Parameters that were not expected to vary under the different growth conditions were allowed very small variations to account for error in the first fitting, whereas parameters that were expected to vary had much wider bounds. 
\\
\\
Numerical simulations of model \eqref{eq:redBR} using the optimised parameters, given in Tables~\ref{tab:BRode}  and \ref{tab:BRode2} for $\beta_k$ determined by \eqref{eq:betak1} and \eqref{eq:betak2} respectively, are plotted against the experimental data in Fig.~\ref{fig:BRode} and  \ref{fig:BRode2} and show good agreement with  experimental data, having $R^{2}$ values of 0.89 and 0.92 respectively.  
 
\begin{table}[h!]\centering
\resizebox{\linewidth}{!}{
\begin{tabular}{cccc|cccc}
\toprule
Constant & Value & Units & Source & Constant & Value & Units & Source\\
\midrule
$\alpha_{b}$ & 0.27 & $\mu$M min\textsuperscript{-1} & fit & $\theta_{z}$ & 3.95 & $\mu$M\textsuperscript{-1} & fit\\
$\beta_{b}$ & 8.33 & $\mu$M\textsuperscript{-1} min\textsuperscript{-1} & fit & $h_{b}$ & 20 & & \cite{Chung_Y_2013}\\
$\beta_{k}$ & 2.73 & $\mu$M\textsuperscript{-1} min\textsuperscript{-1} & \eqref{eq:betak1} & $h_{z}$ & 6 & & fit\\
$\rho_{z}$ & $1.33\times 10^{-4}$ & min\textsuperscript{-1} & fit & $K_{tot}$ & $6.2 \times 10^{-2}$ & $\mu$M & fit\\
$\delta_{z}$ & $1.02 \times 10^{-3}$ & $\mu$M\textsuperscript{-1} min\textsuperscript{-1} & fit & $R_{tot}$ & $6.2 \times 10^{-2}$ & $\mu$M & \cite{VanEsse_G_2011}\\
$\mu_{b}$ & 3.58 & min\textsuperscript{-1} & \eqref{eq:mub} & $Z_{tot}$ & 2.65 & $\mu$M & fit\\
$\theta_{b}$ & 1.96 & $\mu$M\textsuperscript{-1} & fit & $\phi_{b}$ & 7.06 & min\textsuperscript{-1} & fit\\
$\omega_{b}$ & 0.1 & $\mu$M & \cite{Tanaka_K_2005}\\
\bottomrule
\end{tabular}
}
\caption{The model parameters associated with the BR signalling pathway model~\eqref{eq:redBR}   for WT  when fitted using expression~\eqref{eq:betak1} for $\beta_{k}$,  together with their sources (Parameter set~1). }
\label{tab:BRode}
\end{table}

\begin{table}[h!]\centering
\resizebox{\linewidth}{!}{
\begin{tabular}{cccc|cccc}
\toprule
Constant & Value & Units & Source & Constant & Value & Units & Source\\
\midrule
$\alpha_{b}$ & 0.27 & $\mu$M min\textsuperscript{-1} & fit & $\theta_{z}$ & 3.95 & $\mu$M\textsuperscript{-1} & fit\\
$\beta_{b}$ & 8.33 & $\mu$M\textsuperscript{-1} min\textsuperscript{-1} & fit & $h_{b}$ & 20 & & \cite{Chung_Y_2013}\\
$\beta_{k}$ & $2.18\times 10^{-2}$ & $\mu$M\textsuperscript{-1} min\textsuperscript{-1} & \eqref{eq:betak2} & $h_{z}$ & 6 & & fit\\
$\rho_{z}$ & $4.26\times 10^{-4}$ & min\textsuperscript{-1} & fit & $K_{tot}$ & $6.2 \times 10^{-2}$ & $\mu$M & fit\\
$\delta_{z}$ & $1.75 \times 10^{-3}$ & $\mu$M\textsuperscript{-1} min\textsuperscript{-1} & fit & $R_{tot}$ & $6.2 \times 10^{-2}$ & $\mu$M & \cite{VanEsse_G_2011}\\
$\mu_{b}$ & 3.68 & min\textsuperscript{-1} & \eqref{eq:mub} & $Z_{tot}$ & 2.68 & $\mu$M & fit\\
$\theta_{b}$ & 2.0 & $\mu$M\textsuperscript{-1} & fit & $\phi_{b}$ & 160.93 & min\textsuperscript{-1} & fit\\
$\omega_{b}$ & 0.1 & $\mu$M & \cite{Tanaka_K_2005}\\
\bottomrule
\end{tabular}
}
\caption{The model parameters associated with the BR signalling pathway model~\eqref{eq:redBR}  for WT when fitted using expression~\eqref{eq:betak2} for $\beta_{k}$, together with their sources (Parameter set~2).}
\label{tab:BRode2}
\end{table}

\begin{figure}[!ht]\centering
\includegraphics[width=0.98\linewidth]{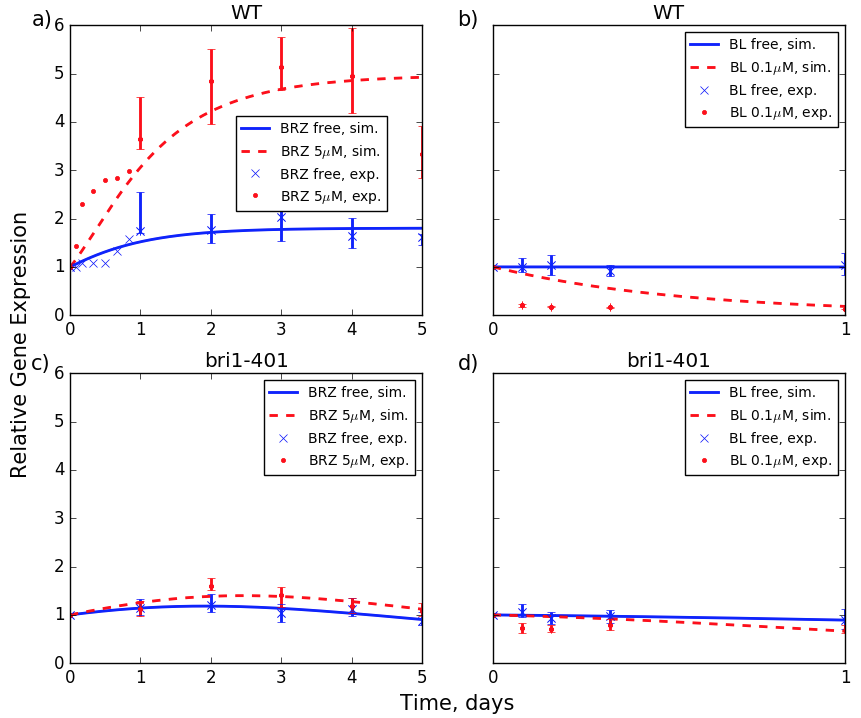}
\caption{BR biosynthetic gene expression calculated from the numerical solutions of model  \eqref{eq:redBR}, plotted against experimental data from \cite{Tanaka_K_2005}. For the WT plants grown under control conditions the parameters given in Table \ref{tab:BRode} were used, parameter sets for other cases can be found in Table \ref{tab:All_Params}.}
\label{fig:BRode}
\end{figure}

\begin{figure}[!ht]\centering
\includegraphics[width=0.98\linewidth]{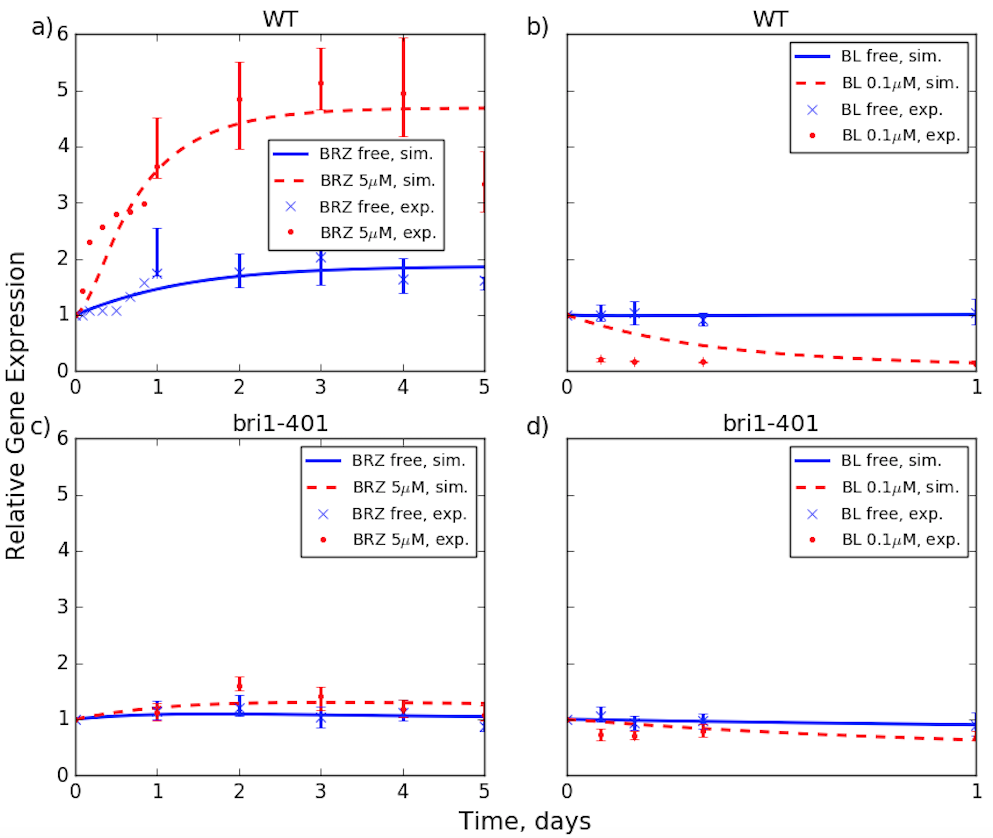}
\caption{ BR biosynthetic gene expression calculated from the numerical solutions of model  \eqref{eq:redBR}, plotted against experimental data from \cite{Tanaka_K_2005}. For the WT plants grown under control conditions the parameters given in Table \ref{tab:BRode2} were used.}
\label{fig:BRode2}
\end{figure}

\section{Qualitative Analysis of the Mathematical Model for the BR Signalling Pathway}\label{sec:BR_analysis}
In this section we consider well-posedness and qualitative analysis of model~\eqref{eq:redBR}.
We start by non-dimensionalising our model, transforming the variables as $t = \frac{1}{\mu_{b}}\bar t$, $b = \frac{\alpha_{b}}{\mu_{b}}\bar{b}$, $k = K_{tot}\bar{k}$, $z = Z_{tot}\bar{z}$, and introducing the dimensionless parameters

\begin{center}
\begin{tabular}{llll}
$\bar{\beta}_{k} = \dfrac{\beta_{k}K_{tot}^2}{\alpha_{b}}$, & $\bar{\beta}_{b} = \dfrac{\beta_{b}K_{tot}}{\mu_{b}}$, & $\kappa = \dfrac{R_{tot}}{K_{tot}}$, & $\bar{\theta}_{b} = \theta_{b}Z_{tot}$,\\
$\epsilon = \dfrac{\alpha_{b}}{K_{tot}\mu_{b}}$, & $\bar{\delta}_{z} = \dfrac{\delta_{z}K_{tot}}{\mu_{b}}$, & $\bar{\rho}_{z} = \dfrac{\rho_{z}}{\mu_{b}}$, & $\bar{\theta}_{z} = \theta_{z}K_{tot}$,
\end{tabular}
\end{center}

\noindent which yields the system (neglecting $\overline{\phantom{bar}}$s)
\begin{equation}\label{eq:BRode}
\begin{aligned}
\frac{db}{d t} & = f_{1}(b,k,z) = \beta_{k}(\kappa -1+k)k - \beta_{b}(1-k)b + \frac{1}{1+(\theta_{b}z)^{h_{b}}} - b,
\\
\frac{dk}{d t} & = f_{2}(b,k,z) = \frac{1}{\epsilon}\left(\beta_{b}(1-k)b - \beta_{k}(\kappa -1+k)k\right),
\\
\frac{dz}{d t} & = f_{3}(b,k,z) = \delta_{z}(1-z)k - \rho_{z}\frac{z}{1+(\theta_{z}k)^{h_{z}}}.
\end{aligned}
\end{equation}

\noindent For simplicity of presentation we denote by $P\subset [1,\infty)\times\mathbb{R}_{+}^{7}\times \mathbb{N}^{2}$ the  parameter space  for  system \eqref{eq:BRode}, where for each $p \in P$, $p = (\kappa, \beta_{k}, \beta_{b}, \theta_{b}, \epsilon, \delta_{z}, \rho_{z}, \theta_{z}, h_{b}, h_{z})$. We assume that $\kappa$ has a minimum value of 1 since $\kappa <1$ would imply saturation of receptor by inhibitor (i.e.~$K_{tot}>R_{tot}$), leading to BR signalling being permanently switched on.

\begin{theorem}\label{thm:ss}
The system (\ref{eq:BRode}) has a unique, global solution $(b,k,z) \in C^{1}([0,\infty))$ for any initial value $(b^{0},k^{0},z^{0}) \in [0,1+\beta_{k}\kappa]\times[0,1]^{2}$, and $(b(t),k(t),z(t)) \in [0,1+\beta_{k}\kappa]\times[0,1]^{2}$ for all $t \in [0,\infty)$ and any $p \in P$.
\end{theorem}
\begin{proof}
Define $u = (b,k,z)^{T}$ and $\textbf{f} = (f_{1},f_{2},f_{3})^{T}$, and hence $\frac{du}{dt} = \textbf{f}$. Since $\textbf{f}$ is locally Lipschitz-continuous, the Picard-Lindel\"{o}f theorem ensures local existence of a unique solution of (\ref{eq:BRode}), see e.g.\ \cite{Amann_H_1990}. To obtain global existence and uniqueness we prove boundedness of solutions by demonstrating the existence of a positive-invariant region for system (\ref{eq:BRode}), i.e.\ showing that for a solution $u$ of (\ref{eq:BRode}) starting in $M = [0,1+\beta_{k}\kappa]\times[0,1]\times[0,1]$ it will always be contained within this region. To show that a region $M$ is positive invariant under the flow of system (\ref{eq:BRode}), we show that  $\textbf{f}(u)\cdot\textbf{n}(u)\geq 0\> \forall u\in \partial M$, see e.g.\   \cite{Amann_H_1990}, where $\textbf{n}$ is the inward normal vector on $\partial M$,  see \ref{Proof_Theo_1} for more details.  This implies uniform boundedness of solutions of \eqref{eq:BRode} with initial values in $M$, and continuous differentiability of $f$ ensures global existence and uniqueness.
\end{proof}

\begin{theorem}\label{thm:steady}
For any parameter set $p\in P$, there exists a unique steady state solution $(b^{*},k^{*},z^{*})\in M$ of  system (\ref{eq:BRode}).
\end{theorem}

\begin{proof}
Considering equations for a steady state solution $(b^{*},k^{*},z^{*})$ of \eqref{eq:BRode} and employing simple algebraic manipulation, $k^{*}$ is defined as a root of the following non-linear function
\begin{equation*}
 \begin{aligned}
  g(k^{*}) : &= \beta_{k}\left(\kappa -1+k^{*}\right)k^{*}\left(1 + \left(\dfrac{\theta_{b}\delta_{z}k^{*}\left(1+(\theta_{z}k^{*})^{h_{z}}\right)}{\rho_{z} + \delta_{z}k^{*}\left(1+(\theta_{z}k^{*})^{h_{z}}\right)}\right)^{h_{b}}\right)
  \\
  &\>\>\>\>\> -\> \beta_{b}\left(1-k^{*}\right),
 \end{aligned}
\end{equation*}
 and $b^\ast$ and $z^\ast$ are determined as functions of $k^\ast$, for more details see \ref{Calc_Theo_2}. We may immediately see, since $g(0) = -\beta_{b}$ and $g(1) = \beta_{k}\kappa(1+(\theta_{b}\delta_{z}(1+\theta_{z}^{h_{z}})/(\rho_{z}+\delta_{z}(1+\theta_{z}^{h_{z}}))^{h_{b}})$, that $g$ must contain at least one root in $[0, 1]$ for any $p\in P$, and hence system (\ref{eq:BRode}) must contain at least one steady state solution in $M$. In order to find the number of  roots of $g(k^\ast)=0$ in $[0,1]$, consider the derivative of $g$ which is positive for all $p\in P$ and $k^{*}\in [0,1]$, see \ref{Calc_Theo_2} for the formula for $g^\prime (k^{*})$. Thus $g$ is monotonically increasing. Strict monotonicity of $g$ coupled with existence of at least one root in $[0,1]$ implies that $g$ has a unique root in $[0,1]$, and hence (\ref{eq:BRode}) has a unique steady state in $M$.
\end{proof}

\subsection{Linearised Stability  and Bifurcation Analysis}\label{subsec:bifurcation}

\noindent  To study the qualitative behaviour of solutions of mathematical model for BR signalling pathway, we performed linearised stability analysis for system \eqref{eq:BRode} and analyse the impact of variations in  values of model parameters  on the behaviour of solutions of \eqref{eq:BRode}. For the parameters obtained via validation of mathematical model by experimental data, see Tables \ref{tab:BRode} and \ref{tab:BRode2},  steady state solutions of  \eqref{eq:BRode} are  linearly stable, with eigenvalues $(-3.8764, -0.1508, -0.0009)$ and $(-3.6863, -0.1138, -0.0006)$ respectively. Further, stability of steady state solutions is maintained under moderate variations of all parameters, suggesting 
that the BR homeostasis is ensured in normally functioning  plant cells.  Large variations in $\delta_{z}$, $\rho_{z}$, and $\theta_{z}$ however cause qualitative changes in the behaviour of  solutions of model \eqref{eq:redBR} and can induce  oscillatory behaviour, but only in the case when all other parameters are as in Table~\ref{tab:BRode} and not as in Table~\ref{tab:BRode2}.    We consider $\delta_{z}$ and $\rho_{z}$ as bifurcation parameters because these parameters directly correspond to processes in the BR signalling pathway, whereas $1/(1+(\theta_{z} k)^{h_z})$ is only an approximation for the dynamics of the cytoplasmic phosphorylation cascade. 
Biologically,  increase of $\delta_{z}$ could potentially correspond to faster phosphorylation of BSU1 by BAK1, or decrease of $\delta_{z}$ corresponding to reduced action of PP2A in dephosphorylating BZR. Further, decrease of $\rho_{z}$ could correspond to BIN2-deficient or insensitive mutants e.g.~bes1-D, and increase of $\rho_{z}$ could correspond to BIN2-overexpressing mutants, e.g.~bin2.  
In the bifurcation  analysis  of system~\eqref{eq:redBR} we considered increased value for the dimensional parameter $\theta_z$ compared to the standard value, Table~\ref{tab:BRode} (i.e.\ $\theta_{z}=41.2\ \mu{\rm M}^{-1}$), which was essential to  determine the region for parameters $\delta_z$ and $\rho_z$ where system~\eqref{eq:BRode} undergoes bifurcation.  For the value of $\theta_{z}=3.95\ \mu{\rm M}^{-1}$  obtained through fitting model solutions to experimental data,  the steady state solution  of \eqref{eq:BRode} is linearly stable for a wide range of values of $\delta_{z}$ and $\rho_{z}$, i.e.~$\delta_{z}, \rho_z \in (0,50)$.

\begin{theorem}
As $\delta_{z}$ and $\rho_{z}$ are continuously varied, system (\ref{eq:BRode}) undergoes a Hopf bifurcation.
\end{theorem}

\begin{proof}
We performed linearised stability analysis to determine the  parameter subspace for which the stationary solution of  \eqref{eq:BRode} is linearly stable, as well as the range of parameters for which we have periodic solutions for the model   \eqref{eq:BRode}.  \\ 
\\
\noindent  Using the Jacobian   \eqref{eq:Jacobian}  of system (\ref{eq:BRode}),  evaluated at the steady-state $(b^\ast, k^\ast, z^\ast)$, we calculate the characteristic equation for the system to be a cubic polynomial of the form $\lambda^{3} + a_{2}\lambda^{2} + a_{1}\lambda + a_{0} = 0$, where $a_{2}$, $a_{1}$ and $a_{0}$ are all positive, real constants (see \ref{Char_eq}). Since the characteristic equation is a cubic polynomial we obtain that there are only 2 possible sets of eigenvalues, either that they are all real or that there is one real eigenvalue $\lambda_{1}$ and two complex conjugate eigenvalues $\lambda_{2}$ and $\lambda_{3}$. Further, since all coefficients have the same sign any real eigenvalue must be negative, specifically zero cannot be an eigenvalue of $J$ for any parameters of system (\ref{eq:BRode}) in $P$. Together these two facts tell us that in the case where the eigenvalues are all real, or that the complex conjugate eigenvalues have negative real part the steady state is stable, and that there is a possible bifurcation point when the two complex conjugate eigenvalues cross the imaginary axis, corresponding to a Hopf bifurcation.\\
\\
We showed numerically that for a closed loop $\mathcal L$ in $(\delta_{z}, \rho_{z})$ the system has complex conjugate eigenvalues with zero real parts,   Fig.~\ref{fig:Crit1},  and in the region $\mathcal{D}$ enclosed by the loop $\mathcal L$ the complex conjugate eigenvalues have positive real part, whereas for $(\delta_{z}, \rho_{z})\in (0,50)^{2}\setminus\overline{\mathcal{D}}$ the real part of the complex conjugate eigenvalues is negative. Hence the points of the loop $\mathcal{L}$ correspond to the bifurcation points where the stability of stationary solutions of \eqref{eq:BRode} changes. We also showed that at such points the eigenvalues have non-zero imaginary parts and hence do not pass through the origin,  Fig.~\ref{fig:Critical_Values}a),  which supports the proof of the fact that zero cannot be an eigenvalue of $J$, presented above. 
For this we designed a scheme in \verb|MATLAB| to calculate the eigenvalues of the Jacobian $J$ in  (\ref{eq:Jacobian}), for values of $(\delta_{z},\rho_{z})\in(0,50)^{2}$,  with  dimensional parameter values  $\theta_z = 41.2~\mu{\rm M}^{-1}$ and all other parameters as in Table~\ref{tab:BRode}. For each $\delta_{z}\in(0,50)$, we determined the values of $\rho_{z}\in(0,50)$ for which $J$ has a pair of non-zero purely imaginary eigenvalues. The derivatives of the real part of the eigenvalues  w.r.t.~both $\delta_{z}$ and $\rho_{z}$  were also calculated numerically,  Fig.~\ref{fig:Critical_Values}b). The values  of $\frac{d}{d\delta_{z}}Re(\lambda_{2,3})$ and $\frac{d}{d\rho_{z}}Re(\lambda_{2,3})$ at the critical points are non-zero, apart from exactly four points on the curve in Fig \ref{fig:Critical_Values}b), where  the derivative w.r.t.~one of the parameters will be zero. Those four points  correspond to  the points where $\delta_{z}$ or $\rho_{z}$ are at their extreme values, i.e.~when $\delta_{z}$ takes an extreme value we have that $\frac{d}{d\rho_{z}}Re(\lambda_{2,3})=0$, and when $\rho_{z}$ takes an extreme value that $\frac{d}{d\delta_{z}}Re(\lambda_{2,3})=0$. When $\delta_{z}$ is fixed at an extreme point, varying $\rho_{z}$ will not cause the point to enter the region bounded by the curve and will not correspond to a bifurcation point w.r.t.~$\rho_{z}$, similar for $\rho_{z}$ fixed at an extreme point. Hence zero derivative at this point does not break the transversality condition.   Thus, in conjunction with Theorems~\ref{thm:ss} and \ref{thm:steady} we have shown that system (\ref{eq:BRode}) satisfies all conditions for the existence of a local Hopf bifurcation \cite{Hassard_B_1981}.

\begin{figure}[!ht]\centering
	\includegraphics[width=.58\linewidth]{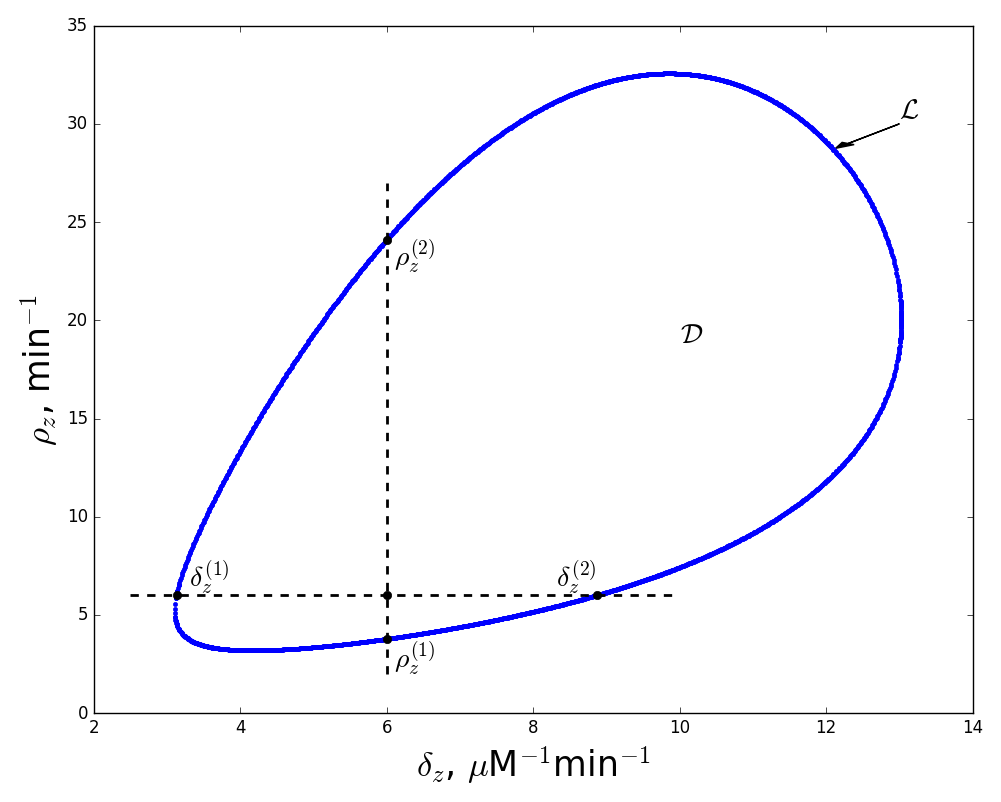}
	\caption{ Critical dimensional values $\delta_{z}$ and $\rho_{z}$, at which the complex conjugate pair of eigenvalues $\lambda_{2}$, $\lambda_{3}$ of (\ref{eq:Jacobian}) are purely imaginary, form a closed curve $\mathcal L$. Note the third eigenvalue $\lambda_{1}$ is always negative. The system (\ref{eq:BRode}) exhibits oscillatory behaviour when values of $(\delta_{z}, \rho_{z})$ are located within the region $\mathcal D$ bounded by curve. For $\delta_{z}=6$ fixed, the eigenvalues cross the imaginary axis at $\rho_{z}^{(1)}=3.78$ and $\rho_{z}^{(2)}=24.1$ with values of $\frac{d}{d\rho_{z}}Re(\lambda_{2,3})$ of $1.78\times 10^{-3}$ and $-4.53\times 10^{-4}$ respectively. For $\rho_{z}=6$ fixed, the eigenvalues cross the imaginary axis at $\delta_{z}^{(1)}=3.14$ and $\delta_{z}^{(2)}=8.87$ with values of $\frac{d}{d\delta_{z}}Re(\lambda_{2,3})$ of $2.73\times 10^{-3}$ and $-1.33\times 10^{-3}$ respectively.}
	\label{fig:Crit1}
\end{figure}

\begin{figure}[!ht]\centering
	\begin{subfigure}{.48\linewidth}
		\includegraphics[width=\linewidth]{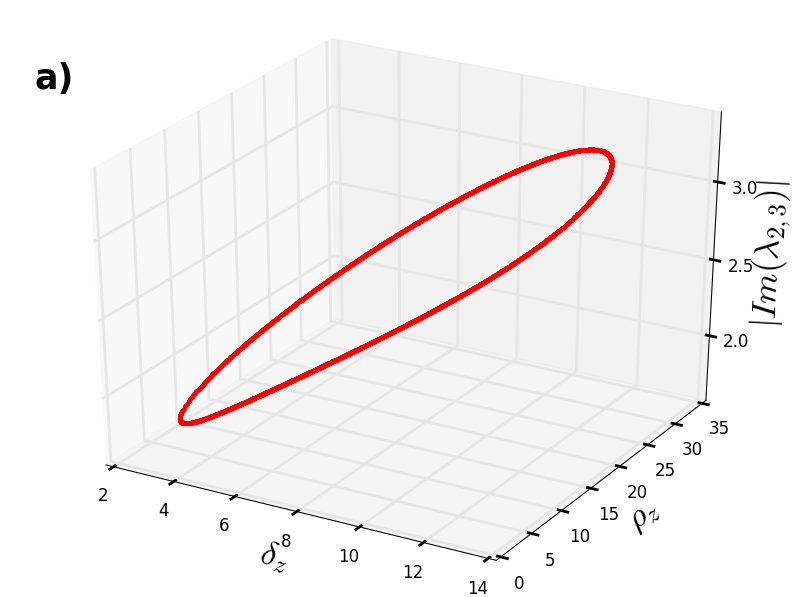}
		\label{fig:Crit2}
	\end{subfigure}
	\begin{subfigure}{.48\linewidth}
		\includegraphics[width=\linewidth]{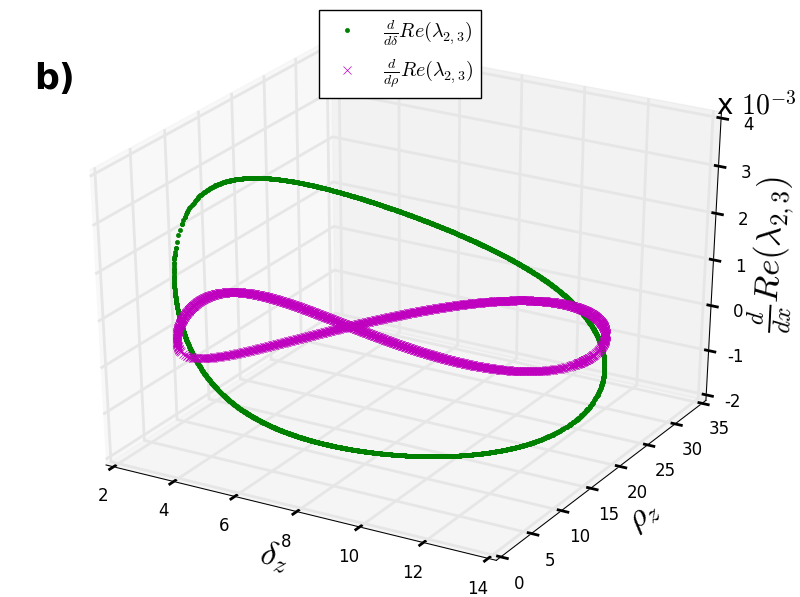}
		\label{fig:Crit3}
	\end{subfigure}
	\caption{Numerical verification of the existence of a Hopf bifurcation.  \textbf{a)} at the critical values, $Im(\lambda_{2,3})$ are non-zero. \textbf{b)} derivatives of $Re(\lambda_{2,3})$ w.r.t.~both $\delta_{z}$ and $\rho_{z}$ (green dots and magenta crosses respectively). Values pass through zero at the points of the curve in a) where $\delta_{z}$ or $\rho_{z}$ take their extrema. At such points there is a bifurcation only in one of the parameters, the parameter for which $Re(\lambda_{2,3})$ has non-zero derivative.}
	\label{fig:Critical_Values}
\end{figure}

\noindent Therefore, for all parameter sets such that $(\delta_{z}, \rho_{z})\in(0,50)^{2}\setminus\overline{\mathcal{D}}$,  $\theta_z=41.2 \mu{\rm M}^{-1}$,  and all other dimensional parameters are defined as in Table~\ref{tab:BRode}, we have that the steady state solution of system \eqref{eq:BRode} is linearly stable. At the points $(\delta_{z},\rho_{z})\in\mathcal{L}$ the system \eqref{eq:BRode} undergoes a Hopf bifurcation, and for $(\delta_{z}, \rho_{z})\in\mathcal{D}$ we have periodic solutions for the BR signalling pathway model. 
\end{proof}

\section{Spatially Heterogeneous Model  for the BR Signalling Pathway}\label{sec:BR_space}
The BR signalling pathway is a process which has distinct functions at different spatial locations in the cell, Fig.~\ref{fig:BR_cell}. Thus it is important to extend the ODE model (\ref{BRode}) and analyse the dependence of the dynamics of the pathway components on the spatial distribution.

\begin{figure}[!ht]\centering
\begin{tikzpicture}[scale=0.65]
	\filldraw[rounded corners=15pt, fill=white, draw=green!30!black, line width=1mm] rectangle (16,6);
	\filldraw[rounded corners=5pt, fill=gray, draw=gray] (4.9,3.9) circle (1);
	\filldraw[fill=brown!75!white, draw=brown!20!black, line width=0.3mm] (4,3) circle (2);
	\BRI2{(10.3,11.5)}{270}
	\BRI2{(10.3,12)}{270}
	\BRI2{(10.3,9)}{270}
	\BRI2{(10.3,8.5)}{270}
	\BRI2{(2, 0.3)}{0}
	\BRI2{(-3, 0.3)}{0}
	\BRI2{(10,5.7)}{180}
	\BRI2{(17.6, 5.7)}{180}
	\BRI2{(5.7,-3)}{90}
	\BRI2{(5.7,-4.3)}{90}
	
	\BKI1{(15.4,4.95)}
	\BKI1{(15.4,4.45)}
	\BR{(15.4,1.95)}
	\BR{(15.4,1.45)}
	\BZR1{(5,3)}
	\BZR1{(4.5,3.5)}
	\BZR1{(4.7,3.1)}
	\BZR1{(4.9,3.3)}
	\BZRp{(3,3)}
	\BZRp{(2.8,3.2)}
	\BZRp{(3.2,3.3)}
	\BZRp{(3.2,2.8)}
	\BR{(14.7,5.1)}
	\BR{(14.6,4.7)}
	\BR{(14.8,4.4)}
	\BR{(14.9,4.8)}
	\BKI1{(15,1.9)}
	\BKI1{(14.7,1.8)}
	\BKI1{(15,1.3)}
	\BKI1{(14.9,1.6)}
	\BKI1{(12.1,1.9)}
	\BKI1{(11.8,1.8)}
	\BKI1{(12.3,1.3)}
	\BKI1{(12,1.5)}
	\BKI1{(3.9,1.5)}
	\BKI1{(4.2,1.6)}
	\BKI1{(4.1,1.2)}
	\BZR1{(7.1,3)}
	\BZR1{(6.9,2.8)}
	\BZR1{(6.8,3.2)}
	\BZRp{(9.1,3.2)}
	\BZRp{(8.9,3)}
	\BZRp{(9.2,2.8)}
	\BR{(6.2,4.7)}
	\BR{(6.3,4.9)}
	\BR{(6.4,4.6)}
	\node (a) at(7.8,5.5){$\emptyset$};
	
	\node (b) at (14.7,5.5){BR};
	\node (c) at (14.8,0.8){BKI1};
	\node (d) at (17,5){BRI1};
	\node (e) at (7.1,2.3){BZR};
	\node (f) at (9.2, 2.3){BZR-p};
	\node (g) at (4,4.5){Nucleus};
	\node (h) at (5.2,5.2){ER};
	
	\draw[transform canvas={yshift=0.4ex}, black,-left to] (3.4,3.2)--(4.5,3.2);
	\draw[transform canvas={yshift=-0.4ex}, black,-left to] (4.5,3.2)--(3.4,3.2);
	\draw[transform canvas={xshift=0.4ex}, black,-left to] (15.2,4.2)--(15.2,2.2);
	\draw[transform canvas={xshift=-0.4ex}, black,-left to] (15.2,2.2)--(15.2,4.2);
	\draw [->,
	line join=round,
	decorate, decoration={
   	 	zigzag,
    		segment length=4,
    		amplitude=1.9,post=lineto,
    		post length=4pt
	}]  (14.5,1.7) -- (12.3,1.7);
	\draw[-|, black] (4,1.7)--(4,3);
	\draw[->, black] (3.8,1.5)to[out=150,in=270] (2.5,3)to[out=90,in=180] (3.5,4)to[out=0,in=90] (4,3.4);
	\draw[-|, black] (11.5,1.7)to[out=180,in=270] (8.05,2.9);
	\draw[->, black] (7.3,3)--(8.7,3);
	\draw[-|, black] (4.7,3.5)--(5.4,4.2);
	\draw[->, black] (5.5,4.5)to[out=45,in=180] (6,4.7);
	\draw [->,
	line join=round,
	decorate, decoration={
   	 	zigzag,
    		segment length=4,
    		amplitude=1.9,post=lineto,
    		post length=4pt
	}]  (6.5,3) -- (5.3,3);
	\draw[->, black] (6.5,3)--(6.7,3);
	\draw[-, black] (6.6,4.8)--(7.9,4.8);
	\draw[->, black] (6.6,4.8)to[out=0, in=270] (a);
	\draw [->,
	line join=round,
	decorate, decoration={
   	 	zigzag,
    		segment length=4,
    		amplitude=1.9,post=lineto,
    		post length=4pt
	}]  (7.9,4.8) -- (14.2,4.8);
	\draw[-, black] (11.5,1.7)--(9,1.7);
	\draw [->,
	line join=round,
	decorate, decoration={
   	 	zigzag,
    		segment length=4,
    		amplitude=1.9,post=lineto,
    		post length=4pt
	}]  (9,1.7) -- (5,1.7);
	\draw [->,
	line join=round,
	decorate, decoration={
   	 	zigzag,
    		segment length=4,
    		amplitude=1.9,post=lineto,
    		post length=4pt
	}]  (6.3,4.5) -- (5.9,3.6);
	\draw [->,
	line join=round,
	decorate, decoration={
   	 	zigzag,
    		segment length=4,
    		amplitude=1.9,post=lineto,
    		post length=4pt
	}]  (5.9,3.6) -- (7.5,3.6);
\end{tikzpicture}
\caption{Diagram of the spatial heterogeneity considered for the model of the BR signalling pathway (\ref{eq:BRpde1})-(\ref{eq:BRpde3}). BR (red circles), BKI1~(blue squares) and BZR-p~(yellow diamonds with black dots) diffuse freely in the cytoplasm, where BR is also degraded. At the membrane, both BR and BKI1 are perceived by BRI1~(black y-shapes) and form complexes with it. In the nucleus~(brown) BKI1 activates dephosphorylation of BZR-p to BZR~(yellow diamonds), and inhibits phosphorylation of BZR to BZR-p. BR is synthesised in the endoplasmic reticulum~(ER, grey crescent) which is continuous with the nuclear membrane.}
\label{fig:BR_cell}
\end{figure}
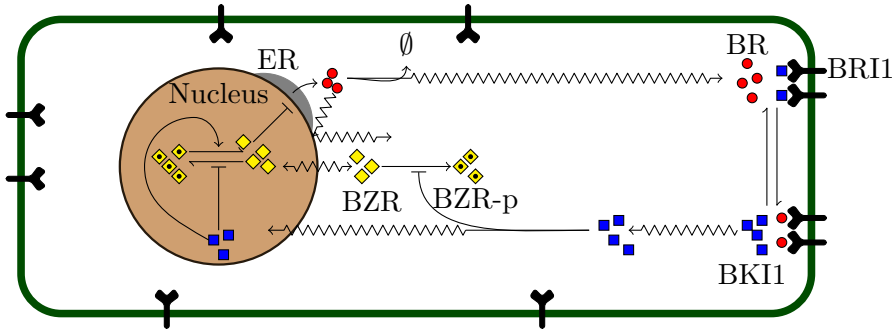

\noindent We consider a spatially heterogeneous model for the BR signalling pathway in the one-dimensional domain $\Omega_{c}=(0,l_{c})$ representing a part of the plant cell cytoplasm, where $l_{c}$ denotes the length of the cell segment we consider. The boundaries of $\Omega_{c}$ are denoted by $\Gamma_n$ modelling the cell nucleus, and $\Gamma_{c}$ representing the cell membrane, Fig.~\ref{fig:domain_Omegac}. \\

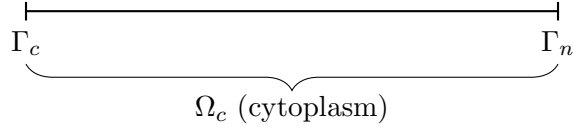
\begin{figure}[!ht]\centering
\begin{tikzpicture}[scale=7]\centering
\draw[-,thick] (0,0) -- (1,0);
\foreach \x/\xtext in {0/$\Gamma_{c}$,1/$\Gamma_{n}$}
	\draw[thick] (\x,0.5pt)-- (\x,-0.5pt) node[below] {\xtext};
\draw [decorate,decoration={brace,amplitude=10pt,mirror},xshift=0pt,yshift=0pt]
(0,-0.1) -- (1,-0.1) node [black,midway,yshift=-0.6cm] {$\Omega_{c}$ (cytoplasm)}; 
\end{tikzpicture}
\caption{A diagram of the one-dimensional domain in which system (\ref{eq:BRpde1})-(\ref{eq:BRpde3}) was solved. $\Omega_{c} = (0, l_{c})$ represents the cytoplasm, $\Gamma_{c}$ the plasma membrane, and $\Gamma_{n}$ the nucleus.} \label{fig:domain_Omegac}
\end{figure}

\noindent We assume the diffusion of BR~($b$), BKI1~($k$) and BZR-p~($z_{p}$) in the cytoplasm: 
\begin{IEEEeqnarray}{rCl}\label{eq:BRpde1}
\left.
 \begin{aligned}
\partial_{t}b  = \; & D_{b} \partial^2_{x}b - \mu_{b}b
\\
\partial_{t}k  = \; & D_{k} \partial^2_{x}k
\\
\partial_{t}z_{p}  = \; & D_{z} \partial^2_{x}z_{p}
\end{aligned}
 \; \right\}
\quad  \text{in }  \Omega_{c}.
\end{IEEEeqnarray}

\noindent The only reaction that takes place in $\Omega_{c}$ is the degradation of BR since we assume that the phosphorylation status of BZR is modulated only in the nucleus. The dynamics occurring on the plasma membrane $\Gamma_{c}$  are  the interactions between  $b$,  $k$, and  receptor  BRI1~($r_{k}$, $r_{b}$). Since  we assume that  the receptors are membrane-bound a system of ODEs, similar to the corresponding equations in system  (\ref{BRode}), is considered to model the dynamics of $r_k$ and $r_b$. The effect of the interactions between $b$, $k$, $r_k$,  and $r_b$  on the dynamics of $b$ and $k$ is defined by Robin boundary conditions for $b$ and $k$ on $\Gamma_c$.  Finally, we assume that the BZR-p cannot diffuse out of the cell, which we model by a zero-flux boundary condition on $\Gamma_c$. Thus on $\Gamma_{c}$ we have
\begin{IEEEeqnarray}{rCl}\label{eq:BRpde2}
\left.
 \begin{aligned}
-D_{b}\partial_{x}b  = \; & \beta_{k}\tilde{r}_{b}k - \beta_{b}\tilde{r}_{k}b
\\
-D_{k}\partial_{x}k  = \; & \beta_{b}\tilde{r}_{k}b - \beta_{k}\tilde{r}_{b}k
\\
-D_{z}\partial_{x}z_{p}  =\; & 0
\\
\frac{d\tilde{r}_{k}}{dt}  = \;  & \beta_{k}\tilde{r}_{b}k - \beta_{b}\tilde{r}_{k}b
\\
\frac{d\tilde{r}_{b}}{dt}  = \; & \beta_{b}\tilde{r}_{k}b - \beta_{k}\tilde{r}_{b}k
\end{aligned}
 \; \right\}
\quad  \text{ on } \Gamma_{c}.
\end{IEEEeqnarray}

\noindent In the nucleus $\Gamma_{n}$ we consider both the phosphorylation and dephosphorylation of BZR. Although the exact subcellular location of BR biosynthesis has not been experimentally demonstrated, the likely location is the endoplasmic reticulum \cite{Shimada_Y_2001}. The endoplasmic reticulum is continuous with the nuclear membrane, hence we model the BR biosynthesis as occurring on $\Gamma_{n}$. We assume that BKI1 cannot enter the nucleus and consider  zero-flux boundary conditions for $k$ on $\Gamma_{n}$.
\begin{IEEEeqnarray}{rCl}\label{eq:BRpde3}
\left.
 \begin{aligned}
D_{b}\partial_{x}b  = & \; \frac{\tilde{\alpha}_{b}}{1 + (\tilde{\theta}_{b}\tilde{z})^{h_{b}}}
\\
D_{k}\partial_{x}k  = &\;  0
\\
D_{z}\partial_{x}z_{p}  = &\;  -\tilde{\delta}_{z}z_{p}k + \rho_{z}\frac{\tilde{z}}{1+(\theta_{z}k)^{h_{z}}}
\\
\frac{d\tilde{z}}{dt}  = &\;  \delta_{z}z_{p}k - \rho_{z}\frac{\tilde{z}}{1+(\theta_{z}k)^{h_{z}}}
\end{aligned}
 \; \right\}
\quad  \text{ on }  \Gamma_{n}.
\end{IEEEeqnarray}

\noindent Since $\tilde{r}_{k}$, $\tilde{r}_{b}$ and $\tilde{z}$ are confined to the boundary, they have units of $mol/m^{2}$. Thus we have the following scaled relationships between variables  $\tilde{r}_{k}$, $\tilde{r}_{b}$ and $\tilde{z}$ in \eqref{eq:BRpde1}-\eqref{eq:BRpde3} and the corresponding variables $r_{k}$, $r_{b}$ and $z$ in \eqref{BRode}: $\tilde{r}_{k}=l_{c}r_{k}$, $\tilde{r}_{b}=l_{c}r_{b}$, $\tilde{z}=l_{c}z$. In order to preserve the balance of units some parameters from model~\eqref{BRode} also had to be rescaled, specifically $\tilde{\alpha}_{b} = l_{c}\alpha_{b}$, $\tilde{\delta}_{z} = l_{c}\delta_{z}$ and $\tilde{\theta}_{b} = \theta_{b}/l_{c}$. We estimated the diffusion constant $D_{b} = 60\>\mu {\rm m}^{2} {\rm min}^{-1}$ by taking the value reported for Progesterone (a steroidal hormone with a similar structure to BL) in physiological solution from \cite{Sieminska_L_1997}. We also took $D_{k}=D_{z}=0.125\>\mu {\rm m}^{2} {\rm min}^{-1}$ for the diffusion constant for proteins from \cite{Sturrock_M_2011}. $l_{c}$ was taken to be 7.43 $\mu$m from measurements of root cell sizes from unpublished data.  All other parameters were taken to be the same as in Table~\ref{tab:BRode} or  Table~\ref{tab:BRode2}.\\
\\
Model \eqref{eq:BRpde1}-\eqref{eq:BRpde3} was solved numerically to analyse the changes of the model solutions due to the spatial heterogeneity of signalling processes. In the fitted parameter regimes, averaged solutions of the model (\ref{eq:BRpde1})-(\ref{eq:BRpde3})  have similar behaviour to solutions of  \eqref{eq:redBR}, Fig.~\ref{fig:BR_Pde1}, however a distinct spatial distribution of BR is characteristic for  (\ref{eq:BRpde1})-(\ref{eq:BRpde3}), see Fig.~\ref{fig:BR_Pde2}.   In the oscillatory parameter regime discussed in section \ref{subsec:bifurcation}, the PDE-ODE model \eqref{eq:BRpde1}-\eqref{eq:BRpde3} was found to have increased amplitude and period of oscillations, see Fig.~\ref{fig:BR_Pde1}. 
For solutions of model (\ref{eq:BRpde1})-(\ref{eq:BRpde3}) we also observe  oscillatory behaviour   for a much wider range of values of $\delta_{z}$ and $\rho_{z}$ than for model \eqref{eq:redBR}.

\begin{figure}[!ht]\centering
	\includegraphics[width=\linewidth]{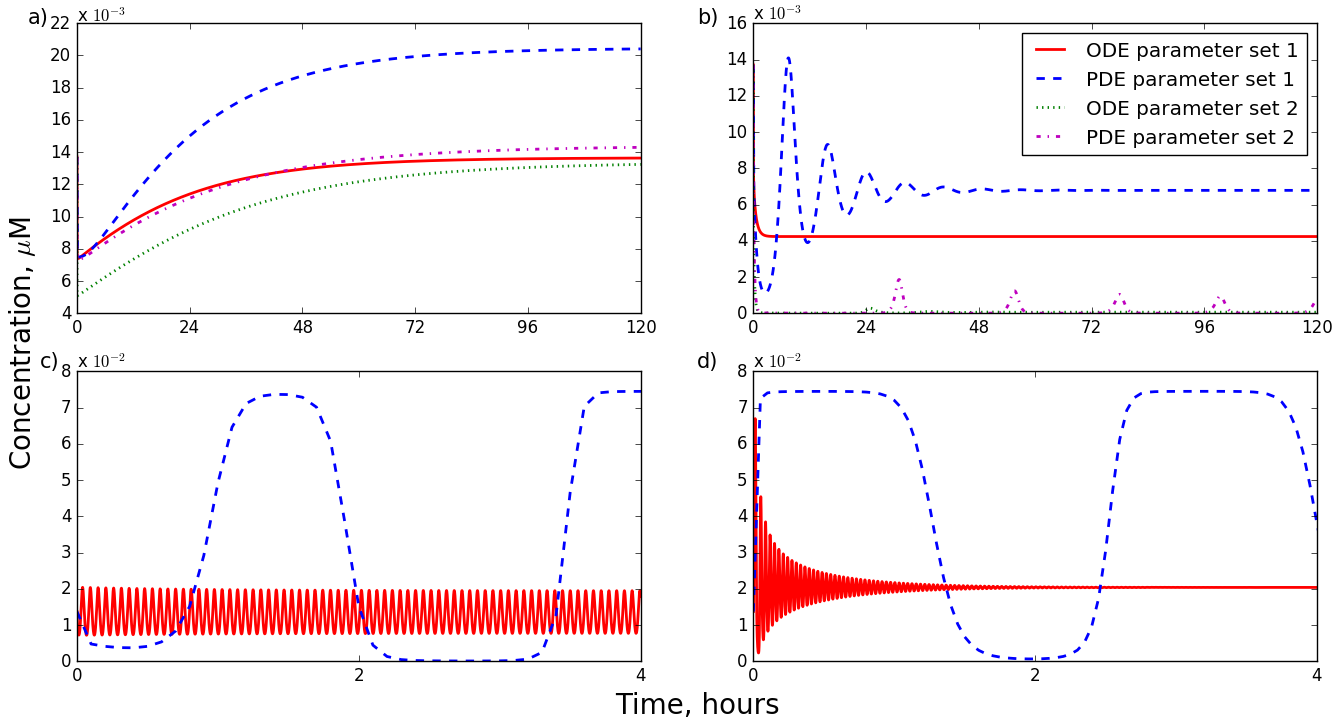}
	\caption{Comparison between the dynamics of BR for the ODE model (\ref{eq:BRode}), and the PDE-ODE model \eqref{eq:BRpde1}-\eqref{eq:BRpde3} for various parameter values. For comparison, the solutions to the PDE-ODE model have been averaged over the space, and the initial conditions are such that they are equal for the ODEs and the averaged PDEs. 
	\textbf{a)} Both models were solved with the fitted parameters  (parameters in Tables~\ref{tab:BRode}  and \ref{tab:BRode2} correspond to set~1 and set~2 respectively), and show similar behaviour. \textbf{b)} Models were solved for $\delta_{z}= 1.02\times 10^{-2}, \rho_{z}=1.33\times 10^{-3}$, $\theta_z=41.2$ and all other parameters as in Tables~\ref{tab:BRode} and \ref{tab:BRode2} respectively. The ODE model tends quickly to steady state, whereas the PDE-ODE model exhibits damped oscillations. 
	 \textbf{c)} Models were solved for $\delta_{z}=\rho_{z}=4$, $\theta_z=41.2$, and all other parameters as in Table~\ref{tab:BRode}. Both systems exhibit periodic solutions, but the PDE-ODE model has a much reduced frequency and much increased amplitude as compared to the ODE model. \textbf{d)} $\delta_{z}=14$, $\rho_{z}=35$, $\theta_z = 41.2$, and all other parameters as in Table~\ref{tab:BRode}, ODE model has moved outside of the oscillatory parameter regime, however the same is not true for the PDE-ODE model, which continues to have oscillatory behaviour for values of $\delta_{z}$, $\rho_{z}$ in excess of 1000.}
	\label{fig:BR_Pde1}
\end{figure}

\begin{figure}[!ht]\centering
	\includegraphics[width=\linewidth]{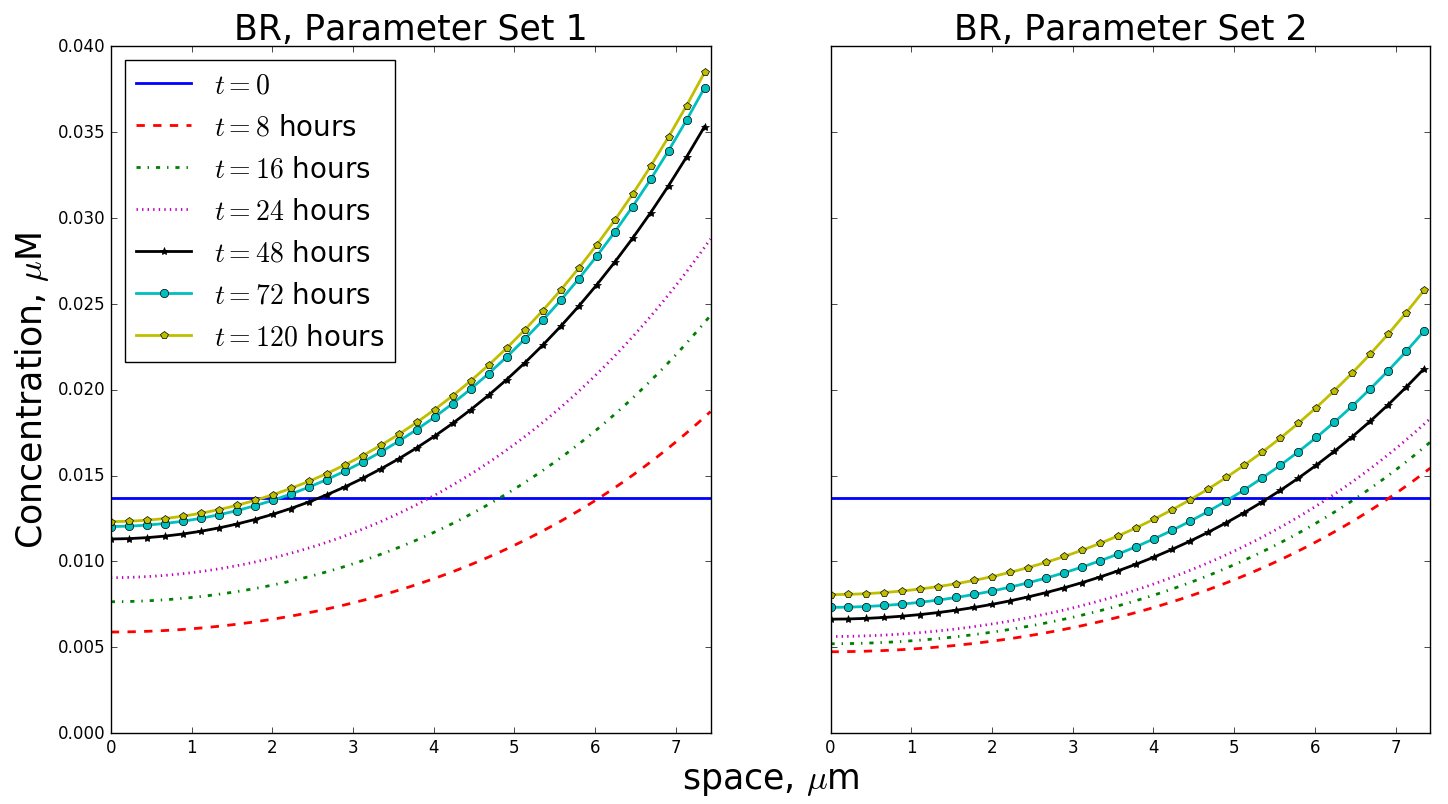}
	\caption{
	The spatial distribution of BR in the cell segment $\Omega_c$, solution of  \eqref{eq:BRpde1}-\eqref{eq:BRpde3}, at different times, where `parameter set 1'  and  `parameter set 2' correspond to parameter values   in Table~\ref{tab:BRode} and   Table~\ref{tab:BRode2} respectively.  The steady state concentration of BR is greater for solutions with parameter set 1, Table~\ref{tab:BRode}, than for solutions with parameter set 2, Table~\ref{tab:BRode2}. }
	\label{fig:BR_Pde2}
\end{figure}

\section{Derivation of a Mathematical Model for Crosstalk between the BR and GA Signalling Pathways}\label{sec:Crosstalk}
Here we consider a mathematical model for the crosstalk between the BR and GA signalling pathways. BRs and GAs can play similar independent roles in development both having important roles in growth \cite{Clouse_S_1998,Ubeda-Tomas_S_2009}, as well as acting together via shared gene expression, and through interactions between their signalling pathways \cite{Bouquin_T_2001,Tanaka_K_2003}.  We aim to use our model to analyse the three different mechanisms of crosstalk between the BR and GA signalling pathways proposed in \cite{Li_QF_2013,Tong_H_2014,Unterholzner_S_2015} and to attempt to establish which mechanism has the most significant influence on the dynamics of the pathways.
\subsection{Mathematical Modelling of GA Signalling}
A detailed model of the GA signalling pathway was derived and examined in \cite{Middleton_A_2012}, consisting of a system of 21 ODEs and 42 parameters. It considers both the interaction of the GA\textsubscript{4}, GID1 and DELLA, and the GA\textsubscript{4} biosynthesis pathway, where GA\textsubscript{12} is converted to GA\textsubscript{15} is converted to GA\textsubscript{24} is converted to GA\textsubscript{9}, all catalysed by enzyme GA20ox, and GA\textsubscript{9} is converted to GA\textsubscript{4}, catalysed by GA3ox.
\\
Since a model for crosstalk must include  variables for both pathways involved, as well as potentially introducing new variables specific to their interactions, a system of ODEs describing such a model may be very large. This being the case, we first reduce the size of the GA signalling model that was derived in \cite{Middleton_A_2012}. We reduce the model by assuming that the dynamics of the molecules involved in the GA biosynthesis pathway are much faster compared to the dynamics of other processes involved in the signalling pathway that their levels remain at a steady state. This enables us to write a system of algebraic equations governing GA\textsubscript{12}, GA\textsubscript{15}, GA\textsubscript{24}, GA\textsubscript{9}, GA20ox and GA3ox, as well as their relevant mRNAs and complexes which may then be solved such that they may therefore be substituted out of the main system. These substitutions generate a new term governing the biosynthesis of GA\textsubscript{4}, where biosynthesis is directly dependent on DELLA concentration, which is then simplified to a fourth order Hill function, signifying the 4 steps in the biosynthesis pathway. We also assume that only one configuration of DELLA.GID1\textsuperscript{c}.GA\textsubscript{4} may be formed. Overall this removes 27 parameters and 13 variables from the system in \cite{Middleton_A_2012}, and introduces a new term for the GA\textsubscript{4} biosynthesis with new parameters $\alpha_{g}$ and $\vartheta_{g}$. The interested reader may examine both the reduction process as well as explicit calculations of $\alpha_{g}$ and $\vartheta_{g}$ in \ref{app:GAred}. Then for $r$, $r_{g}^{o}$, $r_{g}^{c}$, $r_{d}$, $r_{m}$ and $d_{m}$, the concentrations of GID1, GID1\textsuperscript{o}.GA\textsubscript{4}, GID1\textsuperscript{c}.GA\textsubscript{4}, DELLA.GID1\textsuperscript{c}.GA\textsubscript{4}, GID1 mRNA and DELLA mRNA respectively, we obtain
\begin{equation}\label{eq:GAode1}
\begin{aligned}
\frac{dr}{dt} & = -\beta_{g}rg + \gamma_{g}r_{g}^{o} + \alpha_{r}r_{m} - \mu_{r}r,
\\
\frac{dr_{g}^{o}}{dt} & = \beta_{g}rg - \gamma_{g}r_{g}^{o} + \lambda^{o}r_{g}^{c} - \lambda^{c}r_{g}^{o},
\\
\frac{dr_{g}^{c}}{dt} & = -\lambda^{o}r_{g}^{c} + \lambda^{c}r_{g}^{o} - \beta_{d}d_{l}r_{g}^{c} + (\gamma_{d}+\mu_{d})r_{d},
\\
\frac{dr_{d}}{dt} & = \beta_{d}d_{l}r_{g}^{c} - (\gamma_{d}+\mu_{d})r_{d},
\\
\frac{dr_{m}}{dt} & = \phi_{r}\left(\frac{d_{l}}{d_{l}+\vartheta_{r}} - r_{m}\right),
\\
\frac{dd_{m}}{dt} & = \phi_{d}\left(\frac{\vartheta_{d}}{d_{l}+\vartheta_{d}} - d_{m}\right),
\end{aligned}
\end{equation}
 and for $d_{l}$ and $g$, the concentrations of DELLA and GA\textsubscript{4} respectively, we have
\begin{equation}\label{eq:GAode2}
 \begin{aligned}
  \frac{dd_{l}}{dt} & = -\beta_{d}d_{l}r_{g}^{c} + \gamma_{d}r_{d} + \alpha_{d}d_{m},
  \\
  \frac{dg}{dt} & = \phi_{g}(\omega_{g}-g)-\beta_{g}rg + \gamma_{g}r_{g}^{o} + \alpha_{g}\frac{d_{l}^{4}}{d_{l}^{4}+\vartheta_{g}} - \mu_{g}g.
 \end{aligned}
\end{equation}

\begin{table}[h!]\centering
\resizebox{0.98\linewidth}{!}{
\begin{tabular}{ccc|ccc}
\toprule
Constant & Value & Units & Constant & Value & Units\\
\midrule
$\beta_{g}$ & 1.35 & $\mu$M\textsuperscript{-1}min\textsuperscript{-1} & $\gamma_{g}$ & 2.84 & min\textsuperscript{-1}\\
$\alpha_{r}$ & 19.3 & $\mu$M min\textsuperscript{-1} & $\mu_{r}$ & 3.51 & min\textsuperscript{-1}\\
$\lambda^{o}$ & 0.0776 & min\textsuperscript{-1} &  $\lambda^{c}$ & 0.0251 & min\textsuperscript{-1}\\
$\beta_{d}$ & 10 & $\mu$M min\textsuperscript{-1} & $\gamma_{d}$ & 0.133 & min\textsuperscript{-1}\\
$\mu_{d}$ & 6.92 & min\textsuperscript{-1} & $\alpha_{d}$ & $5.28\times 10^{-4}$ & $\mu$M min\textsuperscript{-1}\\
$\phi_{g}$ & $1.061\times 10^{-3}$ & min\textsuperscript{-1} & $\omega_{g}$ & 154.27 & $\mu$M\\
$\alpha_{g}$ & $6.40\times 10^{-3}$ & $\mu$M min\textsuperscript{-1} & 
$\vartheta_{g}$ & $4.27\times 10^{-13}$ & $\mu$M\textsuperscript{4}\\
$\mu_{g}$ & 0.291 & $\mu$M min\textsuperscript{-1} & $\phi_{r}$ & 0.0457 & min\textsuperscript{-1}\\
$\vartheta_{r}$ & $5.6\times 10^{-4}$ & $\mu$M & $\phi_{d}$ & 0.0708 & min\textsuperscript{-1}\\
$\vartheta_{d}$ & 0.01 & $\mu$M & & & \\
\bottomrule
\end{tabular}
}
\caption{Default parameter values used for the GA signalling model. The values of all parameters except $\alpha_{g}$ and $\vartheta_{g}$ were taken directly from \cite{Middleton_A_2012} from which our reduced model has been derived.  The values for  $\alpha_{g}$ and $\vartheta_{g}$ were calculated using the parameters from \cite{Middleton_A_2012} and  the expressions for $\alpha_{g}$ and $\vartheta_{g}$ arising from the model reduction, see \ref{app:GAred}.}
\label{tab:GAode}
\end{table}

\noindent All parameter values are taken from  those for the full model in \cite{Middleton_A_2012}, Table \ref{tab:GAode}. For all components present in both the full and the reduced model, the initial conditions are taken to be the same as in  \cite{Middleton_A_2012}. Both the full and reduced models were solved numerically, and these solutions are plotted in Fig.~\ref{fig:GAode}. The two models show excellent agreement being almost indistinguishable for six terms, and only minor short-term discrepancies for DELLA.GID1\textsuperscript{c}.GA\textsubscript{4} and DELLA. Having thus successfully reduced the GA signalling pathway model while retaining its core behaviour, we derive a model of the crosstalk between the BR and GA signalling pathways by coupling models \eqref{eq:redBR} and  \eqref{eq:GAode1}, \eqref{eq:GAode2}.

\begin{figure}[!ht]\centering
\includegraphics[width=\linewidth]{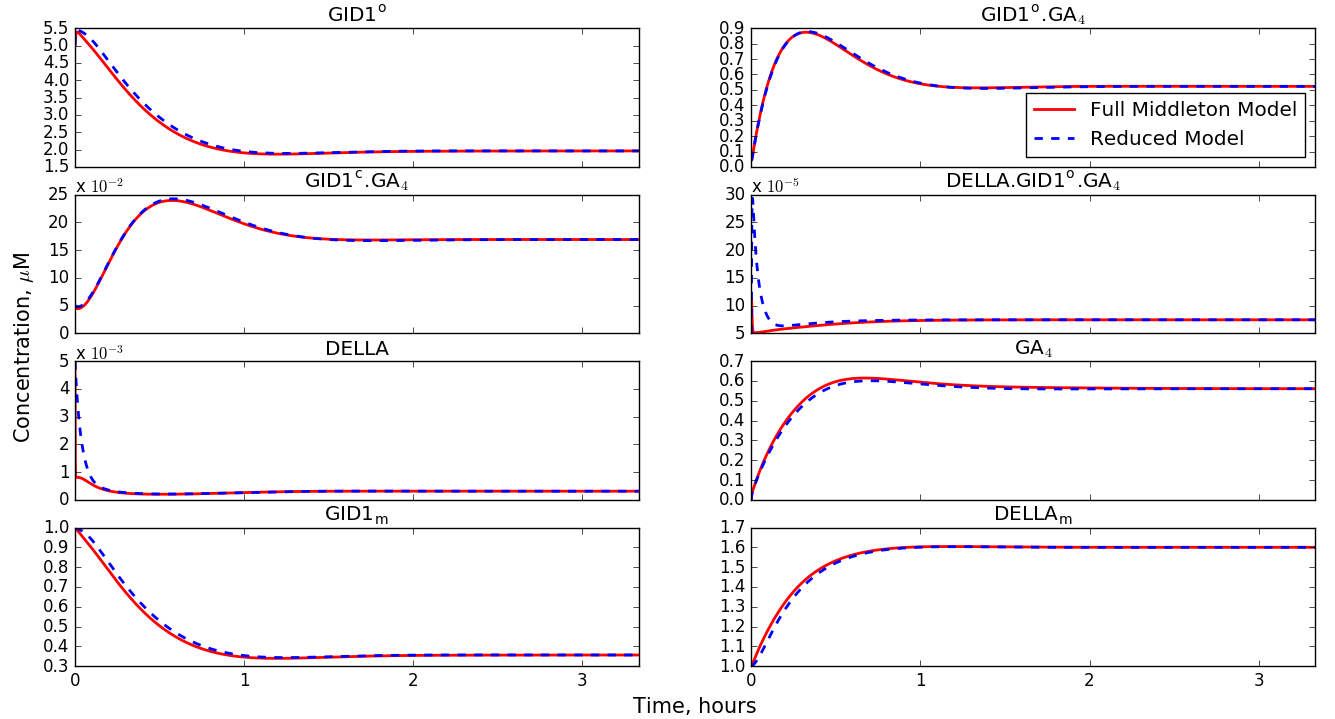}
\caption{A comparison between the numerical simulations of the full model for the GA signalling pathway derived in \cite{Middleton_A_2012}, and the reduced model (\ref{eq:GAode1}), \eqref{eq:GAode2} shows very good agreement.}
\label{fig:GAode}
\end{figure}

\subsection{The Crosstalk Model}\label{subsec:Crosstalk_ODE}
We consider direct crosstalk between the BR and GA signalling pathways.  Besides shared gene expression, two main mechanisms for BR-GA crosstalk have been suggested: direct interaction at the level of transcription factors BZR and DELLA \cite{Li_QF_2013}, and BZR-mediated biosynthesis of GA  \cite{Unterholzner_S_2015}. Here we derive a mathematical model to compare the behaviour of the system under the three different mechanisms  of crosstalk  (interaction between BZR and DELLA, BZR-mediated biosynthesis of GA, and a combination of both of them)  and  to analyse their effects on the BR and GA signalling processes.
\\
\\
The existence of an interaction between BZR and DELLA constituting a crosstalk between the BR and GA signalling pathways was established in \cite{Bai_MY_2012,Gallego-Bartolome_J_2012,Li_QF_2012}. Further to this, in \cite{Li_QF_2013} it was found that the formation of a complex BZR.DELLA inhibits the transcriptional activities of both BZR and DELLA.
\begin{figure}[!ht]\centering
\begin{center}
\begin{tikzpicture}
\node (a) at (0,0){BZR + DELLA};
\node (b) at (4,0){BZR.DELLA};
\draw[transform canvas={yshift=0.4ex}, black,-left to] (a) -- node[above]{$\beta_{z}$} (b);
\draw[transform canvas={yshift=-0.4ex}, black,-left to] (b) -- node[below]{$\gamma_{z}$} (a);
\end{tikzpicture}
\end{center}
\caption{BR-GA crosstalk at the level of transcription factors, where BZR and DELLA form a complex.}
\label{fig:Crosstalk_Complex}
\end{figure}
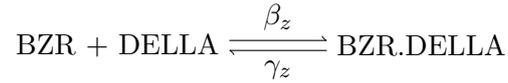
\noindent  A second mechanism of BR-GA crosstalk has been proposed where BZR also has a direct influence on the GA\textsubscript{4} biosynthesis pathway,  Fig.~\ref{fig:Crosstalk_Diagram}. In \cite{Tong_H_2014} it was discovered that BR induces the expression of GA3ox-2, leading to the accumulation of GA\textsubscript{1}, the most bioactive GA for rice. It was also discovered that the external application of BR induces GA20ox expression in Arabidopsis \cite{Stewart-Lilley_J_2013}. Later it was shown that BRs both regulate GA20ox expression, and are required for the transcription of GA3ox \cite{Unterholzner_S_2015}.   Despite good evidence for the coexistence of these two mechanisms, it is still unclear to what extent each mechanism operates to effect changes in BR and GA signalling processes \cite{Ross_J_2016,Tong_H_2016,Unterholzner_S_2016}. 

\begin{figure}[!ht]\centering
\begin{tikzpicture}
\node (a) at (-0.5,3){BZR};
\node (b) at (4.5,3){DELLA};
\node (c) at (2,5){GA\textsubscript{12}};
\node (d) at (2,4){GA\textsubscript{15}};
\node (e) at (2,3){GA\textsubscript{24}};
\node (f) at (2,2){GA\textsubscript{9}};
\node (g) at (2,1){GA\textsubscript{4}};
\node (h) at (-0.5,4){BZR};
\node (i) at (-0.5,2){BZR};

\draw[black,->] (c)--(d);
\draw[black,->] (d)--(e);
\draw[black,->] (e)--(f);
\draw[black,->] (f)--(g);
\draw[black,-] (b)--(3.5,3);
\draw[black,->] (3.5,3)--(2.2,4.5);
\draw[black,->] (3.5,3)--(2.2,3.5);
\draw[black,->] (3.5,3)--(2.2,2.5);
\draw[black,->] (3.5,3)--(2.2,1.5);
\draw[blue,->] (i)--(1.8,1.5);
\draw[red,-] (a)--(0.5,3);
\draw[red,->] (0.5,3)--(1.8,4.5);
\draw[red,->] (0.5,3)--(1.8,3.5);
\draw[red,->] (0.5,3)--(1.8,2.5);
\draw[red,->] (0.5,3)--(1.8,1.5);
\draw[green,-] (h)--(0.5,4);
\draw[green,->] (0.5,4)--(1.8,4.5);
\draw[green,->] (0.5,4)--(1.8,3.5);
\draw[green,->] (0.5,4)--(1.8,2.5);
\end{tikzpicture}
\caption{Proposed mechanisms for the influence of BZR on GA biosynthesis. The first model (green) comes from \cite{Stewart-Lilley_J_2013}, where BZR influences the expression of GA20ox, in the second model (red) from \cite{Unterholzner_S_2015} BZR affects both GA20ox and GA3ox expression, and in the third model (blue) from \cite{Tong_H_2014}  BZR affects only GA3ox expression.}
\label{fig:Crosstalk_Diagram}
\end{figure}
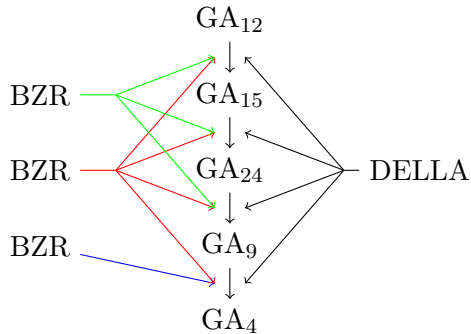

\noindent   To derive the model for BR-GA crosstalk we coupled model \eqref{eq:redBR} and model \eqref{eq:GAode1}, \eqref{eq:GAode2} by adding new interaction terms describing  the crosstalk mechanisms. Two terms were added, a term corresponding to the interaction of BZR and DELLA, which required the introduction of a new variable denoting the concentration of the complex BZR.DELLA, and a term corresponding to the BZR-mediated biosynthesis of GA. By varying the parameters of these new terms, we were then able to analyse the influence that each mechanism exert over the corresponding signalling processes.
\\
\\
BZR.DELLA complex formation is modelled as a reversible reaction between BZR and DELLA. The concentration of BZR.DELLA is denoted by $z_{d}$, and the parameters $\beta_{z}$ and $\gamma_z$ denote the binding rate of BZR and DELLA and  the dissociation rate of BZR.DELLA, respectively,
\begin{equation}\label{eq:CrosstalkODE3}
\frac{dz_{d}}{dt} = \beta_{z}zd_{l} - \gamma_{z}z_{d}.
\end{equation}
 From this we modify system (\ref{eq:redBR}) to take account of the dynamics of BZR.DELLA
\begin{equation}\label{eq:CrosstalkODE1}
\begin{aligned}
\frac{db}{dt} & = \beta_{k}(R_{tot}-K_{tot}+k)k - \beta_{b}(K_{tot}-k)b + \frac{\alpha_{b}}{1+(\theta_{b}z)^{h_{b}}} - \mu_{b}b,
\\
\frac{dk}{dt} & = \beta_{b}(K_{tot}-k)b - \beta_{k}(R_{tot}-K_{tot}+k)k,
\\
\frac{dz}{dt} & = \delta_{z}(Z_{tot}-z-z_{d})k - \rho_{z}\frac{z}{1+(\theta_{z}k)^{h_{z}}} - \beta_{z}zd_{l} + \gamma_{z}z_{d}.
\end{aligned}
\end{equation}
  System (\ref{eq:GAode2}) is modified by taking account of the dynamics of BZR.DELLA, and extending the Hill function describing GA biosynthesis to include BZR-mediated gene expression. For the crosstalk model we do not assume that exogenous GA is entering the system, hence the term $\phi_g(\omega_g - g)$ describing this process  in \eqref{eq:GAode2} is removed. BZR is assumed to control gene expression in the same manner as DELLA, but with relative activity described by $\phi_{z}$, where $\phi_{z}$ is the ratio between the binding thresholds of DELLA and BZR
\begin{equation}\label{eq:CrosstalkODE2}
\begin{aligned}
\frac{dd_{l}}{dt} & \> = \>   -\beta_{d}d_{l}r_{g}^{c} + \gamma_{d}r_{d} + \alpha_{d}d_{m} - \beta_{z}zd_{l} + \gamma_{z}z_{d},
\\
\frac{dg}{dt}   & \>  =  \>   \alpha_{g}\frac{(d_{l}+\phi_{z}z)^{4}}{\vartheta_{g}+(d_{l}+\phi_{z}z)^{4}} - \beta_{g}rg + \gamma_{g}r_{g}^{o} - \mu_{g}g.
\end{aligned}
\end{equation}
 The values for the parameters $\beta_{z}$, $\gamma_{z}$ governing the direct interaction between DELLA and BZR are estimates  from values for similar complex formation and separation from \cite{Middleton_A_2012}. We first assumed that the binding thresholds of BZR and DELLA were equal i.e.~$\phi_{z}=1$, and then analysed the influence of variation in $\phi_{z}$ on the dynamics of the solutions of  \eqref{eq:GAode1}, (\ref{eq:CrosstalkODE3})-(\ref{eq:CrosstalkODE2}).\\
\\
In order to analyse the relative effects of the different crosstalk mechanisms, we examined the behaviour of the model for different values of $\beta_{z}$, $\gamma_{z}$ and $\phi_{z}$. The model was solved numerically for four different conditions: no crosstalk ($\beta_{z},\>\gamma_{z},\>\phi_{z}=0$), crosstalk via BZR activated GA biosynthesis only ($\beta_{z},\>\gamma_{z}=0$), crosstalk via BZR.DELLA complex formation only ($\phi_{z}=0$), and crosstalk via both mechanisms, Fig.~\ref{fig:Crosstalk_Comparison}. 
Only for variables including GA were there any differences when including BZR-mediated biosynthesis of GA, and for all other variables there was close agreement between the cases of no crosstalk and crosstalk via biosynthesis only, and close agreement between the cases of crosstalk via complex formation only and crosstalk via both mechanisms. To examine the influences of both mechanisms on the BR and GA signalling pathways, the relevant parameters were varied. Variation in $\beta_{z}$ and $\gamma_{z}$ led to long-term behaviour changes for the components of the BR signalling pathway,  and short-term changes for the components of the GA signalling pathway, Fig.~\ref{fig:Crosstalk_Complex_mod}. Variations in $\phi_{z}$ had very little effect with differences in behaviour  in variables containing GA for  the cases $\phi_{z}=0$ and $\phi_{z}>0$ (results omitted). From this we concluded that   direct interaction between BZR and DELLA is the predominant form of crosstalk between the BR and GA signalling pathways.

\begin{figure}[!ht]\centering
\includegraphics[width=\linewidth]{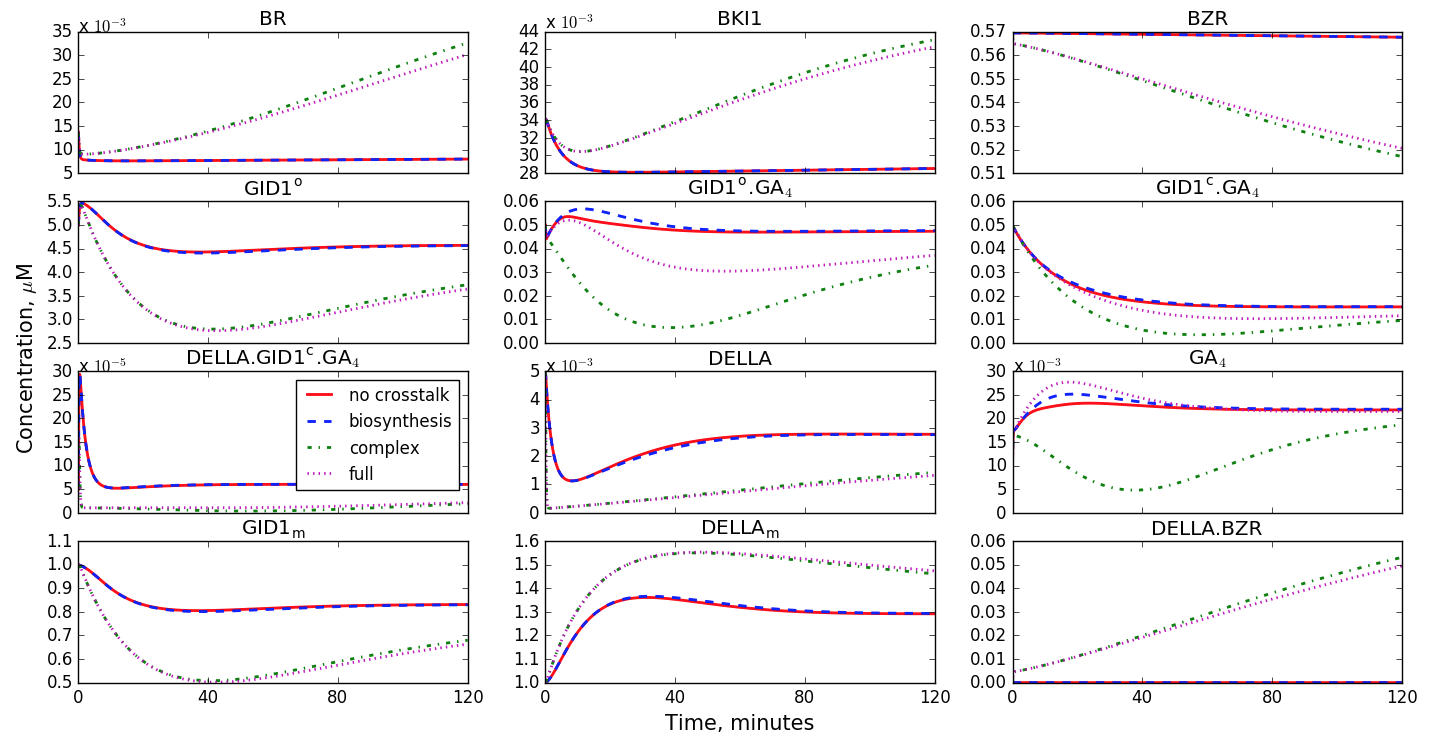}
\caption{Comparisons between solutions of the crosstalk model \eqref{eq:GAode1}, \eqref{eq:CrosstalkODE3}-\eqref{eq:CrosstalkODE2};
 `\textit{no crosstalk}' refers to the case where $\beta_{z}=\gamma_{z}=\phi_{z} = 0$; `\textit{biosynthesis}' refers to the case where the only crosstalk between the BR and GA signalling pathways is BZR-induced biosynthesis of GA ($\beta_{z}=\gamma_{z} = 0$, $\phi_{z}\neq 0$); `\textit{complex}' refers to the case where the only crosstalk between the BR and GA signalling pathways is complex formation of BZR and DELLA ($\beta_{z}, \gamma_{z}\neq 0$, $\phi_{z} = 0$); `\textit{full}' refers to the case where crosstalk between the BR and GA signalling pathways can be via both mechanisms  BZR-mediated GA biosynthesis and complex formation of BZR and DELLA ($\beta_{z}$, $\gamma_{z}$, $\phi_{z}\neq 0$).}
\label{fig:Crosstalk_Comparison}
\end{figure}

\begin{figure}[!ht]\centering
\includegraphics[width=\linewidth]{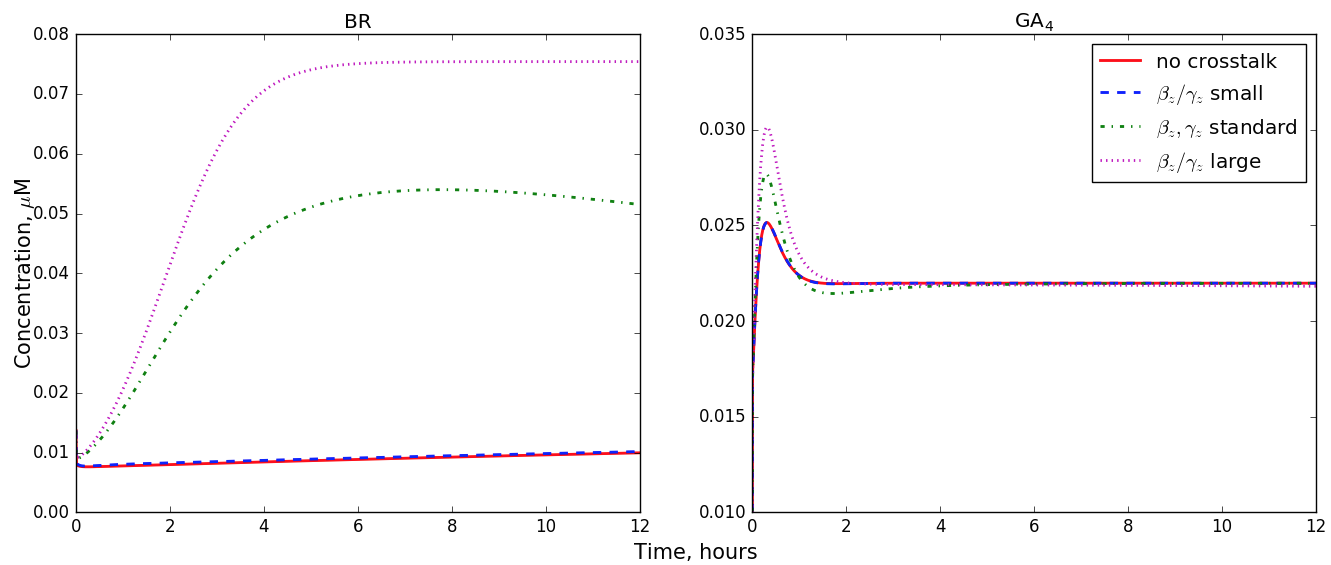}
\caption{Variation of the values of $\beta_{z}$ and $\gamma_{z}$ suggests that the BR signalling pathway is more sensitive to the effects of crosstalk than the GA signalling pathway; `\textit{no crosstalk}' corresponds to the case $\beta_{z}=\gamma_{z} = 0$, `\textit{small}' the case where $\beta_{z}/\gamma_{z}$ has order of magnitude zero, i.e.\ $\beta_{z}/\gamma_{z}\approx 0.75$ , for `\textit{standard}' $\beta_{z}/\gamma_{z}$ has order of magnitude two, i.e.\ $\beta_{z}/\gamma_{z}\approx 75$, and for `\textit{large}' $\beta_{z}/\gamma_{z}$ has order of magnitude four, i.e.\ $\beta_{z}/\gamma_{z}\approx 7500$.}
\label{fig:Crosstalk_Complex_mod}
\end{figure}

\noindent In order to examine the effects of mutations in the BR signalling pathway upon BR-GA crosstalk,  changes in parameters $\delta_{z}$, $\rho_{z}$ and $\theta_{z}$ are introduced into model  \eqref{eq:GAode1},  (\ref{eq:CrosstalkODE3})-(\ref{eq:CrosstalkODE2}),  Fig.~\ref{fig:Crosstalk-Osc}.  Damped oscillations in the dynamics of the components of the BR signalling pathway induced  no such behaviour  in the dynamics of the GA signalling pathway,  and despite causing some changes in the early behaviour of the dynamics of the GA signalling pathway their long term behaviour remains similar.   Alongside the induced perturbation in the BR signalling pathway, we also varied $\beta_{z}$ and $\gamma_{z}$  to examine how crosstalk affects the dynamics of both pathways  in this case. We  found in the case that $\beta_{z}/\gamma_{z}$ is large,  the dynamics of the BR signalling pathway are virtually unchanged from the cases with smaller values of $\beta_{z}/\gamma_{z}$, however the long-term behaviour of the GA signalling pathway was significantly effected, with all components except GA\textsubscript{4} and DELLA\textsubscript{m} having substantially different stationary solutions.

\begin{figure}[!ht]\centering
\includegraphics[width=\linewidth]{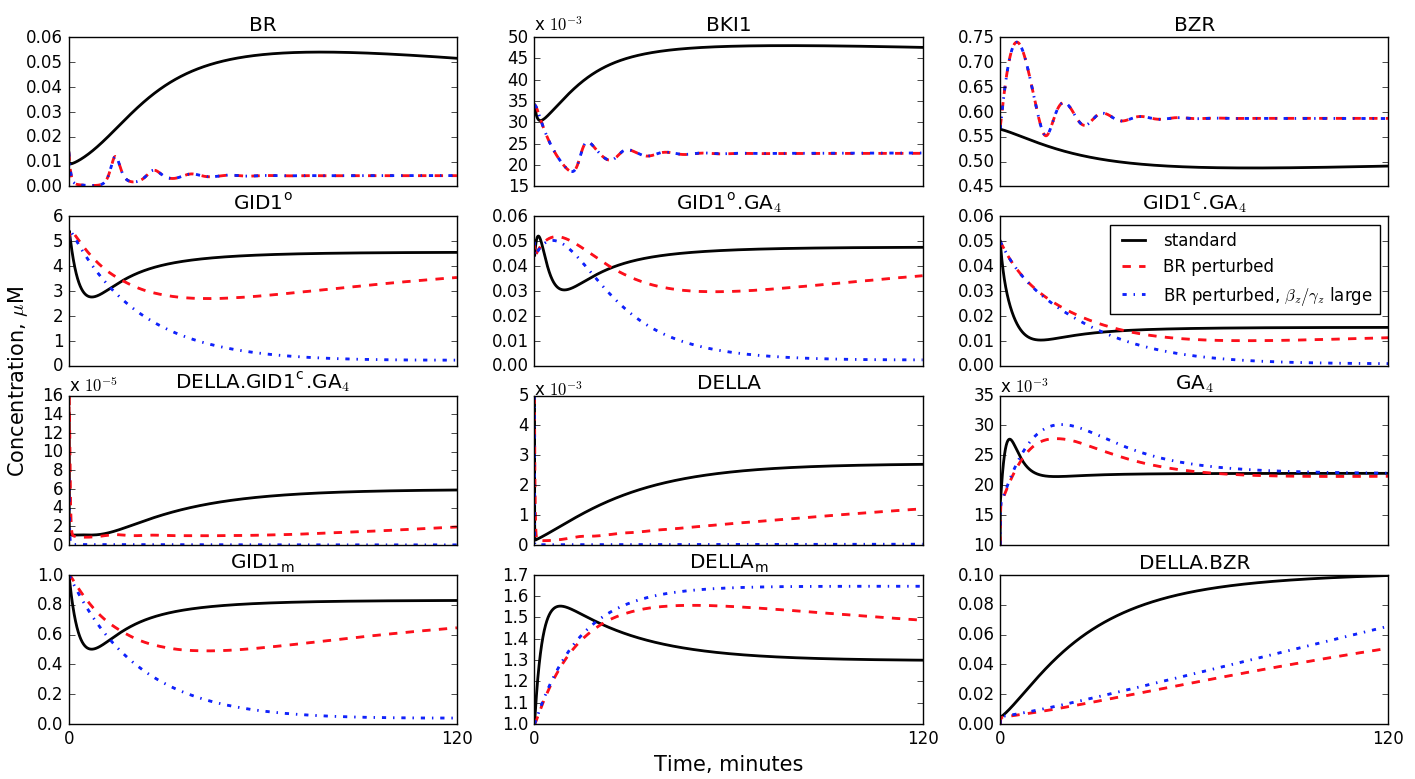}
\caption{Oscillations in the BR pathway lead to short term changes in the dynamics of the GA pathway, and increasing $\beta_{z}/\gamma_{z}$ causes more long term effects on the GA signalling pathway. In the `\textit{standard}' case all parameters are as in Tables \ref{tab:BRode} and \ref{tab:GAode} and $\beta_z=10~\mu {\rm M}^{-1}{\rm min}^{-1}$, $\gamma_z=0.133~{\rm min}^{-1}$. For the `\textit{BR perturbed}' case parameters $\delta_{z}$ and $\rho_{z}$ have had their orders of magnitude increased by three, $\theta_{z} = 41.2\mu{\rm M}^{-1}$, and all other parameters are as in Tables \ref{tab:BRode} and \ref{tab:GAode} and  $\beta_z=10$, $\gamma_z=0.133$. The `\textit{BR perturbed, $\beta_{z}/\gamma_{z}$ large}' has the same parameter values as the `\textit{BR perturbed}', however here $\beta_{z}/\gamma_{z}$ has order of magnitude four, i.e.\ $\beta_{z}/\gamma_{z}\approx 7500$, similar to Fig.~\ref{fig:Crosstalk_Complex_mod}.}
\label{fig:Crosstalk-Osc}
\end{figure}

\noindent Given the observed large effect of crosstalk on the BR pathway as compared with the GA pathway, see Figs~\ref{fig:Crosstalk_Comparison} and \ref{fig:Crosstalk_Complex_mod}, we analysed how directly altering each pathway would affect the other, Fig.~\ref{fig:Crosstalk-Hormonal_Overexpression}. We modelled the overexpression of BR and GA hormones by increasing parameters $\alpha_{b}$ and $\alpha_{g}$. Overexpression of BR (increase of $\alpha_{b}$) resulted in major changes in the BR signalling pathway with increases in the concentrations of BR, BKI1 and BZR, but almost negligible change in the GA signalling pathway. Overexpression of GA had more widespread effects, causing increases in the concentrations of BZR, GID1\textsuperscript{o}.GA\textsubscript{4}, GID1\textsuperscript{c}.GA\textsubscript{4}, DELLA.GID1\textsuperscript{c}.GA\textsubscript{4} and GA\textsubscript{4}, and decreases in the concentrations of BR, BKI1, GID1\textsuperscript{o}, GID1\textsubscript{m} and DELLA.BZR.

\begin{figure}[!ht]\centering
\includegraphics[width=\linewidth]{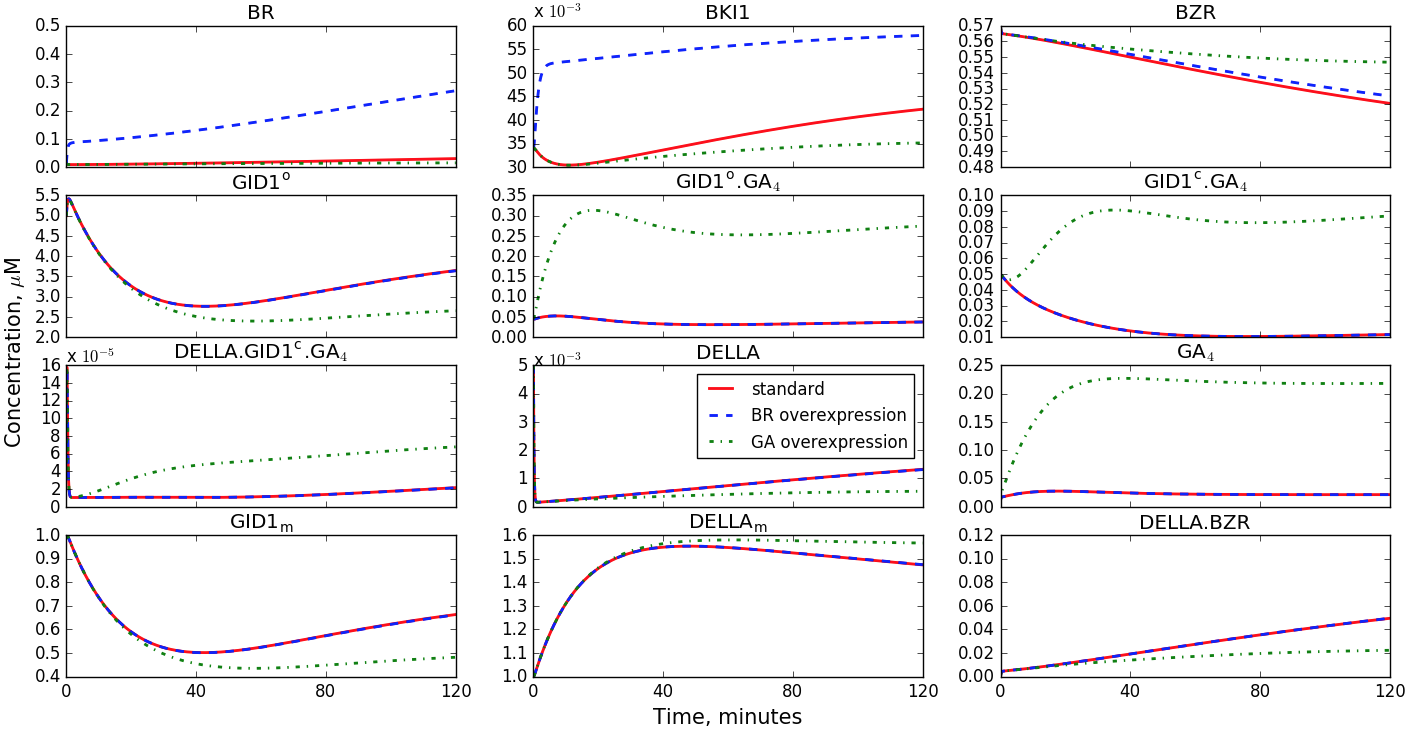}
\caption{Examining the effects of hormonal overexpression on the dynamics of solutions of  model \eqref{eq:GAode1},  (\ref{eq:CrosstalkODE3})-(\ref{eq:CrosstalkODE2}) for the crosstalk between the BR and GA signalling pathways. For the `\textit{standard}' case all parameters are as in Tables \ref{tab:BRode} and \ref{tab:GAode} and $\beta_z=10~\mu {\rm M}^{-1} {\rm min}^{-1}$, $\gamma_z=0.133~{\rm min}^{-1}$;  in the `\textit{BR overexpression}' case, $\alpha_{b}$'s order of magnitude is increased by one;   in the `\textit{GA overexpression}' case $\alpha_{g}$'s order of magnitude is increased by one. 
}
\label{fig:Crosstalk-Hormonal_Overexpression}
\end{figure}

\noindent  In  all cases  of BR-GA interactions discussed  here  we considered the parameter values  for  the BR signalling pathway model as in Table~\ref{tab:BRode}. 
However  similar behaviours are observed  also if considering  the  parameter values  as in Table~\ref{tab:BRode2}, hence we do not present the simulation results for those cases. 

\subsection{Modelling Spatial Heterogeneity in the Crosstalk Signalling}\label{subsec:Crosstalk_PDE}
We also considered the spatial heterogeneity in the crosstalk model. Here, the  coupled PDE-ODE model for the BR signalling pathway (\ref{eq:BRpde1})-(\ref{eq:BRpde3}) was modi\-fied to include the components of the GA signalling pathway. GA was assumed to be able to diffuse freely throughout the cell, and all other components of the GA signalling pathway were assumed to be nuclear-localized. We retain the same assumptions for the components of the BR signalling pathway as in section~\ref{sec:BR_space}. This led to a set of reaction-diffusion equations for BR, BKI1, BZR-p the same as in \eqref{eq:BRpde1}, and for GA:
\begin{IEEEeqnarray}{rCl}\label{eq:Crosstalk_Pde1}
  \partial_{t}g  =  D_{g}\partial_{x}^2 g - \mu_{g}g
\quad 
\text{ in }  \Omega_{c}.
\end{IEEEeqnarray}

\noindent We assume that receptor based interactions of BR, BKI1 and BRI1 remain the same as in \eqref{eq:BRpde2}, and no influx of exogenous GA  on the plasma membrane:
\begin{IEEEeqnarray}{rCl}\label{eq:Crosstalk_Pde2}
  -D_{g}\partial_{x}g  = 0 \quad 
\text{ on  } \Gamma_{c}.
\end{IEEEeqnarray}

\noindent The production of BR, change in phosphorylation status of BZR, and interactions between BZR and DELLA occur in the plant cell nucleus and are modelled as flux boundary conditions and ODEs on the boundary $\Gamma_{n}$:
\begin{IEEEeqnarray}{rCl}\label{eq:Crosstalk_Pde3}
\left.
 \begin{aligned}
  D_{b}\partial_{x}b  = \; & \frac{\tilde{\alpha}_{b}}{1+(\tilde{\theta}_{b}z)^{h_{b}}}\\
  D_{k}\partial_{x}k  =\;  & 0\\
  \frac{d\tilde{z}}{dt}  = \; & \tilde{\delta}_{z}z_{p}k - \rho_{z}\frac{\tilde{z}}{1+(\theta_{z}k)^{h_{z}}} - \tilde{\beta}_{z}\tilde{z}\tilde{d}_{l} + \gamma_{z}\tilde{z}_{d}\\
  D_{z}\partial_{x}z_{p}  =\;  & -\tilde{\delta}_{z}z_{p}k + \rho_{z}\frac{\tilde{z}}{1+(\theta_{z}k)^{h_{z}}}
  \end{aligned}
\; \right\} \quad 
\text{ on  }    \Gamma_{n}.
\end{IEEEeqnarray}

\noindent All processes of the GA signalling pathway apart from degradation  of GA are localised to the nucleus and are modelled by a system of ODEs for GID1\textsuperscript{o}, GA\textsubscript{4}.GID1\textsuperscript{o}, GA\textsubscript{4}.GID1\textsuperscript{c}, DELLA.GA\textsubscript{4}.GID1\textsuperscript{c}, DELLA, GID1\textsubscript{m} and DELLA\textsubscript{m}, and a flux boundary condition for GA on the boundary $\Gamma_{n}$:
\begin{IEEEeqnarray}{rCl}\label{eq:Crosstalk_Pde4}
\left.
 \begin{aligned}
  \frac{d\tilde{r}}{dt}  = \; & -\beta_{g}\tilde{r}g + \gamma_{g}\tilde{r}_{g}^{o} + \tilde{\alpha}_{r}r_{m} - \mu_{r}r\\
  \frac{d\tilde{r}_{g}^{o}}{dt}  =\;  & \beta_{g}\tilde{r}g - (\gamma_{g}+\lambda^{c})\tilde{r}_{g}^{o} + \lambda^{o}\tilde{r}_{g}^{c}\\
  \frac{d\tilde{r}_{g}^{c}}{dt}  = \; & \lambda^{c}\tilde{r}_{g}^{o} - \lambda^{o}\tilde{r}_{g}^{c} - \tilde{\beta}_{d}\tilde{d}_{l}\tilde{r}_{g}^{c} + (\gamma_{d}+\mu_{d})\tilde{r}_{d}\\
  \frac{d\tilde{r}_{d}}{dt}  =\;  & \tilde{\beta}_{d}\tilde{d}_{l}\tilde{r}_{g}^{c} - (\gamma_{d}+\mu_{d})\tilde{r}_{d}\\
  \frac{d\tilde{d}_{l}}{dt}  = \; & -\tilde{\beta}_{d}\tilde{d}_{l}\tilde{r}_{g}^{c} + \gamma_{d}\tilde{r}_{d} + \tilde{\alpha}_{d}d_{m} - \tilde{\beta}_{z}\tilde{z}\tilde{d}_{l} + \gamma_{z}\tilde{z}_{d}\\
  D_{g}\partial_{x}g  = \; & \alpha_{g}\frac{(\tilde{d}_{l}+\phi_{z}\tilde{z})^{4}}{\tilde{\vartheta}_{g}+(\tilde{d}_{l}+\phi_{z}\tilde{z})^{4}} - \beta_{g}\tilde{r}g + \gamma_{g}\tilde{r}_{g}^{o}\\
  \frac{dr_{m}}{dt}  =\;  & \phi_{r}\left(\frac{\tilde{d}_{l}}{\tilde{\theta}_{r}+\tilde{d}_{l}}-r_{m}\right)\\
  \frac{dd_{m}}{dt}  = \; & \phi_{d}\left(\frac{\tilde{\theta}_{d}}{\tilde{\theta}_{d}+\tilde{d}_{l}}-d_{m}\right)
  \end{aligned}
\; \right\} \quad 
\text{ on }  \Gamma_{n}.
\end{IEEEeqnarray}

\noindent Finally crosstalk is modelled via formation of the BZR.DELLA complex
\begin{IEEEeqnarray}{rCl}\label{eq:Crosstalk_Pde5}
  \frac{dz_{d}}{dt} & = & \beta_{z}zd_{l} - \gamma_{z}z_{d} \qquad \qquad 
\text{on $\Gamma_{n}$.}
\end{IEEEeqnarray}

\noindent Similar to model (\ref{eq:BRpde1})-(\ref{eq:BRpde3}), to ensure appropriate units  we obtain  rescaled relations  between  boundary-localised variables in  model    \eqref{eq:Crosstalk_Pde1}-\eqref{eq:Crosstalk_Pde5} and the corresponding variables in model \eqref{eq:GAode1},  \eqref{eq:CrosstalkODE3}-\eqref{eq:CrosstalkODE2}, i.e.\ 
$\tilde r_{k}= l_c r_k$, $\tilde r_{b}= l_c r_b$, $\tilde z= l_c z$, $\tilde r= l_c r$, $\tilde r_{g}^{o}= l_c r_{g}^{o}$, $\tilde r_{g}^{c}= l_cr_{g}^{c}$, $\tilde r_{d}=l_c r_d$, $\tilde d_{l}= l_c d_l$, and $\tilde z_{d}= l_c z_d$.  
 Since $r_{m}$ and $d_{m}$ were already relative quantities with no units there was no need for them to be scaled. To balance units certain parameters also had to be rescaled: parameters $\alpha_{b}$, $\delta_{z}$, $\alpha_{r}$, $\alpha_{d}$, $\alpha_{g}$, $\vartheta_{r}$ and $\vartheta_{d}$ from  \eqref{eq:GAode1},  \eqref{eq:CrosstalkODE3}-\eqref{eq:CrosstalkODE2} were all multiplied by $l_c$ such that $\tilde{\alpha}_{b} = l_{c}\alpha_{b}$ etc.; $\theta_{b}$, $\beta_{z}$ and $\beta_{d}$ were all divided by $l_c$  such that $\tilde{\theta}_{b} = \theta_{b}/l_{c}$ etc.; and $\vartheta_{g}$ was scaled such that $\tilde{\vartheta}_{g} = l_{c}^{4}\vartheta_{g}$. Since GA  molecules have  similar size to BR molecules, we consider $D_{g}=D_{b}$. All other parameters were taken to be the same as those used for model  \eqref{eq:GAode1},   (\ref{eq:CrosstalkODE3})-(\ref{eq:CrosstalkODE2}), i.e.~from Tables \ref{tab:BRode} and  \ref{tab:GAode} and $\beta_z=10~\mu {\rm M}^{-1} {\rm min}^{-1}$, $\gamma_z=0.133~{\rm min}^{-1}$.
 \\
\\
\noindent System \eqref{eq:BRpde1}, \eqref{eq:BRpde2}, \eqref{eq:Crosstalk_Pde1}-\eqref{eq:Crosstalk_Pde5} was solved numerically, and these solutions were averaged over the space and then compared to the solutions of system  \eqref{eq:GAode1},  (\ref{eq:CrosstalkODE3})-(\ref{eq:CrosstalkODE2}).  The behaviour of the two models was similar for most of the parameter sets considered in section \ref{subsec:Crosstalk_ODE}, see e.g.\ standard cases in  Figs.~\ref{fig:Crosstalk-Osc} and \ref{fig:Crosstalk_Pde}, however spatially heterogeneous steady states are characteristic for model \eqref{eq:BRpde1}, \eqref{eq:BRpde2}, \eqref{eq:Crosstalk_Pde1}-\eqref{eq:Crosstalk_Pde5}, see Fig.~\ref{fig:Crosstalk_Pde_2}.  Spatial heterogeneity does have a significant effect in the case where perturbations in the phosphorylation of BZR lead to damped oscillations in the BR signalling pathway, with much different behaviours of solutions of system \eqref{eq:BRpde1}, \eqref{eq:BRpde2}, \eqref{eq:Crosstalk_Pde1}-\eqref{eq:Crosstalk_Pde5} than for  \eqref{eq:GAode1},  (\ref{eq:CrosstalkODE3})-(\ref{eq:CrosstalkODE2}), comparing Figs.~\ref{fig:Crosstalk-Osc} and \ref{fig:Crosstalk_Pde}. When the BR signalling pathway is `strongly perturbed', i.e.~$\delta_{z}$ and $\rho_{z}$ are increased sufficiently high, the oscillations in the components of the BR signalling pathway also cause oscillatory behaviour in the components of the GA signalling pathway, suggesting that spatial heterogeneity may contribute to instability of solutions in extreme cases.

\begin{figure}[!ht]
\includegraphics[width=\linewidth]{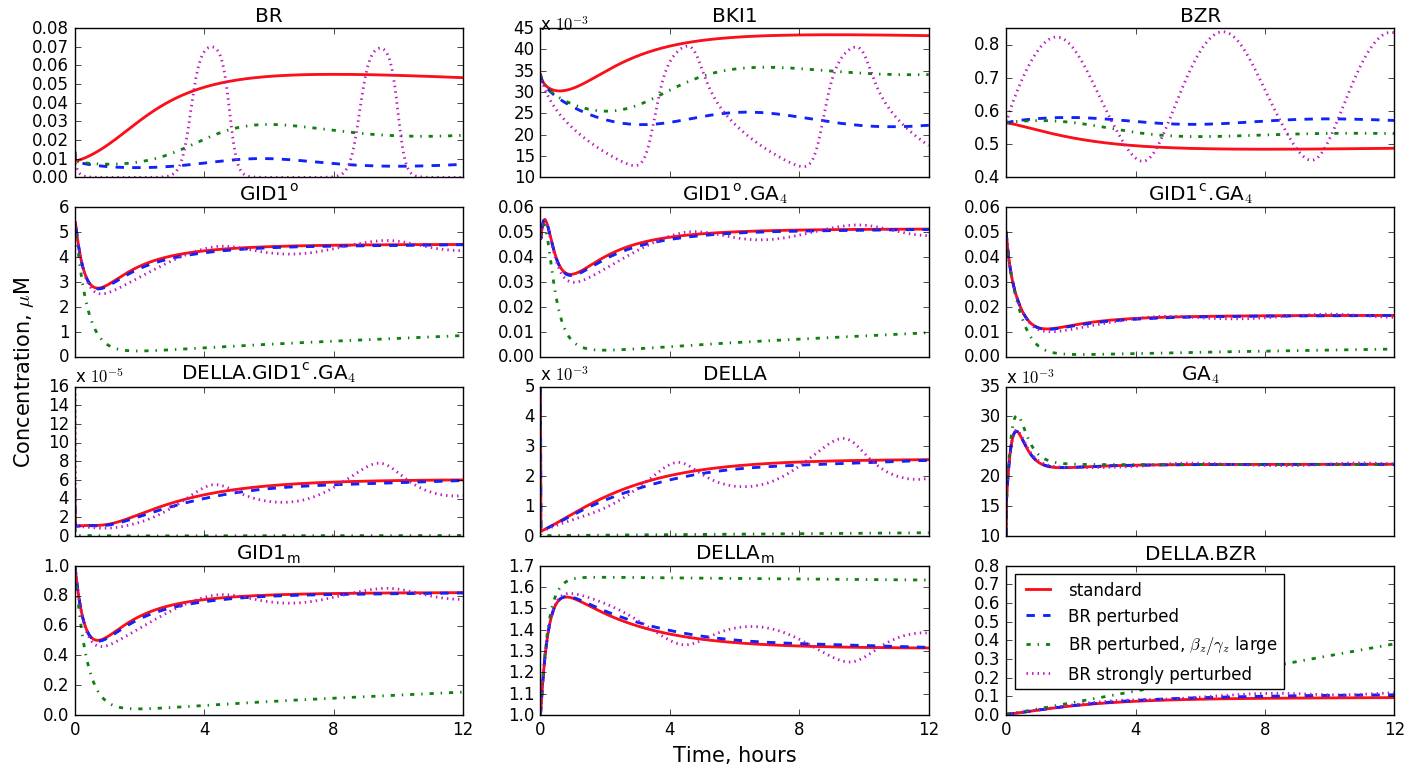}
\caption{Solutions of system \eqref{eq:BRpde1}, \eqref{eq:BRpde2}, (\ref{eq:Crosstalk_Pde1})-(\ref{eq:Crosstalk_Pde5}), averaged over space. For `\textit{standard}'  parameter values are as in Tables \ref{tab:BRode} and \ref{tab:GAode} and $\beta_z=10~\mu {\rm M}^{-1} {\rm min}^{-1}$, $\gamma_z=0.133~{\rm min}^{-1}$. For both `\textit{BR perturbed}' and `\textit{BR perturbed, $\beta_{z}/\gamma_{z}$ large}' $\delta_{z}$ and $\rho_{z}$ have had their orders of magnitude increased by two, $\theta_{z} = 41.2\mu{\rm M}^{-1}$, $\beta_{z}/\gamma_{z}$ has order of magnitude four, i.e.\  $\beta_{z}/\gamma_{z}\approx 7500$, for `\textit{BR perturbed, $\beta_{z}/\gamma_{z}$ large}' only, and all other parameters are as in Tables \ref{tab:BRode} and \ref{tab:GAode}. For the case `\textit{BR strongly perturbed}' $\delta_{z}$ and $\rho_{z}$ have had their orders of magnitude increased by three, $\theta_{z}=41.2\mu{\rm M}^{-1}$, all other parameters as in Tables \ref{tab:BRode} and \ref{tab:GAode} and $\beta_z=10~\mu {\rm M}^{-1} {\rm min}^{-1}$, $\gamma_z=0.133~{\rm min}^{-1}$.}
\label{fig:Crosstalk_Pde}
\end{figure}

\begin{figure}[!ht]
\includegraphics[width=\linewidth]{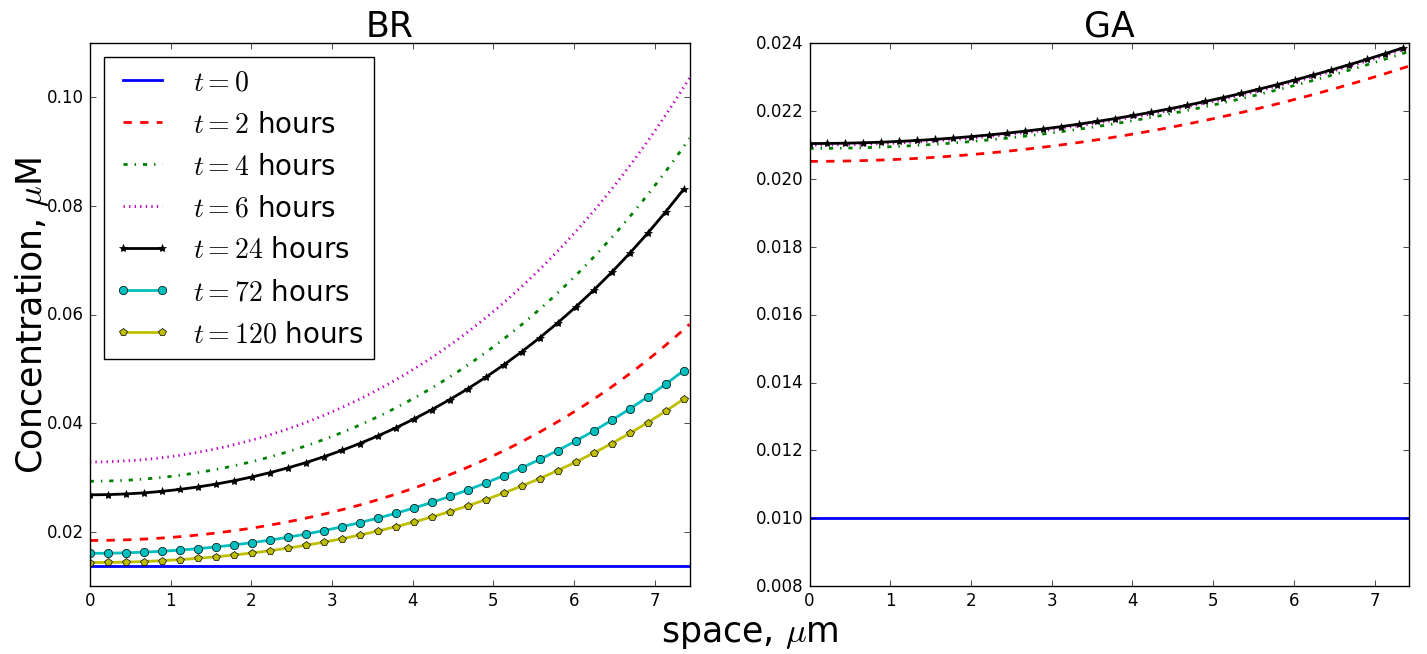}
\caption{
The spatial distribution of BR and GA in the cell segment $\Omega_c$ (solutions of  \eqref{eq:BRpde1}, \eqref{eq:BRpde2}, \eqref{eq:Crosstalk_Pde1}-\eqref{eq:Crosstalk_Pde5})  at different times,  parameter values are as in Tables \ref{tab:BRode} and \ref{tab:GAode} and $\beta_z=10~\mu {\rm M}^{-1} {\rm min}^{-1}$, $\gamma_z=0.133~{\rm min}^{-1}$. The concentration of BR starts of uniform, and is produced at $x = 7.43\>\mu$m, the nucleus, leading to the greatest concentration at this point. The overall concentration increases over the first 8 h, but then decreases steadily, arriving at a spatially heterogeneous steady state after 120 h. The concentration of GA is also greater at $x = 7.43\>\mu$m where it is produced, but GA attains its spatially heterogenous steady state after 24 h, much more quickly than BR.}
\label{fig:Crosstalk_Pde_2}
\end{figure}

\section{Discussion and Conclusion}\label{sec:Discussion}
Plants rely on complex integrated hormonal signalling pathways in order to respond and adapt to changing environmental conditions. The BR and GA signalling pathways have a diverse range of effects on plant growth and developmental processes, some of which overlap. To examine the behaviours of these signalling pathways and their interactions we developed  mathematical models of the BR signalling pathway, and crosstalk between the BR and GA pathways under assumptions of both spatial homogeneity: models (\ref{BRode}) and   \eqref{eq:GAode1},   (\ref{eq:CrosstalkODE3})-(\ref{eq:CrosstalkODE2}), and spatial heterogeneity of signalling processes: models (\ref{eq:BRpde1})-(\ref{eq:BRpde3})  and \eqref{eq:BRpde1}, \eqref{eq:BRpde2}, (\ref{eq:Crosstalk_Pde1})-(\ref{eq:Crosstalk_Pde5}).

\noindent The parameters in the model for the BR signalling pathway (\ref{eq:redBR}) were determined upon validating the model by comparing its numerical solutions to experimental data from \cite{Tanaka_K_2005}. Using numerical optimisation techniques we were able to get a good fit for the diverse data sets corresponding to the different growth conditions considered in \cite{Tanaka_K_2005}. Our calculations resulted in values for $\delta_{z}$ and $\rho_{z}$ of a much lower order of magnitude than the rest of the parameters in the model, suggesting that the phosphorylation of BZR occurs on a much slower time scale than e.g.\  perception of BR by BRI1. 
  The model parameters were optimised under two different constraints on the value of $\beta_{k}$, \eqref{eq:betak1} and \eqref{eq:betak2},  corresponding to two different mechanism governing the interactions of BR, BRI1, and BKI1, where BR and BKI1 binding to and dissociation from BRI1 is instantaneous, or a complex BR.BRI1.BKI1 can be formed, respectively. 
Most parameters had very similar or identical values, with the only notable differences being for (naturally) $\beta_{k}$, and the parameters governing the phosphorylation of BZR: $\delta_{z}$, $\rho_{z}$.
Overall our model fitted the data well for both mechanisms \eqref{eq:betak1} and \eqref{eq:betak2}, with $R^2$ values of $0.89$ and $0.92$ respectively, suggesting that both of these mechanisms may accurately predict the dynamics of the BR signalling pathway.  The only growth condition for which the model had a less accurate fit was for when the growth medium was supplemented with BL. 
\\ 
\\
Qualitative  analysis of system \eqref{eq:BRode} revealed that it has only one steady state for the parameter space $P$, and that this steady state is stable in a neighbourhood of the parameter sets defined in Tables~\ref{tab:BRode} and \ref{tab:BRode2}. Bifurcation analysis showed that for a  bounded set of values of ($\delta_{z}$, $\rho_{z}$), with  
 $\theta_z =41.2~\mu{\rm M}^{-1}$ and all other parameters as in Table \ref{tab:BRode}, system \eqref{eq:BRode} has periodic solutions, i.e.~system \eqref{eq:BRode} undergoes a Hopf bifurcation upon varying parameters $\delta_{z}$, $\rho_{z}$,  corresponding to the dephosphorylation rate and phosphorylation rate  of BZR respectively. 
 No such bifurcations were found when considering  a large range of values for $\delta_z$, $\rho_z$ , and  $\theta_z=41.2~\mu{\rm M}^{-1}$ and all other parameters as in Table~\ref{tab:BRode2}.   However, since the values of $\delta_{z}$, $\rho_{z}$ for which a bifurcation exists are likely to be biologically unrealistic as they are much greater than the values yielded by validation of the model by experimental data, this suggests that within some range of normal function the BR signalling pathway has good stability, which is of crucial importance for proper function of plant tissues.  Results also suggest that if BR and BKI1 associate and dissociate from BRI1 instantaneously, as in mechanism \eqref{eq:betak1}, stability of the pathway is principally dependent upon the phosphorylation state of BZR. Further if BR.BRI1.BKI1 can be formed, as in mechanism \eqref{eq:betak2}, stability of solutions to model \eqref{eq:BRode} for a wide range of all parameters ensures  effective BR homeostasis. 
\\
\\
Numerical solutions  for the spatially heterogeneous model  (\ref{eq:BRpde1})-(\ref{eq:BRpde3}) averaged over space and ODE model   (\ref{eq:redBR})  for the BR signalling pathway
demonstrated  similar behaviours    for the parameter sets when the solutions of   \eqref{eq:BRpde1}-\eqref{eq:BRpde3} and  of   \eqref{eq:redBR}  converge  to the  steady state as $t\to \infty$.  In the oscillatory parameter regime discussed in section~\ref{subsec:bifurcation}, the PDE-ODE model exhibited distinct  behaviour to the ODE model, Fig.~\ref{fig:BR_Pde1},  suggesting that under such conditions the spatial heterogeneity of the signalling pathway has a large influence on the dynamics of the molecules involved in the BR signalling pathway.\\
\\
We investigated the effects of interactions between the BR and GA signalling pathways and the BR-GA signalling crosstalk mechanisms\cite{Li_QF_2013,Tong_H_2014,Unterholzner_S_2015}  by coupling  (\ref{eq:redBR}),  \eqref{eq:GAode1} and (\ref{eq:GAode2}) to obtain  model \eqref{eq:GAode1},  (\ref{eq:CrosstalkODE3})-(\ref{eq:CrosstalkODE2}), in order  to establish which mechanism is more biologically significant \cite{Ross_J_2016,Tong_H_2016,Unterholzner_S_2016}.   We examined the effects of BZR-mediated biosynthesis of GA and  of BZR.DELLA complex formation on the behaviour of the BR and GA signalling pathways. We modelled the cases where there was: no crosstalk, BZR-mediated biosynthesis of GA only, BZR.DELLA complex formation only, and both BZR-mediated biosynthesis of GA and BZR.DELLA complex formation. The cases of no crosstalk and BZR-mediated biosynthesis exhibited similar behaviour to each other, and the cases of BZR.DELLA complex formation and both BZR-mediated biosynthesis of GA and BZR.DELLA complex formation exhibited similar behaviour to each other, Fig.~\ref{fig:Crosstalk_Comparison}. Further, upon variation of the ratio between DELLA and BZR binding thresholds $\phi_{z}$ the change in behaviour of most components of both pathways was negligible, whereas the variations in  BZR.DELLA binding  $\beta_{z}$ and dissociation $\gamma_{z}$ rates  had a strong effect upon the BR signalling pathway, Fig.~\ref{fig:Crosstalk_Complex_mod}. From this we concluded that BZR.DELLA complex formation has a greater effect on the behaviour of the BR and GA signalling pathways. This suggests that interactions between BZR and DELLA exert more influence over the dynamics of the pathways than BZR-mediated biosynthesis of GA and are more likely to be able to promote any effective change in the behaviour of the signalling processes.   Numerical results also demonstrated that for all parameter sets where solutions tended to a steady state, the solutions of model   \eqref{eq:GAode1},   (\ref{eq:CrosstalkODE3})-(\ref{eq:CrosstalkODE2}) and  the averaged over space solutions  of model \eqref{eq:BRpde1}, \eqref{eq:BRpde2}, \eqref{eq:Crosstalk_Pde1}-\eqref{eq:Crosstalk_Pde5}  exhibit similar dynamics,   although spatially heterogenous steady-states are obtained for  BR and GA concentrations, Fig.~\ref{fig:Crosstalk_Pde_2}. \\
\\
To examine whether disturbance of one signalling pathway had a greater effect on the other, we modelled the effects of overexpression of both BR and GA. Overexpression of BR resulted in large changes in the dynamics of the components of the BR signalling pathway only, and overexpression of GA resulted in large changes in the dynamics of the components of the GA signalling pathway, and small changes in the dynamics of the components of the BR signalling pathway, Fig.~\ref{fig:Crosstalk-Hormonal_Overexpression}. In addition  variation of $\beta_{z}$ and $\gamma_{z}$ led to large changes in the dynamics of the BR pathway but  only short term changes in the dynamics  of the GA pathway. From this we concluded that in  general perturbations in the GA signalling pathway have greater influence on the BR signalling pathway than vice versa.\\
\\
To examine the effects of mutations in the BR signalling pathway on the GA pathway, models  \eqref{eq:GAode1},  (\ref{eq:CrosstalkODE3})-(\ref{eq:CrosstalkODE2}) and \eqref{eq:BRpde1}, \eqref{eq:BRpde2}, \eqref{eq:Crosstalk_Pde1}-\eqref{eq:Crosstalk_Pde5} were solved numerically  considering different values for the parameters  $\delta_z$ and $\rho_z$ governing BZR phosphorylation. For the components of the GA signalling pathway solutions to the ODE model showed different short-term behaviour  but tended to the same steady state, whereas  the components of the BR signalling pathway tended to  different steady state. If varying $\beta_{z}$ and $\gamma_{z}$ such that $\beta_{z}/\gamma_{z}$ was large in the case when we have damped oscillation in the components of the BR signalling pathway, the dynamics of some of the components of the GA signalling pathway were significantly altered but the dynamics of the components of the BR signalling pathway were unchanged, Fig.~\ref{fig:Crosstalk-Osc}.  For the spatially heterogeneous model \eqref{eq:BRpde1}, \eqref{eq:BRpde2}, \eqref{eq:Crosstalk_Pde1}-\eqref{eq:Crosstalk_Pde5}, in the oscillatory parameter regime oscillations in the BR signalling pathway propagated into the GA signalling pathway, Fig.~\ref{fig:Crosstalk_Pde}. From this we conclude that only in the case where there is disturbance in both the BR signalling pathway and the mechanism of crosstalk between the BR and GA signalling pathways the dynamics of molecules in the BR signalling pathway have greater influence on the dynamics of molecules in the GA signalling pathway. This, together with the results discussed in the previous paragraph, suggests that under normal crosstalk the stability of individual signalling pathways is maintained even if there are disturbances in the other pathway.\\
\\
To conclude, our analysis of  new mathematical models for BR signalling pathway and the crosstalk between the BR and GA signalling pathways, derived here, provide a better understanding of the dynamics of signalling processes, dependent on model parameters and spatial heterogeneity. Our results for the BR signalling pathway highlight its stability, and suggest that this stability is particularly dependent on the mechanisms governing the phosphorylation state of BZR and the subcellular locations of these processes.  Our results suggest that direct interaction between BZR and DELLA exerts a larger influence on the dynamics of the BR and GA signalling pathways than BZR-mediated biosynthesis of GA, and hence may be the primary mechanism of crosstalk between the two pathways. Our analysis indicates that during normal plant function, the GA signalling pathway exerts more influence over the BR signalling pathway than BR on GA, but mutations in the BR signalling pathway cause BR signalling to exert some short time influence over  GA pathway and  greater influence  when coupled with disturbances in the crosstalk mechanism.  Both BR signalling and crosstalk between BR and GA signalling are important for plant growth and development. Our modelling and analysis results also can be used to model the interactions between growth and signalling processes in order to better understand the influence of the BR and GA signalling pathways on growth and developmental processes in plants.

\section*{Acknowledgments} 
H.R.~Allen   gratefully acknowledges the support of an EPSRC DTA PhD studentship,   research of M.~Ptashnyk was partially supported by the EPSRC First Grant EP/K036521/1. 
\section*{References}
\bibliography{./References}

\begin{thebibliography}{10}

\bibitem{Achard_P_2009}
Patrick Achard and Pascal Genschik.
\newblock Releasing the brakes of plant growth: how gas shutdown della
  proteins.
\newblock {\em Journal of Experimental Botany}, 60(4):1085--1092, 2009.

\bibitem{Achard_P_2008}
Patrick Achard, Fan Gong, Soizic Cheminant, Malek Alioua, Peter Hedden, and
  Pascal Genschick.
\newblock The cold-inducible cbf1 factor–dependent signaling pathway
  modulates the accumulation of the growth-repressing della proteins via its
  effect on gibberellin metabolism.
\newblock {\em The Plant Cell}, 20:2117--2129, 2008.

\bibitem{Ahammed_G_2015}
Golam~J. Ahammed, Xiao-Jian Xia, Xin Li, Kai Shi, Jing-Quan Yu, and Yan-Hong
  Zhou.
\newblock Role of brassinosteroids in plant adaptation to abiotic stresses and
  its interplay with other hormones.
\newblock {\em Current Protein \& Peptide Science}, 16(5):462--473, 2015.

\bibitem{Albrecht_C_2012}
Catherine Albrecht, Freddy Boutrot, C\'{e}cile Segonzac, Benjamin Schwessinger,
  Selena Gimenez-Ibanez, Delphine Chinchilla, John~P. Rathjen, Sacco~C.
  de~Vries, and Cyril Zipfel.
\newblock Brassinosteroids inhibit pathogen-associated molecular
  pattern-triggered immune signaling independent of the receptor kinase bak1.
\newblock {\em PNAS}, 109(1):303--308, 2012.

\bibitem{Amann_H_1990}
Herbert Amann.
\newblock {\em Ordinary Differential Equations: An Introduction to Nonlinear
  Analysis}.
\newblock de Gryuter, illustrated edition, 1990.

\bibitem{Bai_MY_2012}
Ming-Yi Bai, Jian-Xiu Shang, Eunkyoo Oh, Min Fan, Yang Bai, Rodolfo Zentella,
  Tai ping Sun, and Zhi-Yong Wang.
\newblock Brassinosteroid, gibberellin and phytochrome impinge on a common
  transcription module in arabidopsis.
\newblock {\em Nature Cell Biology}, 14(8):810--817, 2012.

\bibitem{Bajguz_A_2009}
Andrzej Bajguz and Shamsul Hayat.
\newblock Effects of brassinosteroids on the plant responses to environmental
  stresses.
\newblock {\em Plant Physiology and Biochemistry}, 47:1--8, 2009.

\bibitem{Belkhadir_Y_2014}
Youssef Belkhadir and Yvon Jaillais.
\newblock The molecular circuitry of brassinosteroid signaling.
\newblock {\em New Phytologist}, 206:522--540, 2015.

\bibitem{Bouquin_T_2001}
Thomas Bouquin, Carsten Meier, Randy Foster, Mads~Eggert Nielsen, and John
  Mundy.
\newblock Control of specific gene expression by gibberellin and
  brassinosteroid.
\newblock {\em Plant Physiology}, 127:450--458, 2001.

\bibitem{Catterou_M_2001}
Manuella Catterou, Fr\'{e}d\'{e}ric Dubois, Hubert Schaller, Laurent Aubanelle,
  Beate Vilcot, Brigitte~S. Sangwan-Norreel, and Rajbir~S. Sangwan.
\newblock Brassinosteroids, microtubules and cell elongation in arabidopsis
  thaliana. ii. effects of brassinosteroids on microtubules and cell elongation
  in the bul1 mutant.
\newblock {\em Planta}, 212:673--683, 2001.

\bibitem{Chung_Y_2013}
Yuhee Chung and Sunghwa Choe.
\newblock The reglation of brassinosteroid biosynthesis in arabidopsis.
\newblock {\em Critical Reviews in Plant Sciences}, 32:396--410, 2013.

\bibitem{Clouse_S_1996}
Steven~D. Clouse.
\newblock Molecular genetic studies confirm the role of brassinosteroids in
  plant growth and development.
\newblock {\em The Plant Journal}, 10(1):1--8, 1996.

\bibitem{Clouse_S_2011}
Steven~D. Clouse.
\newblock Brassinosteroid signal transduction: From receptor kinase activation
  to transcriptional networks regulating plant development.
\newblock {\em The Plant Cell}, 23:1219--1230, 2011.

\bibitem{Clouse_S_2015}
Steven~D. Clouse.
\newblock A history of brassinosteroid research from 1970 through 2005:
  Thirty-five years of phytochemistry, physiology, genes, and mutants.
\newblock {\em Journal of Plant Growth Regulation}, 34:828--844, 2015.

\bibitem{Clouse_S_1998}
Steven~D. Clouse and Jenneth~M. Sasse.
\newblock Brassinosteroids: Essential regulators of plant growth and
  development.
\newblock {\em Annual Review of Plant Physiology and Plant Molecular Biology},
  49:427--451, 1998.

\bibitem{Colebrook_E_2014}
Ellen~H. Colebrook, Stephen~G. Thomas, Andrew~L. Phillips, and Peter Hedden.
\newblock The role of gibberellin signalling in plant responses to abiotic
  stress.
\newblock {\em Journal of Experimental Biology}, 217:67--75, 2014.

\bibitem{Daviere_JM_2013}
Jean-Michel Davi\`{e}re and Patrick Achard.
\newblock Gibberellin signalling in plants.
\newblock {\em Development}, 140(6):1147--1151, 2013.

\bibitem{BLweb}
National~Center for Biotechnology~Information.
\newblock Pubchem compound database, cid=115196.
\newblock https://pubchem.ncbi.nlm.nih.gov/compound/ 115196.
\newblock (accessed July 24, 2017).

\bibitem{Frigola_D_2017a}
David Frigola, Ana~I. Ca$\tilde{\text{n}}$o-Delgado, and Marta
  Iba$\tilde{\text{n}}$es.
\newblock Methods for modeling brassinosteroid-mediated signaling in plant
  development.
\newblock In Eugenia Russinova and Ana~I. Ca$\tilde{\text{n}}$o-Delgado,
  editors, {\em Methods for Modeling Brassinosteroid-Mediated Signaling in
  Plant Development}, chapter~9, pages 103--120. Humana Press, New York, NY,
  2017.

\bibitem{Gallego-Bartolome_J_2012}
Javier Gallego-Bartolom\'{e}, Eugenio~G. Minguet, Federico Grau-Enguix, Mohamad
  Abbas, Antonella Locascio, Stephen~G. Thomas, David Alabad\'{i}, and
  Miguel~A. Bl\'{a}squez.
\newblock Molecular mechanism for the interaction between gibberellin and
  brassinosteroid signaling pathways in arabidopsis.
\newblock {\em PNAS}, 109(33):13446--13451, 2012.

\bibitem{Gibson_L_2012}
Lorna~J. Gibson.
\newblock The hierarchical structure and mechanics of plant materials.
\newblock {\em Journal of the Royal Society Interface}, 9:2749--2766, 2012.

\bibitem{Gordon_S_2009}
Sean~P. Gordon, Vijay~S. Chickarmane, Carolyn Ohno, and Elliot~M. Meyerowitz.
\newblock Multiple feedback loops through cytokinin signaling control stem cell
  number within the arabidopsis shoot meristem.
\newblock {\em PNAS}, 106(38):16529--16534, 2009.

\bibitem{Gruszka_D_2013}
Damian Gruszka.
\newblock The brassinosteroid signaling pathway - new key players and
  interconnections with other signaling networks crucial for plant development
  and stress tolerance.
\newblock {\em International Journal of Molecular Sciences}, 14:8740--8774,
  2013.

\bibitem{Hassard_B_1981}
Brian~D. Hassard, Nicholas~D. Kazarinoff, and Yieh-Hei Wan.
\newblock {\em Theory and Applications of the Hopf Bifurcation}.
\newblock Cambridge University Press, first edition, 1981.

\bibitem{Ibanes_M_2009}
Marta Iba$\tilde{\text{n}}$es, Norma Fabr\`{e}gas, Joanne Chory, and Ana~I.
  Ca$\tilde{\text{n}}$o-Delgado.
\newblock Brassinosteroid signaling and auxin transport are required to
  establish the periodic pattern of arabidopsis shoot vascular bundles.
\newblock {\em PNAS}, 106(32):13630--13635, 2009.

\bibitem{Scipy}
Eric Jones, Travis Oliphant, Pearu Peterson, et~al.
\newblock {SciPy}: Open source scientific tools for {Python}, 2001--.
\newblock [Online; accessed 2017-02-13].

\bibitem{Kim_TW_2012}
Tae-Wuk Kim, Marta Michniewicz, Dominique~C. Bergmann, and Zhi-Yong Wang.
\newblock Brassinosteroid regulate stomatal development by gsk3-mediated
  inhibition of a mapk pathway.
\newblock {\em Nature}, 482:419--422, 2012.

\bibitem{Kim_TW_2010}
Tae-Wuk Kim and Zhi-Yong Wang.
\newblock Brassinosteroid signal transduction from receptor kinases to
  transcription factors.
\newblock {\em Annual Review of Plant Biology}, 61:681--704, 2010.

\bibitem{Li_Jianming_1997}
Jianming Li and Joanne Chory.
\newblock A putative leucine-rich repeat receptor kinase involved in
  brassinosteroid signal transduction.
\newblock {\em Cell}, 90:929--938, 1997.

\bibitem{Li_Lei_2005}
Lei Li and Xing~Wang Deng.
\newblock It runs in the family: regulation of brassinosteroid signaling by the
  bzr1-bes1class of transcription factors.
\newblock {\em Trends in Plant Science}, 10(6):266--268, 2005.

\bibitem{Li_QF_2013}
Qian-Feng Li and Jun-Xian He.
\newblock Mechanisms of signaling crosstalk between brassinosteroids and
  gibberellins.
\newblock {\em Plant Signaling and Behaviour}, 8(7):e24686, 2013.

\bibitem{Li_QF_2012}
Qian-Feng Li, Chunming Wang, Lei Jiang, Shuo Li, Samuel S.~M. Sun, and Jun-Xian
  He.
\newblock An interaction between bzr1 and dellas mediates direct signaling
  crosstalk between brassinosteroids and gibberellins in arabidopsis.
\newblock {\em Science Signaling}, 5(244):ra72, 2012.

\bibitem{Stewart-Lilley_J_2013}
Jodi L.~Stewart Lilley, Yinbo Gan, Ian~A. Graham, and Jennifer~L. Nemhauser.
\newblock The effects of dellas on growth change with developmental stage and
  brassinosteroid levels.
\newblock {\em The Plant Journal}, 76:165--173, 2013.

\bibitem{Liu_J_2010}
Junli Liu, Shaher Mehdi, Jennifer Topping, Petr Tarkowski, and Keith Lindsey.
\newblock Modelling and experimental analysis of hormonal crosstalk in
  arabidopsis.
\newblock {\em Molecular Systems Biology}, 6(373), 2010.

\bibitem{Middleton_A_2010}
Alistair~M. Middleton, John~R. King, Malcolm~J. Bennett, and Markus~R. Owen.
\newblock Mathematical modelling of the aux/iaa negative feedback loop.
\newblock {\em Bulletin of Mathematical Biology}, 72:1383--1407, 2010.

\bibitem{Middleton_A_2012}
Alistair~M. Middleton, Susana \'{U}beda Tom\'{a}s, Jayne Griffiths, Tara
  Holman, Peter Hedden, Stephen~G. Thomas, Andrew~L. Phillips, Michael~J.
  Holdsworth, Malcolm~J. Bennett, John~R. King, and Markus~R. Owen.
\newblock Mathematical modelling elucidates the role of transcriptional
  feedback in gibberellin signaling.
\newblock {\em PNAS}, 109(19):7571--7576, 2012.

\bibitem{Muraro_D_2011}
Daniele Muraro, Helen Byrne, John King, Ute Vo{\ss}, Joseph Kieber, and Malcolm
  Bennett.
\newblock The influence of cytokinin–auxin cross-regulation on cell-fate
  determination in arabidopsis thaliana root development.
\newblock {\em Journal of Theoretical Biology}, 283:152--167, 2011.

\bibitem{Mussig_C_2003}
Carsten M\"{u}ssig, Ga-Hee Shin, and Thoas Altmann.
\newblock Brassinosteroids promote root growth in arabidopsis.
\newblock {\em Plant Physiology}, 133(3):1261--1271, 2003.

\bibitem{Piotrowski_S_2011}
Stephan Piotrowski and Michael Carus.
\newblock Multi-criteria evaluation of lignocellulosic niche crops for use in
  biorefinery processes.
\newblock Technical report, nova-Institut GmbH, H\"{u}rth, Germany, 2011.

\bibitem{Ross_J_2016}
John~J. Ross and Laura~J. Quittenden.
\newblock Intercations between brassinosteroids and gibberellins: Synthesis or
  signaling?
\newblock {\em The Plant Cell}, 28:829--832, 2016.

\bibitem{Ryu_H_2010}
Hojin Ryu, Kangmin Kim, Hyunwoo Cho, and Ildoo Hwang.
\newblock Predominant actions of cytosolic bsu1 and nuclear bin2 regulate
  subcellular localization of bes1 in brassinosteroid signaling.
\newblock {\em Molecules and Cells}, 29:291--296, 2010.

\bibitem{Ryu_H_2007}
Hojin Ryu, Kangmin Kim, Hyunwoo Cho, Joonghyuk Park, Sunghwa Choe, and Ildoo
  Hwang.
\newblock Nucleocytoplasmic shuttling of bzr1 mediated by phosphorylation is
  essential in arabidopsis brassinosteroid signaling.
\newblock {\em The Plant Cell}, 19:2749--2762, 2007.

\bibitem{Sankar_M_2011}
Martial Sankar, Karen~S. Osmont, Jakub Rolcik, Bojan Gujas, Danuse Tarkowska,
  Miroslav Strnad, Ioannis Xenarios, and Christian~S. Hardtke.
\newblock A qualitative continuous model of cellular auxin and brassinosteroid
  signaling and their crosstalk.
\newblock {\em Bioinformatics}, 27(10):1404--142, 2011.

\bibitem{She_J_2011}
Ji~She, Zhifu Han, Tae-Wuk Kim, Jinjing Wang, Wei Cheng, Junbiao Chang, Shuai
  Shi, Jiawei Wang, Maojun Yang, Zhi-Yong Wang, and Jijie chai.
\newblock Structural insight into brassinosteroid perception by bri1.
\newblock {\em Nature}, 474:472--476, 2011.

\bibitem{Shimada_A_2008}
Asako Shimada, Miyako Ueguchi-Tanaka, Toru Nakatsu, Masatoshi Nakajima, Youichi
  Naoe, Hiroko Ohmiya, Hiroaki Kato, and Makoto Matsuoka.
\newblock Structural basis for gibberellin recognition by its receptor gid1.
\newblock {\em Nature}, 456:520--523, 2008.

\bibitem{Shimada_Y_2001}
Yukihisa Shimada, Shozo Fujioka, Narumasa Miyauchi, Masayo Kushiro, Suguru
  Takatsuto, Takahita Nomura, Takao Yokota, Yuji Kamiya, Gerard~J. Bishop, and
  Shigeo Yoshida.
\newblock Brassinosteroid-6-oxidases from arabidopsis and tomato catalyse
  multiple c-6 oxidations in brassinosteroid biosynthesis.
\newblock {\em Plant Physiology}, 126(2):770--779, 2001.

\bibitem{Sieminska_L_1997}
Lucyna Sieminska, Matthew Ferguson, Tadeusz~Waldeck Zerda, and Ernest Couch.
\newblock Diffusion of steroids in porous sol-gel glass: Application in slow
  drug delivery.
\newblock {\em Journal of Sol-Gel Science and Technology}, 8:1105--1109, 1997.

\bibitem{Sturrock_M_2011}
Marc Sturrock, Alan~J. Terry, Dimitris~P. Xirodimas, Alastair~M. Thompson, and
  Mark~A.J. Chaplain.
\newblock Spatio-temporal modelling of the hes1 and p53-mdm2 intracellular
  signalling pathways.
\newblock {\em Journal of Theoretical Biology}, 273:15--31, 2011.

\bibitem{Sunetal_Y_2010}
Yu~Sun, Xi-Ying Fan, Dong-Mei Cao, Wenqiang Tang, Jia-Ying Zhu, Jun-Xian He,
  Ming-Yi Bai, Shengwai Zhu, Eunkyoo Oh, Sunita Patil, Tae-Wuk Kim, Hongkai Ji,
  Wing~Hong Wong, Seung~Y. Rhee, and Zhi-Yong Wang.
\newblock Integration of brassinosteroid signal transduction with the
  transcription network for plant growth regulation in arabidopsis.
\newblock {\em Developmental Cell}, 19:765--777, 2010.

\bibitem{Tanaka_K_2005}
Kiwamu Tanaka, Tadao Asami, Shigeo Yoshida, Yasushi Nakamura, Tomoaki Matsuo,
  and Shigehisa Okamoto.
\newblock Brassinosteroid homeostasis in arabidopsis is ensured by feedback
  expressions of multiple genes involved in its metabolism.
\newblock {\em Plant Physiology}, 138:1117--1125, 2005.

\bibitem{Tanaka_K_2003}
Kiwamu Tanaka, Yasushi Nakamura, Tadao Asami, Shigeo Yoshida, Tomoaki Matsuo,
  and Shigehisa Okamoto.
\newblock Physiological roles of brassinosteroids in early growth of
  arabidopsis: Brassinosteroids have a synergistic relationship with
  gibberellin as well as auxin in light-grown hypocotyl elongation.
\newblock {\em Journal of Plant Growth Regulation}, 22:259--271, 2003.

\bibitem{Tong_H_2016}
Hongning Tong and Chengcai Chu.
\newblock Reply: Brassinosteroid regulates gibberellin synthesis to promote
  cell elongation in rice: Critical comments on ross and quittenden's letter.
\newblock {\em The Plant Cell}, 28:833--835, 2016.

\bibitem{Tong_H_2014}
Hongning Tong, Yunhua Xiao, Dapu Liu, Shaopei Gao, Linchuan Liu, Yanhai Yin,
  Yun Jin, Qian Qian, and Chengcai Chu.
\newblock Brassinosteroid regulates cell elongation by modulating gibberellin
  metabolism in rice.
\newblock {\em The Plant Cell}, 26:4376--4393, 2014.

\bibitem{Ubeda-Tomas_S_2009}
Susana \'{U}beda Tom\'{a}s, Fern\'{a}n Federici, Ilda Casimiro, Gerrit~T.S.
  Beemster, Rishikesh Bhalerao, Ranjan Swarup, Peter Doerner, Jim Haseloff, and
  Malcolm~J. Bennett.
\newblock Gibberellin signaling in the endodermis controls arabidopsis root
  meristem size.
\newblock {\em Current Biology}, 19:1194--1199, 2009.

\bibitem{Ueguchi-Tanaka_M_2005}
Miyako Ueguchi-Tanaka, Motoyuki Ashikari, Masatoshi Nakajima, Hironori Itoh,
  Etsuko Katoh, Masatomo Kobayashi, Ten yuan Chow, Yue ie~C.~Hsing, Hidemi
  Kitano, Isomaro Yamaguchi, and Makoto Matsuoka.
\newblock Gibberellin insensitive dwarf1 encodes a soluble receptor for
  gibberellin.
\newblock {\em Nature}, 437:693--698, 2005.

\bibitem{Unterholzner_S_2015}
Simon~J. Unterholzner, Wilfried Rozhon, Michael Papacek, Jennifer Ciomas, Theo
  Lange, Karl~G. Kugler, Klaus~F. Mayer, Tobias Sieberer, and Brigitte
  Poppenberger.
\newblock Brassinosteroids are master regulators of gibberellin biosynthesis in
  arabidopsis.
\newblock {\em The Plant Cell}, 27:2261--2272, 2015.

\bibitem{Unterholzner_S_2016}
Simon~J. Unterholzner, Wilfried Rozhon, and Brigitte Poppenberger.
\newblock Repy: Interaction between brassinosteroids and gibberellins:
  Synthesis or signaling? in arabidopsis, both!
\newblock {\em The Plant Cell}, 28:836--839, 2016.

\bibitem{VanEsse_G_2012}
G.Wilma van Esse, Simon van Mourik, Hans Stigter, Colette~A. ten Hove, Jaap
  Molenaar, and Sacco~C. de~Vries.
\newblock A mathematical model for brassinosteroid insenstive1-mediated
  signaling in root growth and hypocotyl elongation.
\newblock {\em Plant Physiology}, 160:523--532, 2012.

\bibitem{VanEsse_G_2011}
G.Wilma van Esse, Adrie~H. Westphal, Ramya~Preethi Surendran, Catherine
  Albrecht, Boudewijn van Veen, Jan~Willem Borst, and Sacco~C. de~Vries.
\newblock Quantififcation ofthe brassinosteroid insensitive1 receptor in
  planta.
\newblock {\em Plant Physiology}, 156:1691--1700, 2011.

\bibitem{Frigola_D_2014}
Josep Vilarrasa-Blasi, Mary-Paz Gonz\'alez-Garc\'ia, David Frigola, Norma
  F\`abregas, Konstantinos~G. Alexiou, Nuria L\'opez-Bigas, Susana Rivas, Alain
  Jauneau, Jan~U. Lohmann, Philip~N. Benfey, Marta Iba$\tilde{\text{n}}$es, and
  Ana~I. Ca$\tilde{\text{n}}$o-Delgado.
\newblock Regulation of plant stem cell quiescence by a brassinosteroid
  signaling module.
\newblock {\em Developmental Cell}, 30:36--47, 2014.

\bibitem{Wang_J_2014}
Jie Wang, Jianjun Jiang, Jue Wang, Lei Chen, Shi-Long Fan, Jia-Wei Wu, Xuelu
  Wang, and Zhi-Xin Wang.
\newblock Structural insights into the negative regulation of bri1 signaling by
  bri1-interacting protein bki1.
\newblock {\em Cell Research}, 24:1328--1341, 2014.

\bibitem{Wang_L_2014}
Lu~Wang, Chunfeng Duan, Dapeng Wu, and Yafeng Guan.
\newblock Quantification of endogenous brassinosteroids in sub-gram plant
  tissues by in-line matrix solid-phase dispersion-tandem solid phase
  extraction coupled with high performance liquid chromatography-tandem mass
  spectrometry.
\newblock {\em Journal of Chromatography A}, 1359:44--51, 2014.

\bibitem{Wang_W_2014}
Wenfei Wang, Ming-Yi Bai, and Zhi-Yong Wang.
\newblock The brassinosteroid signaling network - a paradigm of signal
  integration.
\newblock {\em Current Opinion in Plant Biology}, 21:147--153, 2014.

\bibitem{Wang_ZY_2001}
Zhi-Yong Wang, Hideharu Seto, Shozo Fujioka, Shigeo Yoshida, and Joanne Chory.
\newblock Bri1 is a critical component of a plasma-membrane receptor for plant
  steroids.
\newblock {\em Nature}, 410:380--383, 2001.

\bibitem{Yamaguchi_S_2008}
Shinjiro Yamaguchi.
\newblock Gibberellin metabolism and its regulation.
\newblock {\em The Annual Review of Plant Biology}, 59:225--251, 2008.

\bibitem{Yang_CJ_2011}
Cang-Jin Yang, Chi Zhang, Yang-Ning Lu, Jia-Qi Jin, and Xue-Lu Wang.
\newblock The mechanisms of brassinosteroids' action: From signal transduction
  to plant development.
\newblock {\em Molecular Plant}, 4(4):588--600, 2011.

\bibitem{Zhu_JY_2013}
Jia-Ying Zhu, Juthamas Sae-Seaw, and Zhi-Yong Wang.
\newblock Brassinosteroid signalling.
\newblock {\em Development}, 140(8):1615--1620, 2013.

\end{thebibliography}
\appendix
\section{Validation of model (\ref{eq:redBR}) against experimental data}\label{app:constraints}
To improve the optimisation process for parameter estimation we calculated additional constraints for some parameter values using experimental data.\\
\\
We first calculated the level of endogenous BL using a value of 5.26 ng g\textsuperscript{-1} (the midpoint of the reported range in mature plants) which was obtained from \cite{Wang_L_2014}. In order to convert this into an appropriate value in  units of $\mu$M, we first need the density of plant material. We estimate this by noting that the density of cellulose is 1.5 kg L\textsuperscript{-1} \cite{Gibson_L_2012}, and that cellulose constitutes approximately 50\% of plant material \cite{Piotrowski_S_2011}. For the outstanding amount, we make an assumption that rest of the material is largely accounted for by cytoplasm, which we assume to have an approximately equal density to water, i.e.~1 kg L\textsuperscript{-1}. From this we get an estimate of 1.25 kg  L\textsuperscript{-1} for the density of plant material. Finally the molecular weight of BL is 480.686 g mol\textsuperscript{-1} \cite{BLweb}. From this data we calculate the endogenous BR concentration as
\begin{IEEEeqnarray*}{rCl}
\frac{\rm{quantity} \times \rm{density}}{\rm{molecular\> weight}} & = & \frac{(5.26\times 10^{-6}\rm{g\> kg}^{-1})\times (1.25\> \rm{kg\> L}^{-1})}{480.686\> \rm{g\> mol}^{-1}}\\
& = & 0.0137\times 10^{-6} \rm{mol\> L}^{-1} = 13.7\times 10^{-3}\> \mu\rm{M}.
\end{IEEEeqnarray*}
We denote this new constant by $[BR]_{0}$, and  assume it to be the steady state concentration of BR.

\noindent We write $\mu_{b}$ in terms of $\alpha_{b}$, $\theta_{b}$, $\delta_{z}$, $Z_{tot}$, $\rho_{z}$, $\theta_{z}$ and $h_{z}$ by considering the steady state solution of (\ref{eq:redBR}):
\begin{IEEEeqnarray*}{rCl}
\mu_{b} & = & \frac{\alpha_{b}}{[BR]_{0}\left(1 + \left(\theta_{b}\frac{Z_{tot}\delta_{z}[BKI1]_{0}(1+(\theta_{z}[BKI1]_{0})^{h_{z}})}{\rho_{z} + \delta_{z}[BKI1]_{0}(1+(\theta_{z}[BKI1]_{0})^{h_{z}})}\right)^{h_{b}}\right)},
\end{IEEEeqnarray*}
where $[BKI1]_{0}$ denotes the steady state concentration of BKI1, which can be calculated using $[BR]_{0}$.\\
\\
We adopt two approaches for the calculation of $[BKI1]_{0}$ and $\beta_{k}$. First we consider the BR.BRI1 dissociation constant, its value to be dependent on the steady state concentrations of BR, BRI1.BKI1, and BR.BRI1 such that
\begin{equation}\label{eq:K_d}
K_{d} = \frac{r_{k}^{*}b^{*}}{r_{b}^{*}},
\end{equation}
 where $K_{d}$ denotes the BR.BRI1 dissociation constant, and $r_{k}^{*}$, $b^{*}$, and $r_{b}^{*}$ denote the steady state concentrations of BRI1.BKI1, BR, and BR.BRI1 respectively. Taking $b^{*} = [BR]_{0}$, and rewriting $r_{k}^{*}$, $r_{b}^{*}$ in terms of $k^{*}$, the steady state concentration of BKI1, \eqref{eq:K_d} may be rewritten as
\begin{equation}
K_{d}(R_{tot}-K_{tot}+k^{*}) = (K_{tot}-k^{*})[BR]_{0}.
\end{equation}
Thus we may solve for $k^{*}$ to obtain an expression for $[BKI1]_{0}$:
\begin{equation}
[BKI1]_{0} = \frac{K_{tot}[BR]_{0} - K_{d}(R_{tot}-K_{tot})}{K_{d} + [BR]_{0}},
\end{equation}
 and by taking $K_{d} = 11.2$ nM, the midpoint of values reported in \cite{Wang_ZY_2001}, calculate that $[BKI1]_{0} = 34.1\times 10^{-3}$ $\mu$M.
Using this value for $[BKI1]_{0}$ in conjunction with the steady state solution of the second equation in  \eqref{eq:redBR}, we write
\begin{IEEEeqnarray*}{rCl}
\beta_{k} & = & \frac{(K_{tot}-[BKI1]_{0})[BR]_{0}}{(R_{tot}-K_{tot}+[BKI1]_{0})[BKI1]_{0}}\beta_{b}
\\
& = & 0.329\beta_{b},
\end{IEEEeqnarray*}
 giving us an expression for $\beta_{k}$ in terms of $\beta_{b}$.\\
For the second method of calculating $[BKI1]_{0}$ and $\beta_{k}$, we considered that BRI1 has a state where both BR and BKI1 are bound to it, i.e.~the full receptor-based dynamics would look like
\begin{equation}\label{eq:Full_Receptor}
 \begin{aligned}
  \frac{db}{dt} &= -\beta_{1}br_{k} + \gamma_{1}r_{bk}
  \\
  \frac{dk}{dt} &= -\beta_{2}kr_{b} + \gamma_{2}r_{bk}
  \\
  \frac{dr_{k}}{dt} &= -\beta_{1}br_{k} + \gamma_{1}r_{bk}
  \\
  \frac{dr_{b}}{dt} &= -\beta_{2}kr_{b} + \gamma_{2}r_{bk}
  \\
  \frac{dr_{bk}}{dt} &= \beta_{1}br_{k} + \beta_{2}kr_{b} - \gamma_{1}r_{bk} - \gamma_{2}r_{bk}
 \end{aligned}
\end{equation}
where the variable $r_{bk}$ denotes the concentration of BR.BRI1.BKI1, and $\beta_{1}$ and $\beta_{2}$ ($\gamma_{1}$ and $\gamma_{2}$) represent the association (dissociation) rates of BR and BRI1.BKI1, and BKI1 and BR.BRI1 (BR.BRI1.BKI1) respectively. Thus the dissociation constant of BR and BRI1.BKI1, and the dissociation constant of BKI1 and BR.BRI1, denoted $K_{d}$ and $K_{m}$ respectively, may be defined by
\begin{equation*}
\begin{aligned}
K_{d} &  = \frac{\gamma_{1}}{\beta_{1}},  \qquad 
K_{m}  & = \frac{\gamma_{2}}{\beta_{2}}.
\end{aligned}
\end{equation*}
Assuming $r_{bk}$ to be constant allows us to rewrite \eqref{eq:Full_Receptor} as
\begin{equation}
 \begin{aligned}
  \frac{db}{dt} &= -\frac{\gamma_{2}\beta_{1}}{\gamma_{1}+\gamma_{2}}br_{k} + \frac{\gamma_{1}\beta_{2}}{\gamma_{1}+\gamma_{2}}kr_{b},
  \\
  \frac{dk}{dt} &= -\frac{\gamma_{1}\beta_{2}}{\gamma_{1}+\gamma_{2}}kr_{b} + \frac{\gamma_{2}\beta_{1}}{\gamma_{1}+\gamma_{2}}br_{k},
  \\
  \frac{dr_{k}}{dt} &= -\frac{\gamma_{2}\beta_{1}}{\gamma_{1}+\gamma_{2}}br_{k} + \frac{\gamma_{1}\beta_{2}}{\gamma_{1}+\gamma_{2}}kr_{b},
  \\
  \frac{dr_{b}}{dt} &= -\frac{\gamma_{1}\beta_{2}}{\gamma_{1}+\gamma_{2}}kr_{b} + \frac{\gamma_{2}\beta_{1}}{\gamma_{1}+\gamma_{2}}br_{k},
 \end{aligned}
\end{equation}
 which means that we may not only define $\beta_{k}$ and $\beta_{b}$ as 
\begin{equation}
 \begin{aligned}
  \beta_{k} &= \frac{\gamma_{1}\beta_{2}}{\gamma_{1}+\gamma_{2}},
  \qquad 
  \beta_{b} &= \frac{\gamma_{2}\beta_{1}}{\gamma_{1}+\gamma_{2}},
 \end{aligned}
\end{equation}
  but that we may further divide $\beta_{k}$ through by $\beta_{b}$ to show that
\begin{equation*}
\frac{K_{d}}{K_{m}} = \frac{\beta_{k}}{\beta_{b}},
\end{equation*}
and hence
\begin{equation*}
\beta_{k} = \frac{K_{d}}{K_{m}}\beta_{b}.
\end{equation*}
 A value of $K_{m} = 4.28\>\mu$M was reported in \cite{Wang_J_2014}, giving an expression $\beta_{k} = \left(2.62\times 10^{-3}\right) \beta_{b}$. The steady state solution of the second equation in  \eqref{eq:redBR} may thus be written as
\begin{equation}
\left(k^{*}\right)^{2} + \left( (R_{tot}-K_{tot}) + \frac{K_{m}}{K_{d}}[BR]_{0}\right)k^{*} - \frac{K_{m}}{K_{d}}K_{tot}[BR]_{0} = 0,
\end{equation}
 which may be solved for $k^{*}$, and hence an expression for $[BKI1]_{0}$ is found
\begin{equation*}
\begin{aligned}
[BKI1]_{0} &\> =\>  \frac{1}{2}\sqrt{\left((R_{tot}-K_{tot})+\frac{K_{m}}{K_{d}}[BR]_{0}\right)^{2} + 4\frac{K_{m}}{K_{d}}K_{tot}[BR]_{0}}\\
 &\>\>\>\>\>-\>\frac{1}{2}\left((R_{tot}-K_{tot})+\frac{K_{m}}{K_{d}}[BR]_{0}\right),
 \end{aligned}
\end{equation*}
 giving $[BKI1]_{0} = 61.3\times 10^{-3}\>\mu$M.\\
\\
In order to check whether a more accurate estimation of the parameters was given by fitting a non-dimensional model which does not scale time by one of the fitting variables:
\begin{equation}\label{eq:ndfit}
\begin{aligned}
\frac{d\bar{b}}{d\bar{t}} &= \bar{\beta}_{k}(\kappa - 1 + \bar{k})\bar{k} - \bar{\beta}_{b}(1-\bar{k})\bar{b} + \bar{\mu}_{b}\left(\frac{1}{1+(\bar{\theta}_{b}\bar{z})^{h_{b}}} - \bar{b}\right),
\\
\epsilon\frac{d\bar{k}}{d\bar{t}} &= \bar{\beta}_{b}(1-\bar{k})\bar{b} - \bar{\beta}_{k}(\kappa - 1 + \bar{k})\bar{k},
\\
\frac{d\bar{z}}{d\bar{t}} &= \bar{\delta}_{z}(1-\bar{z})\bar{k} - \bar{\rho}_{z}\frac{\bar{z}}{1+(\bar{\theta}_{z}\bar{z})^{h_{z}}}, 
\end{aligned}
\end{equation}
 where
\begin{center}
\begin{tabular}{lll}
$\bar{\beta}_{b} = 60\beta_{b}K_{tot}$, & $\bar{\beta}_{k} = 60\frac{\beta_{k}(K_{tot})^{2}\mu_{b}}{\alpha_{b}}$, & $\kappa = \frac{R_{tot}}{K_{tot}}$,\\
$\bar{\mu}_{b} = 60\mu_{b}$, & $\bar{\theta}_{b} = \theta_{b}Z_{tot}$, & $\epsilon = \frac{K_{tot}\mu_{b}}{\alpha_{b}}$,\\
$\bar{\delta}_{z} = 60\delta_{z}K_{tot}$, & $\bar{\rho}_{z} = 60\rho_{z}$, & $\bar{\theta}_{z} = \theta_{z}K_{tot}$,
\end{tabular}
\end{center}
\eqref{eq:ndfit} was also fitted to the experimental data. The parameters whose dimensional values could be found had close agreement with those found fitting model \eqref{eq:redBR}, Tables \ref{tab:nd1} and \ref{tab:nd2}, and so it is better to consider the fitting of the dimensional model \eqref{eq:redBR} to the experimental data.
\begin{table}[h!]\centering
\begin{tabular}{ccc|ccc}
\toprule
Constant & Value & Units & Constant & Value & Units\\
\midrule
$\beta_{b}$ & 8.33 & $\mu$M\textsuperscript{-1} min\textsuperscript{-1} & $\beta_{k}$ & 2.73 & $\mu$M\textsuperscript{-1} min\textsuperscript{-1}\\
$\alpha_{b}$ & 0.27 & $\mu$M min\textsuperscript{-1} & $\mu_{b}$ & 3.58 & min\textsuperscript{-1}\\
$\delta_{z}$ & $9.97\times 10^{-4}$ & $\mu$M\textsuperscript{-1}min\textsuperscript{-1} & $\rho_{z}$ & $1.30\times 10^{-4}$ & min\textsuperscript{-1}\\
$\theta_{z}$ & 4.05 & $\mu$M\textsuperscript{-1} & $h_{z}$ & 6 & \\
\bottomrule
\end{tabular}
\caption{Transformed parameters when fitting the non-dimensionalised model \eqref{eq:ndfit} against experimental data from \cite{Tanaka_K_2005}, and using 
\eqref{eq:mub} and \eqref{eq:betak1}. Due to the nature of the non-dimensionalisation, the dimensional parameters $\theta_{b}$ and $Z_{tot}$ cannot be found.}
\label{tab:nd1}
\end{table}

\begin{table}[h!]\centering
\begin{tabular}{ccc|ccc}
\toprule
Constant & Value & Units & Constant & Value & Units\\
\midrule
$\beta_{b}$ & 8.06 & $\mu$M\textsuperscript{-1} min\textsuperscript{-1} & $\beta_{k}$ & $2.11\times 10^{-2}$ & $\mu$M\textsuperscript{-1} min\textsuperscript{-1}\\
$\alpha_{b}$ & 0.27 & $\mu$M min\textsuperscript{-1} & $\mu_{b}$ & 3.68 & min\textsuperscript{-1}\\
$\delta_{z}$ & $1.77\times 10^{-3}$ & $\mu$M\textsuperscript{-1}min\textsuperscript{-1} & $\rho_{z}$ & $4.33\times 10^{-4}$ & min\textsuperscript{-1}\\
$\theta_{z}$ & 3.98 & $\mu$M\textsuperscript{-1} & $h_{z}$ & 6 & \\
\bottomrule
\end{tabular}
\caption{Transformed parameters when fitting the non-dimensionalised model \eqref{eq:ndfit} against experimental data from \cite{Tanaka_K_2005}, and using \eqref{eq:mub} and \eqref{eq:betak2}. Due to the nature of the non-dimensionalisation, the dimensional parameters $\theta_{b}$ and $Z_{tot}$ cannot be found.}
\label{tab:nd2}
\end{table}

\begin{table}[h!]\centering
\resizebox{\linewidth}{!}{
\begin{tabular}{c|cccccccc}
\toprule
 & WT1 & BRZ1 & WT2 & BL2 & MUT3 & BRZ3 & MUT4 & BL4\\
\midrule
$\beta_{b}$ & 8.33 & 8.49 & 8.50 & 8.42 & $2.36\times 10^{-3}$ & $2.46\times 10^{-3}$ & $2.43\times 10^{-3}$ & $2.31\times 10^{-3}$\\
$\beta_{k}$ & 2.73 & 2.72 & 2.72 & 2.80 & 3.26$\times 10^{-12}$ & 3.10$\times 10^{-12}$ & 2.95$\times 10^{-12}$ & 3.09$\times 10^{-12}$\\
$\alpha_{b}$ & 0.27 & 6.29$\times 10^{-2}$ & 6.29$\times 10^{-2}$ & 6.08$\times 10^{-2}$ & 0.27 & 0.27 & 0.28 & 0.27\\
$\theta_{b}$ & 1.96 & 2.01 & 2.01 & 2.11 & 1.96 & 1.96 & 1.92 & 1.88\\
$\mu_{b}$ & 3.58 & 3.54 & 3.54 & 13.98 & 3.50 & 3.40 & 3.40 & 3.57\\
$\delta_{z}$ & 1.02$\times 10^{-3}$ & 1.05$\times 10^{-3}$ & 1.07$\times 10^{-3}$ & 1.10$\times 10^{-3}$ & 1.02$\times 10^{-3}$ & 1.01$\times 10^{-3}$ & 1.01$\times 10^{-3}$ & 1.06$\times 10^{-3}$\\
$Z_{tot}$ & 2.65 & 2.72 & 2.72 & 2.75 & 2.65 & 2.58 & 2.58 & 2.58\\
$\rho_{z}$ & 1.33$\times 10^{-4}$ & 1.39$\times 10^{-4}$ & 1.32$\times 10^{-4}$ & 1.32$\times 10^{-4}$ & 1.34$\times 10^{-4}$ & 1.33$\times 10^{-4}$ & 1.29$\times 10^{-4}$ & 1.22$\times 10^{-4}$\\
$\theta_{z}$ & 3.95 & 4.10 & 4.30 & 4.06 & 3.76 & 3.93 & 3.91 & 3.91\\
$h_{z}$ & 6.02 & 5.82 & 5.52 & 5.65 & 5.95 & 5.75 & 5.76 & 5.76\\
\bottomrule
\end{tabular}}
\caption{Fitted model parameters obtained by fitting model \eqref{eq:redBR} under conditions \eqref{eq:betak1},\eqref{eq:mub} to experimental data for each growth condition. Here WT1 corresponds to the case where wild-type plants were grown under control conditions, BZR1 to the case where wild-type plants were grown in a medium containing BRZ, WT2 to the case where wild-type plants had been grown in a medium containing BRZ 5$\mu$M for two days, BL2 to the case where wild-type plants had been grown in a medium containing BRZ 5$\mu$M for two days and then supplemented with 0.1$\mu$M BL, MUT3 corresponds to the case where mutant plants were grown under control conditions, BZR4 to the case where mutant plants were grown in a medium containing BRZ, MUT4 to the case where mutant plants had been grown in a medium containing BRZ 5$\mu$M for two days, BL4 to the case where mutant plants had been grown in a medium containing BRZ 5$\mu$M for two days and then supplemented with 0.1$\mu$M BL.}
\label{tab:All_Params}
\end{table}

\section{Linearised Stability Analysis for the Model (\ref{eq:BRode}) of the BR Signalling Pathway}\label{app:A}

\subsection{Proof of positive invariance of $M= [0, 1 + \beta_k\kappa]\times[0,1]\times[0,1]$.}\label{Proof_Theo_1}

To prove that $M$ is positive invariant, we first  consider $M_{1} := \{(b,k,z)\in\mathbb{R}^{3}\; | \; b,k,z\geq 0\}$ and show the non-negativity of solutions of (\ref{eq:BRode}).  Restricting $b=0$, we obtain that $\textbf{f}\cdot(1,0,0) = \beta_{k}(\kappa -1+k)k+\frac{1}{1+(\theta_{b}z)^{h_{b}}}$, and similarly that $\textbf{f}\cdot(0,1,0) = \frac{\beta_{b}}{\epsilon}b$ restricting $k=0$, and $\textbf{f}\cdot(0,0,1) = \delta_{z}k$ restricting $z=0$. Since $b,k,z\geq 0$ we also obtain  that $\textbf{f}(u)\cdot\textbf{n}(u)\geq 0\> \forall u\in \partial M_{1}$ and  any flow that starts in $M_{1}$ must remain in $M_{1}$ for all $t$, and hence $M_{1}$ is a positive invariant set. 
To show  boundedness of solutions of (\ref{eq:BRode}), consider $M_{2} := \{(b,k,z)\in\mathbb{R}^{3} \; |\;  b\leq 1+\beta_{k}\kappa, k\leq 1, z\leq 1\}$. Restricting $b=1+\beta_{k}\kappa$, we obtain  that $\textbf{f}\cdot(-1,0,0) = (\beta_{b}(1-k)+1)(1+\beta_{k}\kappa)-\beta_{k}(\kappa-1+k)k-\frac{1}{1+(\theta_{b}z)^{h_{b}}}$, and similarly that $\textbf{f}\cdot(0,-1,0) = \frac{\beta_{k}}{\epsilon}\kappa$ restricting $k=1$, and $\textbf{f}\cdot(0,0,-1) = \frac{\rho_{z}}{1+(\theta_{z}k)^{h_{z}}}$ restricting $z=1$. Note that $\textbf{f}\cdot(-1,0,0)$ has no critical points for $0\leq k\leq 1$, and has value $\beta_{b}(1+\beta_{k}\kappa)$ for $k=0$, and 0 for $k=1$. Thus $\textbf{f}(u)\cdot\textbf{n}(u)\geq 0\> \forall (u)\in \partial M_{2}$ and  any flow that starts in $M_{2}$ must remain in $M_{2}$ for all $t$, hence $M_{2}$ is a positive invariant set. 
Thus $M := M_{1}\cap M_{2}$ is a positive invariant region  \cite{Amann_H_1990}.

\subsection{Calculations for the proof of Theorem \ref{thm:steady}}\label{Calc_Theo_2}
Any steady state $(b^{*},k^{*},z^{*})$ of (\ref{eq:BRode}) satisfies
\begin{IEEEeqnarray*}{rCl}
0 & = & \frac{1}{1+(\theta_{b}z^{*})^{h_{b}}} - b^{*},\label{eq:SS1}
\\
0 & = & \beta_{b}(1-k^{*})b^{*} - \beta_{k}(\kappa -1+k^{*})k^{*},\label{eq:SS2}
\\
0 & = & -\rho_{z}\frac{z^{*}}{1+(\theta_{z}k^{*})^{h_{z}}} + \delta_{z}(1-z^{*})k^{*}.\label{eq:SS3}
\end{IEEEeqnarray*}
 Using simple algebraic manipulation on the system above, we find $b^{*}$ and $z^{*}$ in terms of $k^{*}$
\begin{IEEEeqnarray*}{rCl}
b^{*} & = & \dfrac{1}{1 + \left(\theta_{b}\dfrac{\delta_{z}k^{*}\left(1+(\theta_{z}k^{*})^{h_{z}}\right)}{\rho_{z} + \delta_{z}k^{*}\left(1+(\theta_{z}k^{*})^{h_{z}}\right)}\right)^{h_{b}}} = \dfrac{\beta_{k}\left(\kappa -1+k^{*}\right)k^{*}}{\beta_{b}\left(1-k^{*}\right)},
\\
z^{*} & = & \dfrac{\delta_{z}k^{*}\left(1+(\theta_{z}k^{*})^{h_{z}}\right)}{\rho_{z} + \delta_{z}k^{*}\left(1+(\theta_{z}k^{*})^{h_{z}}\right)}.
\end{IEEEeqnarray*}
Then $k^{*}$ is defined as a root of the following non-linear function
\begin{equation*}
 \begin{aligned}
  g(k^{*}) : &= \beta_{k}\left(\kappa -1+k^{*}\right)k^{*}\left(1 + \left(\dfrac{\theta_{b}\delta_{z}k^{*}\left(1+(\theta_{z}k^{*})^{h_{z}}\right)}{\rho_{z} + \delta_{z}k^{*}\left(1+(\theta_{z}k^{*})^{h_{z}}\right)}\right)^{h_{b}}\right)
  \\
  &\>\>\>\>\> -\> \beta_{b}\left(1-k^{*}\right).
 \end{aligned}
\end{equation*}
The derivative of $g$
\begin{equation*}
 \begin{aligned}
  \frac{dg}{dk^{*}} &= \beta_{b} + \beta_{k}(\kappa -1+2k^{*})\left(1+\left(\theta_{b}\frac{\delta_{z}k^{*}\left(1+(\theta_{z}k^{*})^{h_{z}}\right)}{\rho_{z}+\delta_{z}k^{*}\left(1+(\theta_{z}k^{*})^{h_{z}}\right)}\right)^{h_{b}}\right)
  \\
  & \>\>\>\> +\> \beta_{k}(\kappa -1+k^{*})k^{*}\times
  \\
  & \>\>\>\>\times \left(\frac{h_{b}\theta_{b}\delta_{z}\rho_{z}\left(1+(1+h_{z})(\theta_{z}k^{*})^{h_{z}}\right)}{\left(\rho_{z}+\delta_{z}k^{*}(1+(\theta_{z}k^{*})^{h_{z}})\right)^{2}}\left(\frac{\theta_{b}\delta_{z}k^{*}\left(1+(\theta_{z}k^{*})^{h_{z}}\right)}{\rho_{z}+\delta_{z}k^{*}\left(1+(\theta_{z}k^{*})^{h_{z}}\right)}\right)^{h_{b}-1}\right)
 \end{aligned}
\end{equation*}
 which is positive for all $p\in P$ and $k^{*}\in [0,1]$.

\subsection{Characteristic Equation}\label{Char_eq}
The Jacobian for system \eqref{eq:BRode} evaluated at the  steady state  $(b^\ast, k^\ast, z^\ast)$ is given by 
\begin{equation}\label{eq:Jacobian}
\small  \begin{pmatrix}
-\beta_{b}(1-k^{*})-1 & \beta_{k}(\kappa -1+2k^{*})+\beta_{b}b^{*} & -\frac{h_{b}\theta_{b}(\theta_{b}z^{*})^{h_{b}-1}}{\left(1+(\theta_{b}z^{*})^{h_{b}}\right)^{2}}\\
\frac{\beta_{b}(1-k^{*})}{\epsilon} & -\frac{\beta_{b}b^{*}}{\epsilon}-\frac{\beta_{k}(\kappa-1+2k^{*})}{\epsilon} & 0\\
0 & \delta_{z}(1-z^{*}) + \frac{\rho_{z}h_{z}\theta_{z}z^{*}(\theta_{z}k^{*})^{h_{z}-1}}{\left(1+(\theta_{z}k^{*})^{h_{z}}\right)^{2}} & -\delta_{z}k^{*}-\frac{\rho_{z}}{1+(\theta_{z}k^{*})^{h_z}}
\end{pmatrix} \hspace{-0.1 cm}
\end{equation}

The coefficients of the characteristic equation associated to Jacobian $J$ \eqref{eq:Jacobian} are defined as follows
\begin{IEEEeqnarray*}{rCl}
a_{2} & = & \beta_{b}(1-k^{*})+1+\frac{\beta_{b}}{\epsilon}b^{*}+\frac{\beta_{b}}{\epsilon}(\kappa -1+2k^{*})+\delta_{z}k^{*}+\frac{\rho_{z}}{1+(\theta_{z}k^{*})^{h_{z}}},
\\
a_{1} & = & \frac{1}{\epsilon}\left(\beta_{b}b^{*}+\beta_{k}(\kappa-1+2k^{*})\right)\left(\beta_{b}(1-k^{*})+2+\delta_{z}k^{*}+\frac{\rho_{z}}{1+(\theta_{z}k^{*})^{h_{z}}}\right)\\
&& +\>\left(\beta_{b}(1-k^{*})+1\right)\left(\delta_{z}k^{*}+\frac{\rho_{z}}{1+(\theta_{z}k^{*})^{h_{z}}}\right),
\\
a_{0} & = & \frac{1}{\epsilon}\beta_{b}(1-k^{*})\frac{\theta_{b}h_{b}(\theta_{b}z^{*})^{h_{b}-1}}{(1+(\theta_{b}z^{*})^{h_{b}})^{2}}\left(\delta_{z}(1-z^{*})+\frac{\rho_{z}z^{*}\theta_{z}h_{z}(\theta_{z}k^{*})^{h_{z}-1}}{(1+(\theta_{z}k^{*})^{h_{z}})^{2}}\right)\\
&& +\>\frac{1}{\epsilon}(\beta_{b}b^{*}+\beta_{k}(\kappa -1+2k^{*}))\left(\delta_{z}k^{*}+\frac{\rho_{z}}{1+(\theta_{z}k^{*})^{h_{z}}}\right),
\end{IEEEeqnarray*}

\noindent where $(b^{*}, k^{*}, z^{*})$ is a steady state of the system \eqref{eq:BRode}.

\section{Non-dimensionalisation of Model (\ref{eq:BRpde1})-(\ref{eq:BRpde3}) of the BR Signallng Pathway}\label{app:BRpde}
Since in the PDE-ODE model (\ref{eq:BRpde1})-(\ref{eq:BRpde3}) some of the variables are defined on the boundaries we must consider them in units of   $mol/m^{2}$ instead of $mol/m^{3}$, compared to the ODE model \eqref{eq:redBR}. To adapt the units, we scale by the length of the cell segment, i.e. $\tilde{r}_{k} = l_{c}r_{k}$, $\tilde{r}_{k} = l_{c}r_{k}$ and $\tilde{z} = l_{c}z$. To preserve the balance of units in the system, we must scale the following parameters $\tilde{\alpha}_{b} = l_{c}\alpha_{b}$, $\tilde{\theta}_{b} = \theta_{b}/l_{c}$, $\tilde{\delta}_{z} = l_{c}\delta_{z}$.\\
Applying this scaling and non-dimensionalising via $t = t^{*}\bar{t}$, $x = l_{c}y$, $b = b^{*}\bar{b}$, $k = k^{*}\bar{k}$, $\tilde{r}_{k} = r_{k}^{*}\bar{r}_{k}$, $\tilde{r}_{b} = r_{b}^{*}\bar{r}_{b}$, $\tilde{z} = z^{*}\bar{z}$, $z_{p} = z_{p}^{*}\bar{z}_{p}$, system (\ref{eq:BRpde1})-(\ref{eq:BRpde3}) is transformed to
\begin{IEEEeqnarray}{rCl}
\left.
 \begin{IEEEeqnarraybox}[
   \IEEEeqnarraystrutmode
  ][c]{rCl}
\partial_{\bar{t}}\bar{b} & = & \frac{D_{b}t^{*}}{(x^{*})^{2}}\partial_{y}^2\bar{b} - \mu_{b}t^{*}\bar{b}
\\
\partial_{\bar{t}}\bar{k} & = & \frac{D_{k}t^{*}}{(x^{*})^{2}} \partial_{y}^2\bar{k}
\\
\partial_{\bar{t}}\bar{z}_{p} & = & \frac{D_{z}t^{*}}{(x^{*})^{2}} \partial_{y}^2\bar{z}_{p}
\end{IEEEeqnarraybox}
 \, \right\}
 \text{in $\bar{\Omega}_{c}$,}
\end{IEEEeqnarray}

\begin{IEEEeqnarray}{rCl}
\left.
 \begin{IEEEeqnarraybox}[
   \IEEEeqnarraystrutmode
  ][c]{rCl}
-\frac{D_{b}t^{*}}{(x^{*})^{2}}\partial_{y}\bar{b} & = & \frac{\beta_{k}t^{*}r_{b}^{*}k^{*}}{x^{*}b^{*}}\bar{r}_{b}\bar{k} - \frac{\beta_{b}t^{*}r_{k}^{*}}{x^{*}}\bar{b}\bar{r}_{k}
\\
-\frac{D_{k}t^{*}}{(x^{*})^{2}}\partial_{y}\bar{k} & = & \frac{\beta_{b}t^{*}r_{k}^{*}b^{*}}{x^{*}k^{*}}\bar{b}\bar{r}_{k} - \frac{\beta_{k}t^{*}r_{b}^{*}}{x^{*}}\bar{r}_{b}\bar{k}
\\
\frac{d\bar{r}_{k}}{d\bar{t}} & = & \beta_{k}k^{*}t^{*}\bar{r}_{b}\bar{k} - \beta_{b}b^{*}t^{*}\bar{b}\bar{r}_{k}
\\
\frac{d\bar{r}_{b}}{d\bar{t}} & = & \beta_{b}b^{*}t^{*}\bar{b}\bar{r}_{k} - \beta_{k}k^{*}t^{*}\bar{r}_{b}\bar{k}
\\
-\frac{D_{z}t^{*}}{(x^{*})^{2}}\partial_{y}\bar{z}_{p} & = & 0
\end{IEEEeqnarraybox}
 \, \right\}
 \text{on $\bar{\Gamma}_{c}$,}
\end{IEEEeqnarray}

\begin{IEEEeqnarray}{rCl}
\left.
 \begin{IEEEeqnarraybox}[
   \IEEEeqnarraystrutmode
  ][c]{rCl}
\frac{D_{b}t^{*}}{(x^{*})^{2}}\partial_{y}\bar{b} & = & \frac{l_{c}\alpha_{b}t^{*}}{x^{*}b^{*}(1 + (\frac{\theta_{b}z^{*}\bar{z}}{l_{c}})^{h_{b}})}
\\
\frac{D_{k}t^{*}}{(x^{*})^{2}}\partial_{y}\bar{k} & = & 0
\\
\frac{d\bar{z}}{d\bar{t}} & = & \frac{l_{c}\delta_{z}t^{*}p^{*}k^{*}}{z^{*}}\bar{z}_{p}\bar{k} - \rho_{z}t^{*}\frac{\bar{z}}{1+(\theta_{z}k^{*}\bar{k})^{h_{z}}}
\\
\frac{D_{z}t^{*}}{(x^{*})^{2}}\partial_{y}\bar{z}_{p} & = & -\frac{l_{c}\delta_{z}t^{*}k^{*}}{x^{*}}\bar{z}_{p}\bar{k} + \frac{\rho_{z}t^{*}z^{*}}{x^{*}p^{*}}\frac{\bar{z}}{1+(\theta_{z}k^{*}\bar{k})^{h_{z}}}
\end{IEEEeqnarraybox}
 \, \right\}
 \text{on $\bar{\Gamma}_{n}$.}
\end{IEEEeqnarray}
Take $t^{*}=1/\mu_{b}$, $x^{*}=l_{c}$, $b^{*}=\alpha_{b}/\mu_{b}$, $k^{*}=K_{tot}$, $r_{k}^{*}=l_{c}R_{tot}$, $r_{b}^{*}=l_{c}R_{tot}$, $z^{*}=l_{c}Z_{tot}$ and $p^{*}=Z_{tot}$. Then
\begin{IEEEeqnarray}{rCl}
\left.
 \begin{aligned}
\partial_{\bar{t}}\bar{b}  =\,  & \bar{D}_{b}\partial_{y}^2\bar{b} - \bar{b}
\\
\partial_{\bar{t}}\bar{k}  =\,  & \bar{D}_{k} \partial_{y}^2\bar{k}
\\
\partial_{\bar{t}}\bar{z}_{p}  =\,  & \bar{D}_z \partial_{y}^2 \bar{z}_{p}
\end{aligned}
 \, \right\}
 \text{ in  } \bar{\Omega}_{c}.
\end{IEEEeqnarray}

\begin{IEEEeqnarray}{rCl}
\left.
 \begin{aligned}
-\bar{D}_{b}\partial_{y}\bar{b}  =\,  & \bar{\beta}_{k}\bar{r}_{b}\bar{k} - \bar{\beta}_{b}\bar{b}\bar{r}_{k}
\\
-\bar{D}_{k}\partial_{y}\bar{k}  =\,  & \epsilon_{1}\left(\bar{\beta}_{b}\bar{r}_{k}\bar{b} - \bar{\beta}_{k}\bar{r}_{b}\bar{k}\right)
\\
\frac{d\bar{r}_{k}}{d \bar{t}}  =\,  & \epsilon_{2}\left(\bar{\beta}_{k}\bar{r}_{b}\bar{k} - \bar{\beta}_{b}\bar{b}\bar{r}_{k}\right)
\\
\frac{d\bar{r}_{b}}{d \bar{t}}  =\,  & \epsilon_{2}\left(\bar{\beta}_{b}\bar{b}\bar{r}_{k} - \bar{\beta}_{k}k^{*}t^{*}\bar{r}_{b}\bar{k}\right)
\\
-\bar{D}_{z}\partial_{y}\bar{z}_{p}  =\,  & 0
\end{aligned}
 \, \right\}\, 
 \text{ on  } \bar{\Gamma}_{c},
\end{IEEEeqnarray}

\begin{IEEEeqnarray}{rCl}
\left.
 \begin{aligned}
\bar{D}_{b}\partial_{y}\bar{b}  = \; & \frac{1}{1 + (\bar{\theta}_{b}\bar{z})^{h_{b}}}
\\
\bar{D}_{k}\partial_{y}\bar{k} =\;  & 0
\\
\frac{d\bar{z}}{d \bar{t}}  = \; & \bar{\delta}_{z}\bar{z}_{p}\bar{k} - \bar{\rho}_{z}\frac{\bar{z}}{1+(\bar{\theta}_{z}\bar{k})^{h_{z}}}
\\
\bar{D}_{z}\partial_{y}\bar{z}_{p}  =\;  & -\bar{\delta}_{z}\bar{z}_{p}\bar{k} + \bar{\rho}_{z}\frac{\bar{z}}{1+(\bar{\theta}_{z}\bar{k})^{h_{z}}}
\end{aligned}
 \, \right\}
 \text{on $\bar{\Gamma}_{n}$,}
\end{IEEEeqnarray}
where
\begin{center}
\begin{tabular}{llll}
$\bar{D}_{b} = \dfrac{D_{b}}{\mu_{b}l_{c}^{2}}$, & $\bar{D}_{k} = \dfrac{D_{k}}{\mu_{b}l_{c}^{2}}$, & $\bar{D}_{z} = \dfrac{D_{z}}{\mu_{b}l_{c}^{2}}$, & $\bar{\beta}_{k} = \dfrac{\beta_{k}R_{tot}K_{tot}}{\alpha_{b}}$,\\
$\bar{\beta}_{b} = \dfrac{\beta_{b}R_{tot}}{\mu_{b}}$, & $\epsilon_{1} = \dfrac{\alpha_{b}}{K_{tot}\mu_{b}}$, & $\epsilon_{2} = \dfrac{\alpha_{b}}{R_{tot}\mu_{b}}$, & $\bar{\theta}_{b} = \theta_{b}Z_{tot}$,\\
$\bar{\delta}_{z} = \dfrac{\delta_{z}K_{tot}}{\mu_{b}}$, & $\bar{\rho}_{z} = \dfrac{\rho_{z}}{\mu_{b}}$, & $\bar{\theta}_{z} = \theta_{z}K_{tot}$. &
\end{tabular}
\end{center}

\section{Rigorous Reduction of the Model for GA Signalling Pathway from \cite{Middleton_A_2012}}\label{app:GAred}
Here we present the rigourous reduction of the model for the GA signalling pathway from \cite{Middleton_A_2012} to the system of ODEs \eqref{eq:GAode1} and \eqref{eq:GAode2} for GA, GID1, DELLA, GA.GID1\textsuperscript{o}, GA.GID1\textsuperscript{c}, DELLA.GA.GID1\textsuperscript{c}, GID1\textsubscript{m} and DELLA\textsubscript{m}. To do this we assume that the dynamics of the actual biosynthesis may be captured from the DELLA dynamics only. We also assume only one functional form of the DELLA.GID1.GA complex. To save space we denote the independent variables by $x_{i}$ for $i=1,...,21$, in the same order as the full statement of the model in \cite{Middleton_A_2012}. \\
\\
To describe GA signal transduction:
\begin{IEEEeqnarray}{rCl}
\IEEEyesnumber\IEEEyessubnumber*
\frac{dx_{1}}{dt} & = & -l_{a}x_{1}x_{11} + l_{d}x_{2} + \delta_{gid1}x_{20} - \mu_{gid1}x_{1},\label{eq:x1full}
\\
\frac{dx_{2}}{dt} & = & l_{a}x_{1}x_{11} - l_{d}x_{2} + px_{3} - qx_{2},\label{eq:x2full}
\\
\frac{dx_{3}}{dt} & = & -px_{3} + qx_{2} - (u_{a1}+u_{a2})x_{6}x_{3} + (u_{d1}+u_{m})x_{4} + u_{d2}x_{5},\label{eq:x3full}
\\
\frac{dx_{4}}{dt} & = & u_{a1}x_{6}x_{3} - (u_{d1}+u_{m})x_{4},\label{eq:x4full}
\\
\frac{dx_{5}}{dt} & = & u_{a2}x_{6}x_{3} - u_{d2}x_{5},\label{eq:x5full}
\\
\frac{dx_{6}}{dt} & = & -(u_{a1}+u_{a2})x_{6}x_{3} + u_{d1}x_{4} + u_{d2}x_{5} + \delta_{della}x_{21}.\label{eq:x6full}
\end{IEEEeqnarray}
To describe GA biosynthesis:
\begin{IEEEeqnarray*}{rCl}
\IEEEyesnumber\IEEEyessubnumber*
\frac{dx_{7}}{dt} & = & \omega_{ga12} - k_{a12}x_{7}x_{16} + k_{d12}x_{12} - \mu_{ga}x_{7},\label{eq:x7full}
\\
\frac{dx_{8}}{dt} & = & -k_{a15}x_{8}x_{16} + k_{d15}x_{13} + k_{m12}x_{12} - \mu_{ga}x_{8},\label{eq:x8full}
\\
\frac{dx_{9}}{dt} & = & -k_{a24}x_{9}x_{16} + k_{d24}x_{14} + k_{m15}x_{13} - \mu_{g1}x_{9},\label{eq:x9full}
\\
\frac{dx_{10}}{dt} & = & -k_{a9}x_{10}x_{17} + k_{d9}x_{15} + k_{m24}x_{14} - \mu_{ga}x_{10},\label{eq:x10full}
\\
\frac{dx_{11}}{dt} & = & P_{mem}\frac{S_{root}}{R_{root}}(A_{1}\omega_{ga4}-B_{1}x_{11}) + k_{m9}x_{15} - l_{a}x_{1}x_{11} + l_{d}x_{2}\nonumber\\
&& -\> \mu_{ga}x_{11}.\label{eq:x11full}
\end{IEEEeqnarray*}
To describe complexes of GAs and enzymes:
\begin{IEEEeqnarray}{rCl}
\IEEEyesnumber\IEEEyessubnumber*
\frac{dx_{12}}{dt} & = & k_{a12}x_{7}x_{16} - (k_{d12}+k_{m12})x_{12},\label{eq:x12full}
\\
\frac{dx_{13}}{dt} & = & k_{a15}x_{8}x_{16} - (k_{d15}+k_{m15})x_{13},\label{eq:x13full}
\\
\frac{dx_{14}}{dt} & = & k_{a24}x_{9}x_{16} - (k_{d24}+k_{m24})x_{14},\label{eq:x14full}
\\
\frac{dx_{15}}{dt} & = & k_{a9}x_{10}x_{17} - (k_{d9}+k_{m9})x_{15}.\label{eq:x15full}
\end{IEEEeqnarray}
To describe the enzymes:
\begin{IEEEeqnarray}{rCl}
\IEEEyesnumber\IEEEyessubnumber*
\frac{dx_{16}}{dt} & = & -k_{a12}x_{7}x_{16} - k_{a15}x_{8}x_{16} - k_{a24}x_{9}x_{16} + (k_{d12}+k_{m12})x_{12}\nonumber\\
&& +\> (k_{d15}+k_{m15})x_{13} + (k_{d24}+k_{m24})x_{14} + \delta_{ga20ox}x_{18}\nonumber\\
&& -\> \mu_{ga20ox}x_{16},\label{eq:x16full}
\\
\frac{dx_{17}}{dt} & = & -k_{a9}x_{10}x_{17} + (k_{d9}+k_{m9})x_{15} + \delta_{ga3ox}x_{19} - \mu_{ga3ox}x_{17}.\label{eq:x17full}
\end{IEEEeqnarray}
To describe the mRNAs:
\begin{IEEEeqnarray}{rCl}
\IEEEyesnumber\IEEEyessubnumber*
\frac{dx_{18}}{dt} & = & \phi_{ga20ox}\left(\frac{x_{6}}{x_{6}+\theta_{ga20ox}} - x_{18}\right),\label{eq:x18full}
\\
\frac{dx_{19}}{dt} & = & \phi_{ga3ox}\left(\frac{x_{6}}{x_{6}+\theta_{ga3ox}} - x_{19}\right),\label{eq:x19full}
\\
\frac{dx_{20}}{dt} & = & \phi_{gid1}\left(\frac{x_{6}}{x_{6}+\theta_{gid1}} - x_{20}\right),\label{eq:x20full}
\\
\frac{dx_{21}}{dt} & = & \phi_{della}\left(\frac{\theta_{della}}{x_{6}+\theta_{della}} - x_{21}\right).\label{eq:x21full}
\end{IEEEeqnarray}

\noindent We set Eq (\ref{eq:x5full}) to zero since we assume that only one configuration of the DELLA.GID1\textsuperscript{c}.GA\textsubscript{4} can form. Since the conversion of GA\textsubscript{4} precursors occurs on a relatively fast scale we set the time-derivatives in (\ref{eq:x7full})-(\ref{eq:x10full}) and (\ref{eq:x12full})-(\ref{eq:x19full}) equal to zero. We immediately remove 3 variables by noting from Eqs. (\ref{eq:x5full}), (\ref{eq:x18full}) and (\ref{eq:x19full}) that
\begin{IEEEeqnarray}{rCl}
\IEEEyesnumber\IEEEyessubnumber*
x_{5} & = & \frac{u_{a2}}{u_{d2}}x_{6}x_{3},\label{eq:x5red}
\\
x_{18} & = & \frac{x_{6}}{x_{6}+\theta_{ga20ox}},\label{eq:x18red}
\\
x_{19} & = & \frac{x_{6}}{x_{6}+\theta_{ga3ox}},\label{eq:x19red}
\end{IEEEeqnarray}

\noindent and further from Eqs. (\ref{eq:x7full})-(\ref{eq:x10full}), and (\ref{eq:x12full})-(\ref{eq:x17full}) that

\begin{IEEEeqnarray}{rCl}
\IEEEyesnumber\IEEEyessubnumber*
x_{16} & = & \frac{\delta_{ga20ox}}{\mu_{ga20ox}}\frac{x_{6}}{x_{6}+\theta_{ga20ox}},\label{eq:x16red}
\\
x_{17} & = & \frac{\delta_{ga3ox}}{\mu_{ga3ox}}\frac{x_{6}}{x_{6}+\theta_{ga3ox}},\label{eq:x17red}
\\
k_{m9}x_{15} & = & \omega_{ga12} - \mu_{ga}(x_{7}+x_{8}+x_{9}+x_{10}),\label{eq:x15red}
\end{IEEEeqnarray}
 with
\begin{IEEEeqnarray}{rCl}
\IEEEyesnumber\IEEEyessubnumber*
x_{7} & = & \frac{\omega_{ga12}}{\mu_{ga}+K_{12}},\label{eq:x7red}
\\
x_{8} & = & \frac{\omega_{ga12}K_{12}}{\left(\mu_{ga}+K_{12}\right)\left(\mu_{ga}+K_{15}\right)},\label{eq:x8red}
\\
x_{9} & = & \frac{\omega_{ga12}K_{12}K_{15}}{\left(\mu_{ga}+K_{12}\right)\left(\mu_{ga}+K_{15}\right)\left(\mu_{ga}+K_{24}\right)},\label{eq:x9red}
\\
x_{10} & = & \frac{\omega_{ga12}K_{12}K_{15}K_{24}}{\left(\mu_{ga}+K_{12}\right)\left(\mu_{ga}+K_{15}\right)\left(\mu_{ga}+K_{24}\right)\left(\mu_{ga}+K_{9}\right)},\label{eq:x10red}
\end{IEEEeqnarray}
 where
\begin{IEEEeqnarray*}{rCl}
K_{12} & = &  \frac{\delta_{ga20ox}k_{a12}k_{m12}}{\mu_{ga20ox}(k_{d12}+k_{m12})}\frac{x_{6}}{x_{6}+\theta_{ga20ox}},
\\
K_{15} & = & \frac{\delta_{ga20ox}k_{a15}k_{m15}}{\mu_{ga20ox}(k_{d15}+k_{m15})}\frac{x_{6}}{x_{6}+\theta_{ga20ox}},
\\
K_{24} & = & \frac{\delta_{ga20ox}k_{a24}k_{m24}}{\mu_{ga20ox}(k_{d24}+k_{m24})}\frac{x_{6}}{x_{6}+\theta_{ga20ox}},
\\
K_{9} & = & \frac{\delta_{ga3ox}k_{a9}k_{m9}}{\mu_{ga3ox}(k_{d9}+k_{m9})}\frac{x_{6}}{x_{6}+\theta_{ga3ox}}.
\end{IEEEeqnarray*}
 Substituting these into Eq.(\ref{eq:x15red}) 
\begin{IEEEeqnarray*}{rCl}
k_{m9}x_{15} & = & \omega_{ga12}
\left\langle
 \begin{IEEEeqnarraybox}[
   \IEEEeqnarraystrutmode
  ][c]{rCl}
	1-\mu_{ga}
	\left[
	 \begin{IEEEeqnarraybox}[
	  \IEEEeqnarraystrutmode
	 ][c]{rCl}
	   \frac{
	    \left(
	     \begin{IEEEeqnarraybox}[
	      \IEEEeqnarraystrutmode
	     ][c]{rCl}
	      \left(\mu_{ga}+K_{15}\right)\left(\mu_{ga}+K_{24}\right)\left(\mu_{ga}+K_{9}\right)\\
	      +\> K_{12}\left(\mu_{ga}+K_{24}\right)\left(\mu_{ga}+K_{9}\right)\\
	      +\> K_{12}K_{15}\left(\mu_{ga}+K_{9}\right) + K_{12}K_{15}K_{24}
	     \end{IEEEeqnarraybox}
	      \,\right)
	     }{\left(\mu_{ga}+K_{12}\right)\left(\mu_{ga}+K_{15}\right)\left(\mu_{ga}+K_{24}\right)\left(\mu_{ga}+K_{9}\right)}
	   \end{IEEEeqnarraybox}
	    \, \right]
\end{IEEEeqnarraybox}
 \, \right\rangle
\\
& = & \omega_{ga12}\frac{K_{12}K_{15}K_{24}K_{9}}{\left(\mu_{ga}+K_{12}\right)\left(\mu_{ga}+K_{15}\right)\left(\mu_{ga}+K_{24}\right)\left(\mu_{ga}+K_{9}\right)},
\end{IEEEeqnarray*}
which can be multiplied through by the Hill function denominators, and upon simplifying to fourth order terms in $x_{6}$ becomes
\begin{equation*}
k_{m9}x_{15} = \alpha_{g}\frac{x_{6}^{4}}{x_{6}^{4}+\theta_{ga}},
\end{equation*}
where
\begin{IEEEeqnarray*}{rCl}
\alpha_{g} & = & \frac{\omega_{ga12}\delta_{ga3ox}(\delta_{ga20ox})^{3}k_{a12}k_{m12}k_{a15}k_{m15}k_{a24}k_{m24}k_{a9}k_{m9}}{
 \left(
  \begin{IEEEeqnarraybox}[
   \IEEEeqnarraystrutmode
  ][c]{rCl}
   (\mu_{ga}\mu_{ga20ox}(k_{d12}+k_{m12})+\delta_{ga20ox}k_{a12}k_{m12})\times\\
   \times(\mu_{ga}\mu_{ga20ox}(k_{d15}+k_{m15})+\delta_{ga20ox}k_{a15}k_{m15})\times\\
   \times(\mu_{ga}\mu_{ga20ox}(k_{d24}+k_{m24})+\delta_{ga20ox}k_{a24}k_{m24})\times\\
   \times(\mu_{ga}\mu_{ga3ox}(k_{d9}+k_{m9})+\delta_{ga3ox}k_{a9}k_{m9})
  \end{IEEEeqnarraybox}
  \, \right)
  }
  \end{IEEEeqnarray*}
  and 
  \begin{IEEEeqnarray*}{rCl}
\theta_{ga} & = & \frac{
 \left(
  \begin{IEEEeqnarraybox}[
   \IEEEeqnarraystrutmode
  ][c]{rCl}
  \mu_{ga}^{4}\theta_{ga3ox}(\theta_{ga20ox})^{3}\mu_{ga3ox}(\mu_{ga20ox})^{3}(k_{d12}+k_{m12})*\\
  *(k_{d15}+k_{m15})(k_{d24}+k_{m24})(k_{d9}+k_{m9})
  \end{IEEEeqnarraybox}
   \, \right)
}{
 \left(
  \begin{IEEEeqnarraybox}[
   \IEEEeqnarraystrutmode
  ][c]{rCl}
   (\mu_{ga}\mu_{ga20ox}(k_{d12}+k_{m12})+\delta_{ga20ox}k_{a12}k_{m12})\times\\
   \times(\mu_{ga}\mu_{ga20ox}(k_{d15}+k_{m15})+\delta_{ga20ox}k_{a15}k_{m15})\times\\
   \times(\mu_{ga}\mu_{ga20ox}(k_{d24}+k_{m24})+\delta_{ga20ox}k_{a24}k_{m24})\times\\
   \times(\mu_{ga}\mu_{ga3ox}(k_{d9}+k_{m9})+\delta_{ga3ox}k_{a9}k_{m9})
  \end{IEEEeqnarraybox}
  \, \right)
  }.
\end{IEEEeqnarray*}
Thus, we obtain a reduced model that reads
\begin{IEEEeqnarray}{rCl}
\IEEEyesnumber\IEEEyessubnumber*
\frac{dx_{1}}{dt} & = & -l_{a}x_{1}x_{11} + l_{d}x_{2} + \delta_{gid1}x_{20} - \mu_{gid1}x_{1},\label{eq:x1red}
\\
\frac{dx_{2}}{dt} & = & l_{a}x_{1}x_{11} - l_{d}x_{2} + px_{3} - qx_{2},\label{eq:x2red}
\\
\frac{dx_{3}}{dt} & = & -px_{3} + qx_{2} - (u_{a1}+u_{a2})x_{6}x_{3} + (u_{d1}+u_{m})x_{4},\label{eq:x3red}
\\
\frac{dx_{4}}{dt} & = & u_{a1}x_{6}x_{3} - (u_{d1}+u_{m})x_{4},\label{eq:x4red}
\\
\frac{dx_{6}}{dt} & = & -u_{a1}x_{6}x_{3} + u_{d1}x_{4} + \delta_{della}x_{21},\label{eq:x6red}
\\
\frac{dx_{11}}{dt} & = & P_{mem}\frac{S_{root}}{R_{root}}(A_{1}\omega_{ga4}-B_{1}x_{11}) + \alpha_{g}\frac{x_{6}^{4}}{x_{6}^{4}+\theta_{ga}} - l_{a}x_{1}x_{11}\nonumber\\
 && +\> l_{d}x_{2} - \mu_{ga}x_{11},\label{eq:x11red}
\\
\frac{dx_{20}}{dt} & = & \phi_{gid1}\left(\frac{x_{6}}{x_{6}+\theta_{gid1}} - x_{20}\right),\label{eq:x20red}
\\
\frac{dx_{21}}{dt} & = & \phi_{della}\left(\frac{\theta_{della}}{x_{6}+\theta_{della}} - x_{21}\right).\label{eq:x21red}
\end{IEEEeqnarray}

\section{Non-dimensionalisation of the model including spatial heterogeneity in the crosstalk signalling}\label{app:Crosstalk_Pde}
Again applying scaling and non-dimensionalising via $t = t^{*}\bar{t}$ etc, system (\ref{eq:Crosstalk_Pde1})-(\ref{eq:Crosstalk_Pde5}) transforms into
\begin{IEEEeqnarray*}{rCl}
\left.
 \begin{IEEEeqnarraybox}[
   \IEEEeqnarraystrutmode
  ][c]{rCl}
  \partial_{\bar{t}}\bar{b} & = & \frac{D_{b}t^{*}}{(x^{*})^{2}} \partial_{y}^2\bar{b} - \mu_{b}t{*}\bar{b}\\
  \partial_{\bar{t}}\bar{k} & = & \frac{D_{k}t^{*}}{(x^{*})^{2}} \partial_{y}^2\bar{k}\\
  \partial_{\bar{t}}\bar{z}_{p} & = & \frac{D_{z}t^{*}}{(x^{*})^{2}} \partial_{y}^2\bar{z}_{p}\\
  \partial_{\bar{t}}\bar{g} & = & \frac{D_{g}t^{*}}{(x^{*})^{2}} \partial_{y}^2\bar{g} - \mu_{g}t^{*}\bar{g}
 \end{IEEEeqnarraybox}
\, \right\}
\text{in $\bar{\Omega}_{c}$,}
\end{IEEEeqnarray*}
and receptor based interactions of BR, BKI1 and BRI1 and influx of exogenous GA occur on the plasma membrane
\begin{IEEEeqnarray*}{rCl}
\left.
 \begin{IEEEeqnarraybox}[
   \IEEEeqnarraystrutmode
  ][c]{rCl}
  -\frac{D_{b}t^{*}}{(x^{*})^{2}}\partial_{y}\bar{b} & = & \frac{\beta_{k}t^{*}r_{b}^{*}k^{*}}{x^{*}b^{*}}\bar{r}_{b}\bar{k} - \frac{\beta_{b}t^{*}r_{k}^{*}}{x^{*}}\bar{r}_{k}\bar{b}\\
  -\frac{D_{k}t^{*}}{(x^{*})^{2}}\partial_{y}\bar{k} & = & \frac{\beta_{b}t^{*}r_{k}^{*}b^{*}}{x^{*}k^{*}}\bar{r}_{k}\bar{b} - \frac{\beta_{k}t^{*}r_{b}^{*}}{x^{*}}\bar{r}_{b}\bar{k}\\
  -\frac{D_{z}t^{*}}{(x^{*})^{2}}\partial_{y}\bar{z}_{p} & = & 0\\
  -\frac{D_{g}t^{*}}{(x^{*})^{2}}\partial_{y}\bar{g} & = & 0\\
  \frac{d\bar{r}_{k}}{d\bar{t}} & = & \frac{\beta_{k}r_{b}^{*}k^{*}t^{*}}{r_{k}^{*}}\bar{r}_{b}\bar{k} - \beta_{b}b^{*}t^{*}\bar{r}_{k}\bar{b}\\
  \frac{d\bar{r}_{b}}{d\bar{t}} & = & \frac{\beta_{b}r_{k}^{*}b^{*}t^{*}}{r_{b}^{*}}\bar{r}_{k}\bar{b} - \beta_{k}k^{*}t^{*}\bar{r}_{b}\bar{k}
 \end{IEEEeqnarraybox}
\, \right\}
\text{on $\bar{\Gamma}_{c}$.}
\end{IEEEeqnarray*}
 We assume that production of BR, change in phosphorylation status of BZR and interactions between BZR  and DELLA occur in the nucleus
\begin{IEEEeqnarray*}{rCl}
\left.
 \begin{IEEEeqnarraybox}[
   \IEEEeqnarraystrutmode
  ][c]{rCl}
  &\frac{D_{b}t^{*}}{(x^{*})^{2}} &  \partial_{y}\bar{b} = \frac{l_{c}\alpha_{b}t^{*}}{1+\left(\frac{\theta_{b}z^{*}\bar{z}}{l_{c}}\right)^{h_{b}}}\\
  & \frac{D_{k}t^{*}}{(x^{*})^{2}} & \partial_{y}\bar{k}  = 0\\
  &\frac{d\bar{z}}{d\bar{t}} & = \frac{l_{c}\delta_{z}t^{*}z_{p}^{*}k^{*}}{z^{*}}\bar{z}_{p}\bar{k} - \rho_{z}t^{*}\frac{\bar{z}}{1+(\theta_{z}k^{*}\bar{k})^{h_{z}}} - \frac{\beta_{z}t^{*}}{l_{c}}\bar{z}\bar{d}_{l} + \frac{\gamma_{z}z_{d}^{*}t^{*}}{z^{*}}\bar{z}_{d}\\
  & \frac{D_{z}t^{*}}{(x^{*})^{2}} &  \partial_{y}\bar{z}_{p}  = -\frac{l_{c}\delta_{z}t^{*}k^{*}}{x^{*}}\bar{z}_{p}\bar{k} + \frac{\rho_{z}z^{*}t^{*}}{x^{*}z_{p}^{*}}\frac{\bar{z}}{1+(\theta_{z}k^{*}\bar{k})^{h_{z}}}
  \end{IEEEeqnarraybox}
\, \right\}
\text{on $\bar{\Gamma}_{n}$.}
\end{IEEEeqnarray*}
We assume all constituent processes of the GA signalling pathway apart from degradation and influx of GA also occur in the nucleus
\begin{IEEEeqnarray*}{rCl}
\left.
 \begin{aligned}
  \frac{d\bar{r}}{d\bar{t}}  =\,  & -\beta_{g}t^{*}g^{*}\bar{r}\bar{g} + \frac{\gamma_{g}r_{g}^{o*}t^{*}}{r^{*}}\bar{r}_{g}^{o} + \frac{l_{c}\alpha_{r}t^{*}}{r^{*}}\bar{r}_{m} - \mu_{r}t^{*}\bar{r}\\
  \frac{d\bar{r}_{g}^{o}}{d\bar{t}}  =\,  & \frac{\beta_{g}r^{*}g^{*}}{r_{g}^{o*}}\bar{r}\bar{g} - (\gamma_{g}+\lambda^{c})t^{*}\bar{r}_{g}^{o} + \frac{\lambda^{o}r_{g}^{c*}t^{*}}{r_{g}^{o*}}\bar{r}_{g}^{c}\\
  \frac{d\bar{r}_{g}^{c}}{d\bar{t}} = \, & \frac{\lambda^{c}r_{g}^{o*}t^{*}}{r_{g}^{c*}}\bar{r}_{g}^{o} - \lambda^{o}t^{*}\bar{r}_{g}^{c} - \frac{\beta_{d}d_{l}^{*}t^{*}}{l_{c}}\bar{d}_{l}\bar{r}_{g}^{c} + \frac{(\gamma_{d}+\mu_{d})r_{d}^{*}t^{*}}{r_{g}^{c*}}\bar{r}_{d}\\
  \frac{d\bar{r}_{d}}{d\bar{t}}  = \, & \frac{\beta_{d}d_{l}^{*}r_{g}^{c*}t^{*}}{l_{c}r_{d}^{*}}\bar{d}_{l}\bar{r}_{g}^{c} - (\gamma_{d}+\mu_{d})t^{*}\bar{r}_{d}\\
  \frac{d\bar{d}_{l}}{d\bar{t}}  = & -\frac{\beta_{d}r_{g}^{c*}t^{*}}{l_{c}}\bar{d}_{l}\bar{r}_{g}^{c} + \frac{\gamma_{d}r_{d}^{*}t^{*}}{d_{l}^{*}}\bar{r}_{d} + \frac{l_{c}\alpha_{d}t^{*}}{d_{l}^{*}}\bar{d}_{m} - \beta_{z}z^{*}t^{*}\bar{z}\bar{d}_{l} + \frac{\gamma_{z}z_{d}^{*}}{d_{l}^{*}}\bar{z}_{d}\\
   \frac{D_{g}t^{*}}{\left(x^{*}\right)^{2}}\partial_{y}\bar{g}  = \, & \frac{l_{c}\alpha_{g}t^{*}}{x^{*}g^{*}}\frac{\left(\bar{d}_{l}+\frac{\phi_{z}z^{*}}{d_{l}^{*}}\bar{z}\right)^{4}}{\frac{l_{c}^{4}\vartheta_{g}}{\left(d_{l}^{*}\right)^{4}} + \left(\bar{d}_{l}+\frac{\phi_{z}z^{*}}{d_{l}^{*}}\bar{z}\right)^{4}} - \frac{\beta_{g}r^{*}t^{*}}{x^{*}}\bar{r}\bar{g} + \frac{\gamma_{g}r_{g}^{o*}}{x^{*}g^{*}}\bar{r}_{g}^{o}\\
  \frac{d\bar{r}_{m}}{d\bar{t}}  =\,  & \phi_{r}t^{*}\left(\frac{\bar{d}_{l}}{\frac{l_{c}\theta_{r}}{d_{l}^{*}}+\bar{d}_{l}}-\bar{r}_{m}\right)\\
  \frac{d\bar{d}_{m}}{d\bar{t}}  = \, & \phi_{d}t^{*}\left(\frac{\frac{\theta_{d}}{d_{l}^{*}}}{\frac{\theta_{d}}{d_{l}^{*}}+\bar{d}_{l}}-\bar{d}_{m}\right)
  \end{aligned}
 \right\} 
\text{ on } \Gamma_{n}.
\end{IEEEeqnarray*}
 Finally we assume cross-talk via formation of the BZR.DELLA complex also occurs in the nucleus
\begin{IEEEeqnarray*}{rCl}
  \frac{d\bar{z}_{d}}{d\bar{t}} & = & \frac{\beta_{z}z^{*}d_{l}^{*}}{z_{d}^{*}}\bar{z}\bar{d}_{l} - \gamma_{z}t^{*}\bar{z}_{d} \quad 
\text{ on }  \Gamma_{n}.
\end{IEEEeqnarray*}
We set
\begin{center}
\begin{tabular}{lllll}
$t^{*} = \dfrac{1}{\gamma_{z}}$, & $x^{*} = l_{c}$, & $b^{*} = \dfrac{\alpha_{b}}{\mu_{b}}$, & $k^{*} = K_{tot}$, & $r_{k}^{*} = l_{c}R_{tot}$,\\
$r_{b}^{*} = l_{c}R_{tot}$, & $z^{*} = l_{c}Z_{tot}$, & $z_{p}^{*} = Z_{tot}$, & $r^{*} = \dfrac{l_{c}\alpha_{r}}{\mu_{r}}$, & $r_{g}^{o*} = \dfrac{l_{c}\alpha_{r}}{\mu_{r}}$,\\
$r_{g}^{c*} = \dfrac{l_{c}\alpha_{r}}{\mu_{r}}$, & $r_{d}^{*} = \dfrac{l_{c}\alpha_{r}}{\mu_{r}}$, & $d_{l}^{*} = \dfrac{l_{c}\alpha_{d}}{\mu_{d}}$, & $g^{*} = \dfrac{\alpha_{g}}{\mu_{g}}$, & $z_{d}^{*} = l_{c}Z_{tot}$,
\end{tabular}
\end{center}
and introduce the new dimensionless parameters
\begin{center}
\begin{tabular}{llll}
$\bar{D}_{b} = \dfrac{D_{b}}{\gamma_{c}l_{c}^{2}}$, & $\bar{D}_{k} = \dfrac{D_{k}}{\gamma_{Z}l_{c}^{2}}$, & $\bar{D}_{z} = \dfrac{D_{z}}{\gamma_{z}l_{c}^{2}}$, & $\bar{D}_{g} = \dfrac{D_{g}}{\gamma_{z}l_{c}^{2}}$,\\
$\bar{\mu}_{b} = \dfrac{\mu_{b}}{\gamma_{z}}$, & $\bar{\mu}_{g} = \dfrac{\mu_{g}}{\gamma_{z}}$, & $\bar{\beta}_{k} = \dfrac{\beta_{k}R_{tot}K_{tot}\mu_{b}}{\gamma_{z}\alpha_{b}}$, & $\bar{\beta}_{b} = \dfrac{\beta_{b}R_{tot}}{\gamma_{z}}$,\\
$\epsilon_{1} = \dfrac{\alpha_{b}}{\mu_{b}K_{tot}}$, & $\epsilon_{2} = \dfrac{\alpha_{b}}{\mu_{b}R_{tot}}$, & $\bar{\theta}_{b} = \theta_{b}Z_{tot}$, & $\bar{\delta}_{z} = \dfrac{\delta_{z}K_{tot}}{\gamma_{z}}$, \\ $\bar{\rho}_{z} = \dfrac{\rho_{z}}{\gamma_{z}}$, & $\bar{\theta}_{z} = \theta_{z}K_{tot}$, & $\bar{\beta}_{z} = \dfrac{\beta_{z}\alpha_{d}}{\gamma_{z}\mu_{d}}$, & $\bar{\beta}_{g} = \dfrac{\beta_{g}\alpha_{g}}{\gamma_{z}\mu_{g}}$, \\ $\bar{\gamma}_{g} = \dfrac{\gamma_{g}}{\gamma_{z}}$, & $\bar{\mu}_{r} = \dfrac{\mu_{r}}{\gamma_{z}}$, & $\bar{\lambda}^{c} = \dfrac{\lambda^{c}}{\gamma_{z}}$, & $\bar{\lambda}^{o} = \dfrac{\lambda^{o}}{\gamma_{z}}$, \\ $\bar{\beta}_{d} = \dfrac{\beta_{d}\alpha_{d}}{\gamma_{z}\mu_{d}}$, & $\bar{\gamma}_{d} = \dfrac{\gamma_{d}}{\gamma_{z}}$, & $\bar{\mu}_{d} = \dfrac{\mu_{d}}{\gamma_{z}}$, & $\epsilon_{3} = \dfrac{\alpha_{r}\mu{d}}{\mu_{r}\alpha_{d}}$, \\ $\epsilon_{4} = \dfrac{Z_{tot}\mu_{d}}{\alpha_{d}}$, & $\bar{\vartheta}_{g} = \dfrac{\vartheta_{g}\mu_{d}^{4}}{\alpha_{d}^{4}}$, & $\epsilon_{5} = \dfrac{\alpha_{r}\mu_{g}}{\mu_{r}\alpha_{g}}$, & $\bar{\phi_{r}} = \dfrac{\phi_{r}}{\gamma_{z}}$, \\
$\bar{\vartheta}_{r} = \dfrac{\vartheta_{r}\mu_{d}}{\alpha_{d}}$, & $\bar{\phi}_{d} = \dfrac{\phi_{d}}{\gamma_{z}}$, & $\bar{\vartheta}_{d} = \dfrac{\vartheta_{d}\mu_{d}}{\alpha_{d}}$. &  
\end{tabular}
\end{center}
Neglecting $\bar{\phantom{bar}}$s we write out the non-dimensionalised PDE-ODE system, with diffusion and hormone degradation occurring in the cytoplasm
\begin{IEEEeqnarray}{rCl}
\left.
 \begin{IEEEeqnarraybox}[
   \IEEEeqnarraystrutmode
  ][c]{rCl}
  \frac{\partial b}{\partial t} & = & D_{b} \partial_{y}^2 b - \mu_{b}b\\
  \frac{\partial k}{\partial t} & = & D_{k} \partial_{y}^2 k\\
  \frac{\partial z_{p}}{\partial t} & = & D_{z} \partial_{y}^2 z_{p}\\
  \frac{\partial g}{\partial t} & = & D_{g} \partial_{y}^2 g - \mu_{g}g
 \end{IEEEeqnarraybox}
\, \right\}
\text{in $\Omega_{c}$.}
\end{IEEEeqnarray}
Receptor binding and dissociation occurring on the plasma membrane
\begin{IEEEeqnarray}{rCl}
\left.
 \begin{IEEEeqnarraybox}[
   \IEEEeqnarraystrutmode
  ][c]{rCl}
  -D_{b}\partial_{y}b & = & \beta_{k}r_{b}k - \beta_{b}r_{k}b\\
  -D_{k}\partial_{y}k & = & \epsilon_{1}\left(\beta_{b}r_{k}b - \beta_{k}r_{b}k\right)\\
  -D_{z}\partial_{y}z_{p} & = & 0\\
  -D_{g}\partial_{y}g  & = & 0\\
  \frac{dr_{k}}{dt} & = & \epsilon_{2}\left(\beta_{k}r_{b}k - \beta_{b}r_{k}b\right)\\
  \frac{dr_{b}}{dt} & = & \epsilon_{2}\left(\beta_{b}r_{k}b - \beta_{k}r_{b}k\right)
 \end{IEEEeqnarraybox}
\, \right\}
\text{on $\Gamma_{c}$.}
\end{IEEEeqnarray}
The BR signalling processes occurring in the nucleus
\begin{IEEEeqnarray}{rCl}
\left.
 \begin{IEEEeqnarraybox}[
   \IEEEeqnarraystrutmode
  ][c]{rCl}
   D_{b}\partial_{y}b & = & \frac{\mu_{b}}{1+(\theta_{b}z)^{h_{b}}}\\
  D_{k}\partial_{y}k & = & 0\\
  \frac{dz}{dt} & = & \delta_{z}z_{p}k - \rho_{z}\frac{z}{1+(\theta_{z}k)^{h_{z}}} - \beta_{z}zd_{l} + z_{d}\\
  D_{z}\partial_{y}z_{p} & = & -\delta_{z}z_{p}k + \rho_{z}\frac{z}{1+(\theta_{z}k)^{h_{z}}}
   \end{IEEEeqnarraybox}
\, \right\}
\text{on $\Gamma_{n}$.}
\end{IEEEeqnarray}
The GA signalling processes occurring in the nucleus
\begin{IEEEeqnarray}{rCl}
\left.
 \begin{IEEEeqnarraybox}[
   \IEEEeqnarraystrutmode
  ][c]{rCl}
  \frac{dr}{dt} & = & -\beta_{g}rg + \gamma_{g}r_{g}^{o} + \mu_{r}(r_{m} - r)\\
  \frac{dr_{g}^{o}}{dt} & = & \beta_{g}rg - (\gamma_{g}+\lambda^{c})r_{g}^{o} + \lambda^{o}r_{g}^{c}\\
  \frac{dr_{g}^{c}}{dt} & = & \lambda^{c}r_{g}^{o} - \lambda^{o}r_{g}^{c} - \beta_{d}d_{l}r_{g}^{c} + (\gamma_{d}+\mu_{d})r_{d}\\
  \frac{dr_{d}}{dt} & = & \beta_{d}d_{l}r_{g}^{c} - (\gamma_{d}+\mu_{d})r_{d}\\
  \frac{dd_{l}}{dt} & = & \epsilon_{3}\left(-\beta_{d}d_{l}r_{g}^{c} + \gamma_{d}r_{d}\right) + \mu_{d}d_{m} - \epsilon_{4}\left(\beta_{z}zd_{l} + z_{d}\right)\\
  D_{g}\partial_{y}g & = & \mu_{g}\frac{\left(d_{l} + \epsilon_{4}\phi_{z}z\right)^{4}}{\vartheta_{g} + \left(d_{l} + \epsilon_{4}\phi_{z}z\right)^{4}} - \epsilon_{5}\left(\beta_{g}rg + \gamma_{g}r_{g}^{o}\right)\\
  \frac{dr_{m}}{dt} & = & \phi_{r}\left(\frac{d_{l}}{\vartheta_{d} + d_{l}} - r_{m}\right)\\
  \frac{dd_{l}}{dt} & = & \phi_{d}\left(\frac{\vartheta_{d}}{\vartheta_{d} + d_{l}} - d_{m}\right)
 \end{IEEEeqnarraybox}
\, \right\}
\text{on $\Gamma_{n}$.}
\end{IEEEeqnarray}
The crosstalk interactions occurring in the nucleus
\begin{IEEEeqnarray}{rCl}
  \frac{dz_{d}}{dt} & = & \beta_{z}zd_{l} - z_{d}\qquad
\text{on $\Gamma_{n}$.}
\end{IEEEeqnarray}

\end{document}